\documentclass[acmsmall,screen]{acmart}
\usepackage{xspace}
\usepackage[altpo, epsilon]{backnaur}
\usepackage{formal-grammar}
\usepackage{amsfonts,amsmath}
\usepackage{mathpartir}
\usepackage{stmaryrd}
\usepackage{multicol}
\usepackage{enumitem}
\usepackage{caption}
\usepackage{rotating}
\usepackage{subcaption}
\usepackage{wrapfig}
\usepackage{placeins}
\usepackage{xcolor}
\usepackage{listings}
\usepackage [listings, skins, breakable]{tcolorbox}
\tcbuselibrary{listings, breakable}
\usepackage[scale=0.79] {sourcecodepro}
\usepackage{setspace}
\usepackage{tabularx}

\newif\iffullversion
\fullversiontrue

\renewcommand*\copyright{{%
  \textcopyright}}

\AtBeginDocument{%
  \providecommand\BibTeX{{%
    \normalfont B\kern-0.5em{\scshape i\kern-0.25em b}\kern-0.8em\TeX}}}

\newcommand{\systemname}{\textsc{Syren}\xspace}

\definecolor{lightestgray}{RGB}{240,240,240}
\definecolor{lightgray}{RGB}{215,215,215}

\def\numTotalBenchmarks{54}

\definecolor{darkgray}{rgb}{.4,.4,.4}
\definecolor{purple}{rgb}{.3,0,.5}
\definecolor{magenta}{rgb}{.7, 0, .1}
\definecolor{darkblue}{rgb}{0, .2, .4}
\definecolor{lightblue}{rgb}{.6,.8,1}
\colorlet{lblue}{lightblue!35}
\definecolor{orange}{rgb}{1, .6, 0}
\colorlet{cameraready}{.}

\lstdefinelanguage{Syren}{
  classoffset=0, 
  keywords={let,if,else,list,true,false,retry,until,for,return,::},
  classoffset=1, 
  morekeywords={where, \$, br},
  classoffset=2, 
  morekeywords={ec2},
  classoffset=3, 
  morekeywords={StopInstances, DescribeInstanceStatus, A, B, C, D, E, F, G},
  classoffset=4, 
  morekeywords={extractInstanceStatus, f, h},
  classoffset=0, 
  sensitive=true,
  comment=[l]{\#},
  morecomment=[s]{/*}{*/},
  morecomment=[s]{(**}{*)},
  morestring=[b]",
  alsoletter=:,
}

\NewTCBListing{syrenColorListing}{ !O{} }{%
listing only,
colback=orange!2,
colframe=orange!8,
hbox,
leftrule=2ex,
left=.15ex,
top=-.6ex,
bottom=-.8ex,
code={\setstretch{1.15}},
listing options={%
    escapeinside={||},
    language=Syren,
    numbers=left,
    numbersep=2ex,
    basicstyle=\small\linespread{.9}\ttfamily,
  },
#1}

\NewTCBListing{syrenBreakableColorListing}{ !O{} }{%
breakable,
listing only,
colback=orange!2,
colframe=orange!8,
leftrule=2ex,
left=.6ex,
top=-1ex,
bottom=-1ex,
code={\setstretch{1.15}},
listing options={%
    escapeinside={||},
    language=Syren,
    numbers=left,
    numbersep=2ex,
    basicstyle=\small\linespread{.9}\ttfamily,
  },
#1}

\newcommand{\simplesyreninline}[1] { \lstinline[language=Syren,escapeinside=||]^#1^}

\NewTCBListing{syrenListing}{ !O{} }{%
listing only,
hbox,
colback=white,
colframe=white,
left=0pt,
top=-1ex,
bottom=-1ex,
code={\setstretch{1.15}},
listing options={%
    escapeinside={||},
    language=Syren,
    numbers=none,
    basicstyle=\small\linespread{.9}\ttfamily,
  },
#1}

\newcommand{\dsl}[1]{\text{\simplesyreninline{#1}}}

\newcommand{\dummybranch}{{\texttt{\textcolor{magenta}{br}}}\xspace}
\newcommand{\dsldummybranch}{\texttt{\textcolor{magenta}{br}}\xspace}

\newcommand{\dsllambda}{\text{\textcolor{magenta}{\(\lambda\)}}\xspace}
\newcommand{\dsllet}{\texttt{\textcolor{purple}{let}}\xspace}
\newcommand{\dslin}{}

\newcommand{\dslend}{\texttt{\textcolor{purple}{return}}\xspace}
\newcommand{\dslif}{\texttt{\textcolor{purple}{if}}\xspace}
\newcommand{\dslelse}{\texttt{\textcolor{purple}{else}}\xspace}
\newcommand{\dslwhere}{\dsl{where}\xspace}
\newcommand{\dslretry}{\texttt{\textcolor{purple}{retry}}\xspace}
\newcommand{\dsluntil}{\texttt{\textcolor{purple}{until}}\xspace}
\newcommand{\dslretryuntil}[2]{{\dslretry}\ \{ {#1} \}\ \dsluntil\ \{ {#2} \}}
\newcommand{\dslfor}{\texttt{\textcolor{purple}{for}}\xspace}
\newcommand{\dslite}[3]{{\dslif}\ {#1}\ \{ {#2} \}\ \dslelse\ \{ {#3} \}}

\newcommand{\dsltern}[3]{{#1} ?\: {#2}\: : \: {#3}}
\newcommand{\dslforeach}[3]{\dslfor\ #1 \in #2\ \{ #3 \}}
\def\emptytrace{\langle\rangle}

\def\eval{\ensuremath{\Rightarrow}\xspace}
\def\returns{\textcolor{purple}{\lightning}}
\def\continue{\texttt{cont}}
\newcommand{\valuation}[3]{\llbracket{#1}\rrbracket_{{#2},{#3}}}
\def\scoreFuncOne{\ensuremath{\chi_{\mathrm{syn}}}} 
\def\scoreFuncTwo{\ensuremath{\chi_{\mathrm{T}}}}  

\makeatletter
\renewcommand{\sectionautorefname}{\S\@gobble}
\renewcommand{\subsectionautorefname}{\S\@gobble}
\makeatother

\newtheorem{lemma}{Lemma}

\setcopyright{cc}
\setcctype{by-sa}
\acmDOI{10.1145/3729316}
\acmYear{2025}
\acmJournal{PACMPL}
\acmVolume{9}
\acmNumber{PLDI}
\iffullversion
\else
\acmArticle{213}
\fi
\acmMonth{6}
\received{2024-11-15}
\received[accepted]{2025-03-06}

\acmBadgeR[https://www.acm.org/publications/policies/artifact-review-and-badging-current]{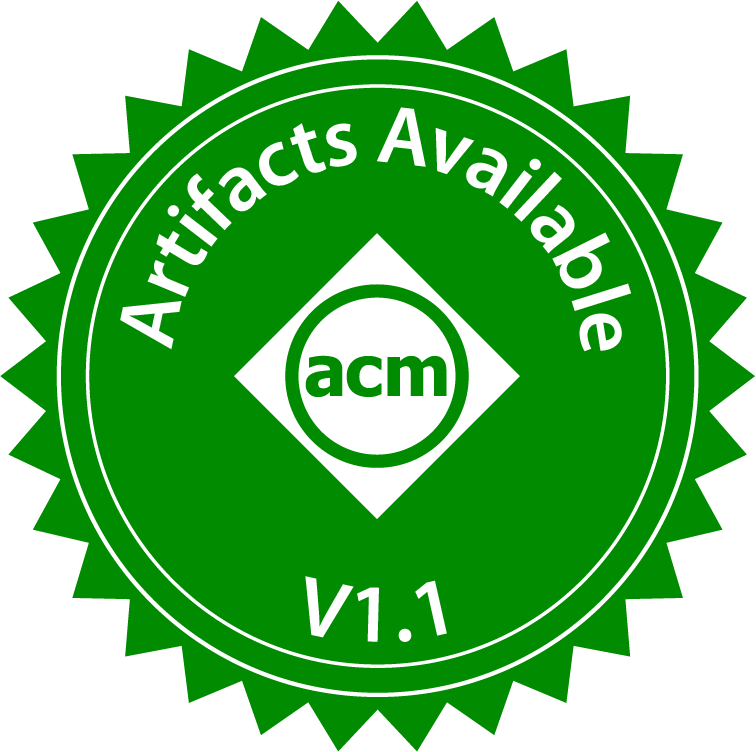}
\acmBadgeR[https://www.acm.org/publications/policies/artifact-review-and-badging-current]{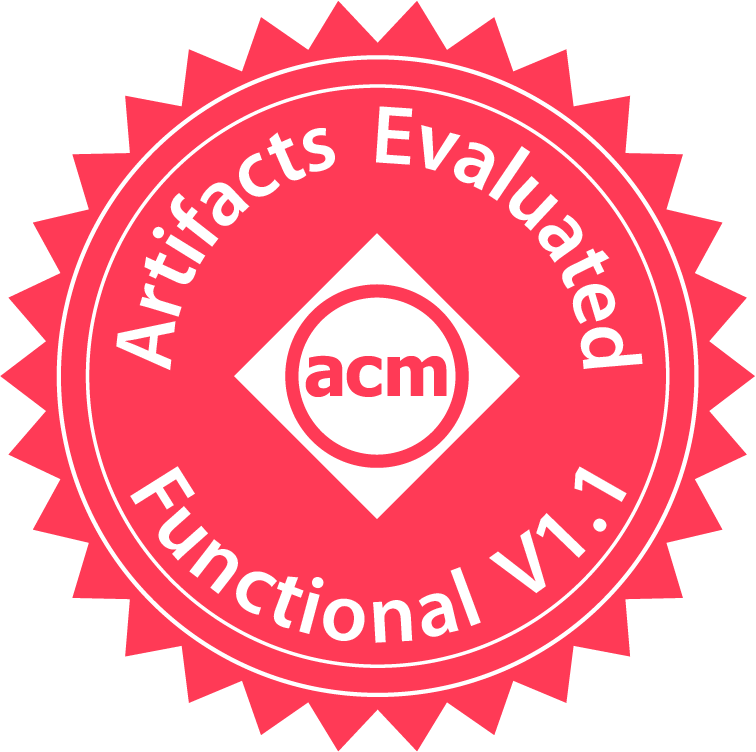}
\acmBadgeR[https://www.acm.org/publications/policies/artifact-review-and-badging-current]{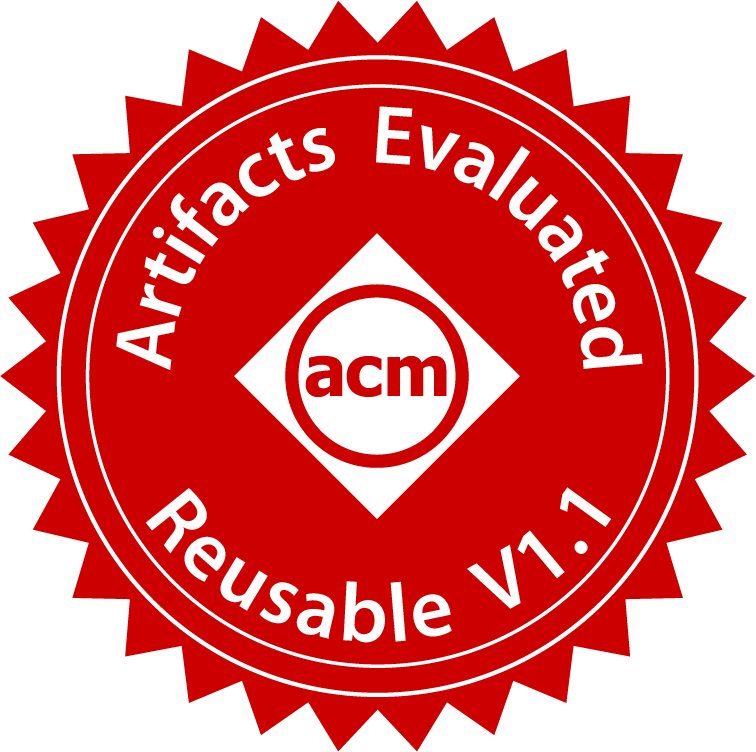}

\begin{document}

\title{Program Synthesis from Partial Traces}


\author{Margarida Ferreira}
\email{margarida@cmu.edu}
\orcid{0000-0002-1170-5124}
\affiliation{%
  \institution{Carnegie Mellon University}
  \city{Pittsburgh}
  \country{USA}
}
\affiliation{%
  \institution{INESC-ID/IST}
  \city{Lisbon}
  \country{Portugal}
}

\author{Victor Nicolet}
\email{victornl@amazon.com}
\orcid{0000-0002-3743-7498}

\author{Joey Dodds}
\email{jldodds@amazon.com}
\orcid{0009-0004-1534-6968}

\author{Daniel Kroening}
\email{dkr@amazon.com}
\orcid{0000-0002-6681-5283}

\affiliation{%
  \institution{Amazon}
  \city{Seattle}
  \state{Washington}
  \country{USA}
}

\begin{abstract}
We present the first technique to synthesize programs that compose side-effecting functions, pure functions, and control flow, from partial traces containing records of only the side-effecting functions. This technique can be applied to synthesize API composing scripts from logs of calls made to those APIs, or a script from traces of system calls made by a workload, for example. All of the provided traces are positive examples, meaning that they describe desired behavior. Our approach does not require negative examples. Instead, it generalizes over the examples and uses cost metrics to prevent over-generalization. Because the problem is too complex for traditional monolithic program synthesis techniques, we propose a new combination of optimizing rewrites and syntax-guided program synthesis. The resulting program is correct by construction, so its output will always be able to reproduce the input traces.
We evaluate the quality of the programs synthesized when considering various optimization metrics and the synthesizer's efficiency on real-world benchmarks. The results show that our approach can generate useful real-world programs.
\end{abstract}



\maketitle



\section{Introduction}
Program synthesis is for lazy people. The promise of program synthesis is that a user writes a simple specification and gets a complex program. If specifications are complex to write or synthesis tools are hard to apply, users will prefer to write the program directly. The best specification is one they already have with no changes needed. Unfortunately (for synthesis tool writers), these specs often take the form of traces of external side effects. These traces may be sequences of messages exchanged between networked servers, logs of calls made to an Application Programming Interface (API), or traces of system calls made by a workload. Program synthesis from traces can be widely applied across different tasks; there are many instances of this approach in Programming-By-Demonstration (PBD) work \cite{DBLP:conf/nips/ShinPS18,concise-models-from-long-traces,WatchWhatIDo,psnow}, but all of this work is either limited to pure functional examples, simple trace replay, or assumes no computation occurs between effects visible in~traces.

Real-world traces are an incomplete view of a task: function calls with no side effects may not be recorded.
A program \dsl{a=F(); b=h(a); return G(b)} where only the calls to \dsl{F}~and~\dsl{G} are logged may produce traces \dsl{F()=0 :: G("id-0")=true}\xspace and \dsl{F()=42 :: G("id-42")=false}, with no occurrence of~\dsl{h}.
These \emph{hidden} function calls not present in the traces are a significant challenge for trace-guided synthesis.
In our example, the synthesizer infers from the data that there is a non-visible function that transforms \dsl{0} into \dsl{"id-0"} and \dsl{42} into \dsl{"id-42"}.
These traces come from external recordings of program behavior. For example, the administrator of a cloud deployment might repeatedly follow a sequence of steps to create a data store and then attach an access policy to that data store. For convenience, they create a policy name by appending the string "policy" to an ID generated by the creation of the data store. The log of these actions kept on the cloud will contain the creation of the data store, the policy, and the connection between the two. It will not contain the hidden function where the administrator bases the role name on the store's~metadata. 

To help automate such repetitive tasks, we propose a new synthesis technique, \systemname. \systemname is the first approach to synthesize programs from \textit{partial} traces. It can infer both control flow and non-trivial hidden pure function calls with no additional input from the user. We combine optimizing program rewrites of an initial trivial solution with calls to a syntax-guided synthesizer (SyGuS) with input-output examples as specifications.
Synthesis from input-output examples, also known as Programming-By-Example (PBE), has a long history of research and enables the synthesis of nontrivial functions. 
In \systemname, we rely on a PBE synthesizer to synthesize hidden functions without additional input from the user. These hidden functions, in turn, let us perform further computation so that the rewrites can expose more intricate relations between~data.

In example- and trace-based synthesis, we can assume the existence of a hidden \textit{target program}, the program with the exact desired behavior. \textcolor{cameraready}{In the cloud administrator example above, the complete target program exists only in the administrator's mind.}
We can describe the \textit{behavior} of a program as the set of traces that result from its execution on any possible input. When the synthesis specification is a set of traces, there is a trivial solution to the synthesis problem---the program that exactly reproduces the input traces. This trivial solution, though correct by construction (in the sense that it satisfies the specification), is likely not the target program that the user desires. The user probably wants a program that \textit{generalizes} the provided traces. For example, our cloud administrator does not want a program that can reproduce all of the data store names they have used in the past, they want a program that takes a data store name as a parameter, allowing them to provide important information like the name, while saving them many repetitive clicks on a web interface. 
The trivially correct program is a lower bound on the behavior of all possible correct programs because it produces the minimal set of traces to be considered correct. The target program, on the other hand, is an upper bound of the desired behaviors---it would provide us with all the possible traces that are considered correct, but behaviors not exhibited by the target are undesirable. Since we do not have access to the target program, we need another way to quantify as accurately as possible how close to the target behavior a program in the space is. In other words, we need a \textit{cost function} to efficiently traverse the program space. 

Besides generalizing to traces beyond those observed (allowing more behavior), we considered readability when building \systemname's cost functions. Like other synthesis works before us, we generally follow Occam's razor principle and favor shorter programs to achieve both goals. We evaluate \systemname using two different cost functions we built, but our approach is agnostic to the cost function; a user can write different cost functions for \systemname if they choose to optimize for something else. For example, a user might specify that a specific input to a given method call must be generalized, that they want to minimize the syntactic statements in the program, or a combination of both. 

\begin{figure}
    \centering
    \includegraphics[height=15ex]{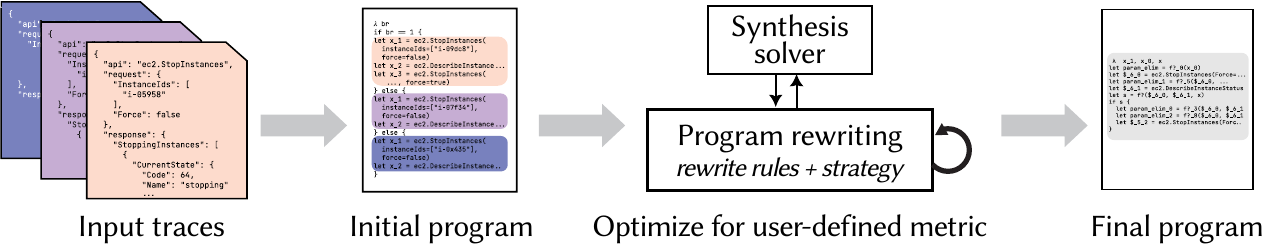}
    \vspace{-1.5ex}
    \caption{Overview of our synthesis approach.}
    \label{fig:overview}
    \Description{On the left, there is a set of input traces. These traces are used as input to the synthesizer. The synthesizer starts by writing an initial program using the traces. Then, the initial program is fed into a loop that optimizes for the user-defined metric. This loop uses a synthesis solver and a program rewriting engine, which entails the definition of program rewrites and a rewriting strategy. The outcome of this loop is the final output program.}
    \vspace{-1.7em}
\end{figure}

\autoref{fig:overview} summarizes our approach: an initial set of traces is provided, under the assumption that this set describes a task to be performed. We use that set of traces to construct an initial program. The core of our technique is the rewriting process guided by a user-defined program metric (the cost function) that may use an underlying syntax-guided synthesizer to generate some of the program's components. We present \systemname's rewrite rules and how we use the synthesizer in \autoref{sec:rewrite}. 

While our synthesis problem could be encoded as an optimization Satisfiability Modulo Theories (SMT) problem or as a SyGuS problem, the large search space prevents off-the-shelf state-of-the-art solvers from solving it.
In \systemname, we efficiently traverse the search space by combining SyGuS with \emph{program rewriting}. We start our search by building a trivial program, a lower bound on program behavior that is correct by construction. Then, we progressively rewrite it as long as we can decrease a cost function, generalizing the program and adding desired behavior, while \emph{provably maintaining correctness}. 
Program rewriting is a natural approach when it comes to optimizing programs for a given cost. For example, compiler optimization and superoptimization~\cite{superoptimization-cryptopt,superoptimizer-denali,superoptimizer-stochastic,sigcomm-21-rewrites} uses rewrites to improve the performance of programs across various dimensions.
%
A challenge to consider when developing a rewrite system is that an unsound rewrite can lead to incorrect programs. To reason about the correctness of our solution, we formally define a domain-specific language and a correctness statement that allow us to prove our rewrite rules correct w.r.t. the language and statement. The final result is guaranteed to be correct and is optimized for the user-defined~metric.


To show the practical applicability of \systemname, we implemented the algorithm and evaluated it on benchmarks gathered from cloud automation, filesystem manipulation, and document edition scripts. The \numTotalBenchmarks{} benchmarks were collected from custom tasks, existing AWS Automation Runbooks~\cite{aws-runbooks}, Blink Automations~\cite{blink-automation} and related work~\cite{apiphany}. We show that our approach generates scripts that accurately perform the task intended by the user for many tasks. This includes tasks with conditional control flow, loops, and hidden function calls between visible calls. We experiment with different strategies to apply rewrite rules and various metrics to optimize the final synthesized program. This shows that our approach is flexible and adaptable to different user requirements. 

In summary, we make the following contributions:
\begin{itemize}[topsep=.2ex]
    \item We describe a synthesis problem where the specification is a set of traces containing visible function calls, and the program to synthesize must perform those function calls and additionally implement hidden function calls and control flow.
    \item We implement \systemname, using a new approach that combines optimizing rewrites with traditional input-output guided program synthesis.
    \item We evaluate \systemname on a set of benchmarks built from AWS automation runbooks, previous work on API synthesis, and publicly available libraries.

\end{itemize}

\section{From Partial Traces to Programs}\label{sec:motivating-example}

We start by illustrating one execution of our synthesis procedure.  In this example, we synthesize a cloud management script that shuts down computing instances. The specification is a set of traces built from logs collected by the cloud provider (in this case, AWS \cite{what-is-aws}, one of the largest cloud providers) as the user performed the desired task manually a few times using a visual interface. The logs contain calls to the AWS API. For our synthesizer, the API calls are \textit{visible functions} in our program, whereas the data transformations with input and output data used for these API calls will be \textit{hidden functions}.

\paragraph{Motivating Example} An administrator might stop computing instances using the AWS console (shown in \autoref{fig:stop-ec2-instances-console}) by: 
\begin{enumerate}[noitemsep,topsep=.1em]
    \item selecting the instances they want to stop,
    \item clicking ``Instance state'', and
    \item selecting ``Stop instance'' from the options in the drop-down menu that appears.
\end{enumerate}
Clicks on the visual interface trigger calls to a cloud API, in this case, from the \dsl{ec2} service. A program could accomplish the same task by making the exact same API calls.

\begin{figure}
  \centering
  \begin{subfigure}[b]{0.43\textwidth}
    \centering
    \includegraphics[width=\textwidth]{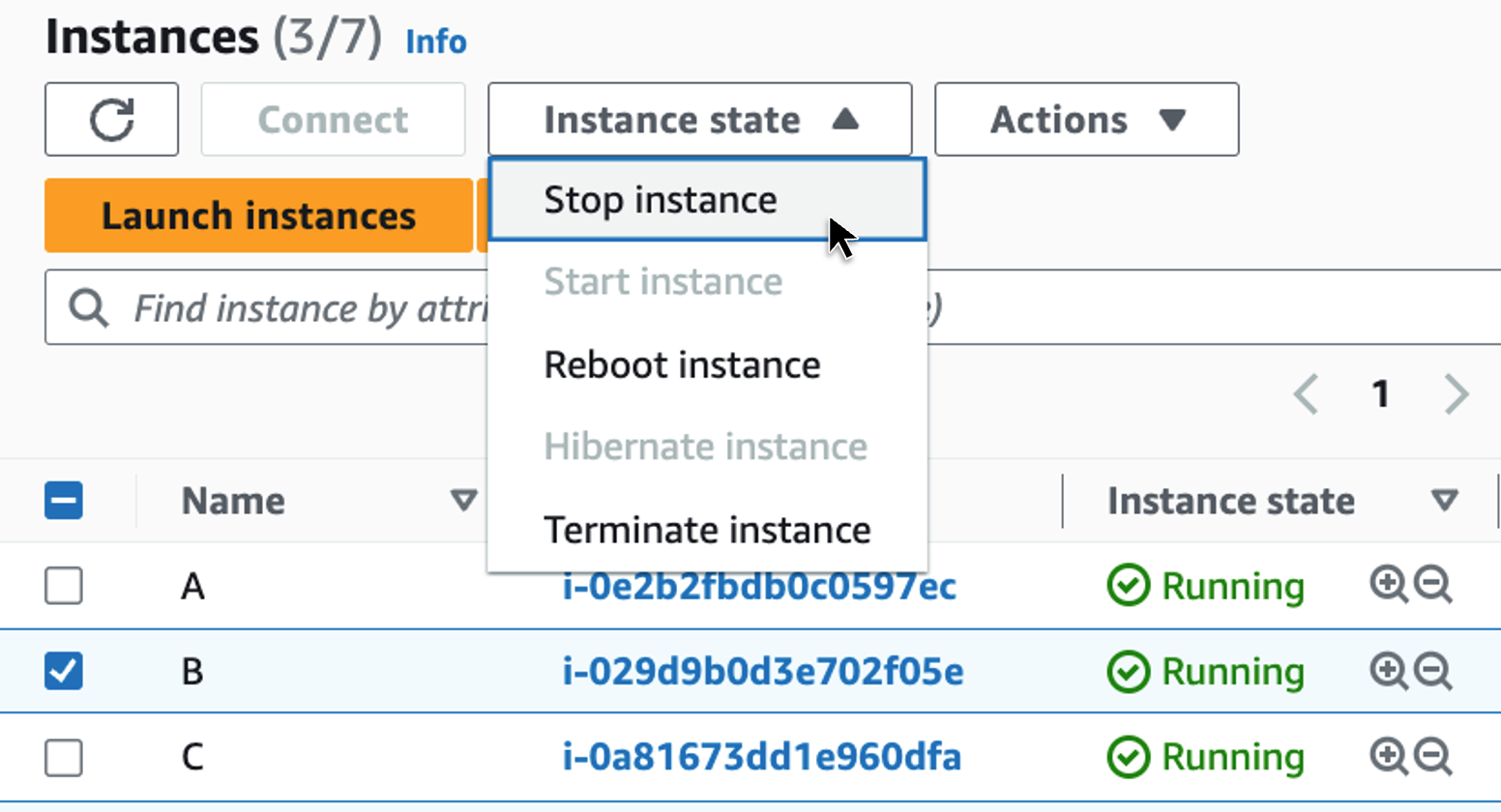}
    \vspace{-1em}
    \caption{Stopping instances in the AWS console}
    \Description{An AWS visual interface shows a list of computing instances, one of which is selected. A mouse cursor is clicking a button labeled "Stop Instance."}
    \label{fig:stop-ec2-instances-console}
  \end{subfigure}
  \hspace{2em}
  \begin{subfigure}[b]{0.43\textwidth}
    \centering
    \includegraphics[width=\textwidth]{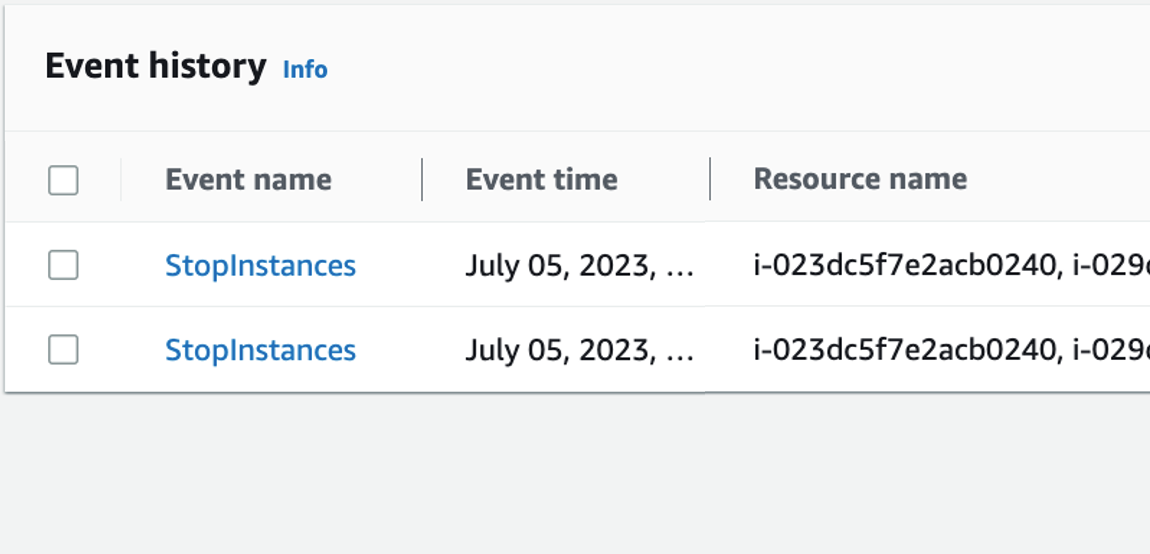}
    \vspace{-1em}
    \caption{The resulting log showing API calls}
    \Description{An AWS visual interface showing a list of timestamped API calls to a function called StopInstances.}
    \label{fig:cloudtrail-log}
  \end{subfigure}
  \vspace{-.8em}         
    \caption{Performing and monitoring actions on the AWS console}
    \Description{The figure contains two subfigures.}
    \label{fig:AWS-console-actions}
    \vspace{-1.3em}
\end{figure}

\begin{figure}
\begin{syrenColorListing}
lambda instanceId.
let ids = list(instanceId)
let _ = ec2.StopInstances(instanceIds=ids, force=false)
let s = ec2.DescribeInstanceStatus(instanceIds=ids, includeAllInstances=true)
let status = extractInstanceStatus(s)
if status != "stopped" { let _ = ec2.StopInstances(instanceIds=ids, force=true) }
where
extractInstanceStatus := $.InstanceStatuses[0].InstanceState.Name
\end{syrenColorListing}
\vspace{-1em}
\caption{Example of a program that stops an EC2 instance.}
\label{lst:stop-intances-cond}
\Description{A code snippet showing a program in Syren's programming language. It starts by making a call to ec2.StopInstances for a list of instance IDs, with the parameter force set to false. Then, it makes a call to ec2.DescribeInstanceStatus with the same instance IDs. Then, it uses its output in a call to a function extractInstanceStatus, whose output is saved to a variable status. Then, if status is not equal to "stopped", the script makes another call to ec2.StopInstances with the same IDs but with the parameter force set to true. The program terminates with the implementation of the function extractInstanceStatus, which extracts the field Name from the field InstanceState from the first element of the field InstanceStatuses fo the input object.}
\vspace{-1.3em}
\end{figure}

The program in \autoref{lst:stop-intances-cond} automates a task slightly more complex than the 3 steps described before. To perform this more complex task in the visual interface, a user would have to, after the previous~steps:
\begin{enumerate}[noitemsep,topsep=.1em]
    \setcounter{enumi}{3}
    \item click ``refresh'' to view the current status for their computing instances,
    \begin{enumerate}
        \item if the instance status is ``stopped'', then terminate the task here;
        \item otherwise once again click ``Instance state'', and
    \end{enumerate}
    \item select ``Force stop instance'' from the drop-down menu.
\end{enumerate}

\autoref{fig:cloudtrail-log} shows the logs created in AWS by these actions. We use these logs, which from now on we will refer to as \textit{traces}, as a specification for the synthesis problem. \textcolor{cameraready}{The synthesized program can be used to automate the task, reducing a many-click, repetitive, and error-prone task to a single press of a button.} Our approach ultimately generates a program similar to the one shown in \autoref{lst:stop-intances-cond} from the traces. We now give an overview of how our approach solves this specific problem.

\paragraph{Synthesis algorithm overview}

We start with a set of traces, each containing the API calls made while carrying out a task. The following two traces exemplify the task described previously:

\vspace{.5em}
\noindent
\begin{minipage}{\linewidth}
\begin{syrenListing}
|Trace \#1| := 
  (ec2.StopInstances("InstanceIds": ["i-09dc8"], "force": false), { ... })
  (ec2.DescribeInstanceStatus("InstanceIds":["i-09dc8"]),
      {"Statuses": [{"InstanceState": {"Code":64, "Name":"stopping"}, ...}], ... })
  (ec2.StopInstances("InstanceIds": ["i-09dc8"], "force": true), { ... })
\end{syrenListing}
\vspace{-2.2ex}
\begin{syrenListing}
|Trace \#2| := 
  (ec2.StopInstances("InstanceIds": ["i-07f34"], "force": false), { ... })
  (ec2.DescribeInstanceStatus("InstanceIds": ["i-07f34"]),
      {"Statuses":  [{"InstanceState": {"Code":80, "Name":"stopped"}, ...}], ...}) 
\end{syrenListing}
\end{minipage}
\vspace{.1em}

\noindent
Each trace is a sequence of pairs, and each pair represents an API call: the first element shows the API method name and its inputs, and the second shows the response to that API call. Both example traces start with two calls to the API methods, \dsl{ec2.StopInstances} with the parameter \dsl{force} set to \dsl{false}, and \dsl{ec2.DescribeInstancesStatus}. In trace \#1, the output of the call to \dsl{ec2.DescribeInstancesStatus} does \emph{not} show the current status as \dsl{"stopped"}, so we see a second call to \dsl{ec2.StopInstances} with \dsl{force} set to \dsl{true}. The call to \dsl{ec2.DescribeInstancesStatus} shows the current status as \dsl{"stopped"} in trace \#2, so no further calls are recorded.\looseness=-1

The first step in our synthesis pipeline, as shown in \autoref{fig:overview}, is to build an initial program that provably generates all the input traces for some initial global state. We do this by branching the execution on the value of a fresh integer variable, \dummybranch, and replaying each trace on a different branch. 
For our running example with two traces, we generate the following initial program:
\vspace{-1.7ex}
\begin{syrenColorListing}
lambda br
if br == 1 {
    let x_1_1 = ec2.StopInstances(instanceIds=["i-09dc8"], force=false)
    let x_1_2 = ec2.DescribeInstanceStatus(instanceIds=["i-09dc8"])
} else {
    let x_2_1 = ec2.StopInstances(instanceIds=["i-07f34"], force=false)
    let x_2_2 = ec2.DescribeInstanceStatus(instanceIds=["i-07f34"])
}
\end{syrenColorListing}
\vspace{-1.2ex}
\noindent
Our approach progressively transforms the program by applying rewrite rules that decrease an optimization metric. In this example, we use a metric that measures the syntactic complexity of the program:  we add 10 for each statement, 1 for each parameter, and 1 for each usage of \dummybranch, which is a synthetic variable that should only be used by the initial program. Initially, the program has cost 62.
The first rewrite applied to the program pulls the first call to \dsl{ec2.StopInstances} out of the if-statement and replaces its arguments, which were constants, with a ternary expression. A second application of the same rewrite rule extracts the call to \dsl{ec2.DescribeInstanceStatus}. Each rule reduces the cost by 9, eliminating one statement but introducing one usage of \dummybranch.
After applying these two rules, the intermediate program has cost~44:

\vspace{-1.7ex}
\begin{syrenColorListing}
lambda br
let x_1_1 = ec2.StopInstances(instanceIds=(br==1)?["i-09dc8"]:["i-07f34"], force=false)
let x_1_2 = ec2.DescribeInstanceStatus(instanceIds=(br==1)?["i-09dc8"]:["i-07f34"])
if br == 1 { let x_1_3 = ec2.StopInstances(instanceIds=["i-09dc8"], force=true) } 
\end{syrenColorListing}
\vspace{-1.2ex}

To eliminate usages of \dsl{br} in the ternary expressions, we need to replace the conditional expression \dsl{br==1} with another expression that evaluates to the same value but does not use \dummybranch. There are two ways to achieve this: either introduce a new input parameter that takes the value of the expression, or synthesize a function that will eventually evaluate to the conditional expression value. To synthesize a (nonconstant) function, \systemname considers as potential inputs all variables bound in the scope of the expression being replaced. In the first appearance of the expression \dsl{(br==1)?["i-09dc8"]:["i-07f34"]}, there are no variables bound in the scope that could be used as input to a data transformation. So, we have no choice but to introduce a new input parameter, \dsl{i_1}. \systemname replaces all usages of the original expression with \dsl{i_1}. Within the conditional branches, \systemname also replaces the usages of the value the expression evaluates to (considering the conditional). This results in the following program, with cost~43:

\vspace{-1.7ex}
\begin{syrenColorListing}
lambda br, i_1
let x_1_1 = ec2.StopInstances(instanceIds=i_1, force=false)
let x_1_2 = ec2.DescribeInstanceStatus(instanceIds=i_1)
if br == 1 { let x_1_3 = ec2.StopInstances(instanceIds=i_1, force=true) }
\end{syrenColorListing}
\vspace{-1.2ex}

The final rewrite for this example replaces the last conditional that depends on \dummybranch with the output of a new data transformation \( \phi \) over all variables in scope. 
After this last rewrite rule, the synthesized program is parametric on an implementation of that data transformation:

\vspace{-1.7ex}
\begin{syrenColorListing}
LAMBDA phi. lambda i_1
let x_1_1 = ec2.StopInstances(instanceIds=i_1, force=false)
let x_1_2 = ec2.DescribeInstanceStatus(instanceIds=i_1)
let c = phi(i_1, x_1_1, x_1_2)
if c { let x_1_3 = ec2.StopInstances(instanceIds=i_1, force=true) }
\end{syrenColorListing}
\vspace{-1.2ex}

\noindent
This rewrite is valid only if we can provide an implementation \dsl{f} for \(\phi\) that ensures the program can reproduce the input traces.
During the rewrite process, we maintain a mapping from the identifiers in the program to corresponding values in the traces. Then, we use these mappings to compute a set of input-output constraints that \dsl{f} must satisfy. For the traces and program in this example, we can extract the following two input-output pairs for the desired implementation \dsl{f}:
\vspace{-2ex}
\begin{lstlisting}[language=Syren,numbers=none,basicstyle=\small\linespread{.9}\ttfamily]
f(["i-09dc8"], {"StoppingInstances": [|\tiny\ldots|], "ResponseMetadata": {|\tiny\ldots|}},
   {"Statuses":[{"InstanceState":{"Code":64, "Name":"stopping"},|\tiny\ldots|}],|\tiny\ldots|}) = true |\textrm{\small for trace~\(\tau_1\),}|

f(["i-07f34"], {"StoppingInstances": [|\tiny\ldots|], "ResponseMetadata": {|\tiny\ldots|}},
   {"Statuses":[{"InstanceState":{"Code":80, "Name":"stopped"},|\tiny\ldots|}],|\tiny\ldots|}) = false |\textrm{\small for trace~\(\tau_2\).}|
\end{lstlisting}
\noindent
We encode the problem into a syntax-guided synthesis solver to generate a solution, which yields:\footnote{In practice, we need at least one more trace and its respective input/output example to synthesize this solution. If we consider only the two example traces shown, a simpler implementation is synthesized for \dsl{f}: \simplesyreninline{f := (i_1, x_1_1, x_1_2) -> i_1 == ["i-09dc8"]}.}

\noindent
\lstinline[language=Syren,escapeinside=||,basicstyle=\small\linespread{.9}\ttfamily]^f := (i_1, x_1_1, x_1_2) -> x_1_2.InstancesStatuses[0].InstanceState.Name != "stopped"^.

\noindent
Substituting \(\phi\) for \dsl{f} yields a program that is correct by construction, with a minimal cost of 41. This final program is syntactically equivalent to the one in \autoref{lst:stop-intances-cond}.
\section{Background and Definitions}\label{sec:prelim}

In this section, we formally introduce concepts necessary to explain our approach to synthesizing scripts that compose visible side-effecting function calls with conditionals, loops, and (hidden) pure function calls.
In \autoref{sec:dsl}, we define the domain-specific language (DSL) syntax of our synthesized programs. This DSL is an intermediate representation that we can easily convert to most common scripting languages. 
%
Next, in \autoref{sec:traces-definition}, we define \textit{program traces}, which we use as an input specification for synthesis. 
Finally, in \autoref{sec:semantics}, we define the semantics of our language, which relates a program to input and output states, as well as to traces. These semantics allow us to prove properties about the manipulation of the DSL to prove our approach~correct. 

\subsection{Core Language Syntax}\label{sec:dsl}

\begin{figure}
\footnotesize
    \centering
    \begin{grammar}
\firstcase{P}{\dsllambda \overline{x}. \nonterm{I}\ \dslwhere \ \overline{f := \nonterm{F}}}{Program}

\firstcase{I}{ \epsilon \gralt \nonterm{S}\ \nonterm{I}}{Instructions}
\firstcase{S}{\dsllet\ x = \nonterm{E}\ \dslin}{Pure binding}
\otherform{\dslif\ \nonterm{B}\ \{ \nonterm{I} \} \ \dslelse\ \{ \nonterm{I}\}}{Conditional}
\otherform{\dslretryuntil{\nonterm{I}}{\nonterm{B}}}{Retry until}
\otherform{\dslforeach{x}{L}{\nonterm{S}}}{Foreach loop}
\otherform{\dslend}{Return}

\firstcase{E}{\mathbb{A}(\overline{x})}{Visible function call}
\otherform{f(\overline{x})}{Hidden function call}
\otherform{\nonterm{B}\ \dsl{?}\ \nonterm{E}\ \dsl{:}\ \nonterm{E}}{Ternary expression}

\firstcase{B}{\top \gralt \bot \gralt \nonterm{B} \vee \nonterm{B} \gralt \nonterm{B} \wedge \nonterm{B} \gralt \neg \nonterm{B}}{Predicates}
\otherform{x = \nonterm{C}}{Value check}
\otherform{x (\geq \mid > \mid \leq \mid <) y}{Value comparison}

\firstcase{F} 
{??}{Pure Function}

\firstcase{C}{s \in \mathtt{string} \gralt n \in \mathbb{Z} \gralt b \in \{\dsl{true}, \dsl{false}\} }{Constants}
    \end{grammar}
    \caption{Core scripting language.}
    \label{fig:core-dsl}
    \Description{A set of rules describing a program in the core scripting language. }
    \vspace{-3ex}
\end{figure}

Figure~\ref{fig:core-dsl} presents the syntax of our core DSL. A program $\nonterm{P}$ is a function with input variables\footnote{We write $\overline{x}$ to denote zero or more occurrences of $x$.} $\overline{x}$ and a set of hidden functions $\overline{f := \nonterm{F}}$.  The body of the program is an instruction list $\nonterm{I}$, either empty ($\epsilon$) or with statements. A statement $\nonterm{S}$ can be a simple binding, a conditional, a loop, or the instruction that marks the end of the script. 
A binding $\dsllet \ x = e_1\ \dslin\ s_2$ binds $e_1$ to $x$ in $s_2$. 
Conditionals $\dslite{b}{s_1}{s_2}\ s_3$ execute $s_1$ if $b$ is true, otherwise $s_2$, and then $s_3$.
Our language has two forms of loops. $\dslretry\, s\, \dsluntil\, b$ (retry until) executes the instructions in $s$ at least once, until $b$ is true, or some predefined maximum number of retries is reached. $\dslforeach{x}{L}{s}$ iterates through the list $L$, binding $x$ to each element and executing $s$.

Expressions $\nonterm{E}$ can be visible or hidden function calls. A visible function call $\mathbb{A}(\overline{x})$ is a call to some externally defined function $\mathbb{A}$ with arguments $\overline{x}$, and a hidden call $f(\overline{x})$ is a call to a pure function $f$ whose implementation \textcolor{cameraready}{$\mathcal{F}$} is defined in the program. 
Visible functions can have side effects on the outside world, changing the results of future calls. However, they do not change the \textcolor{cameraready}{local} state of the DSL execution. The hidden functions are \emph{pure}, and their specific syntax depends on the chosen domain.

We clearly separate \textcolor{cameraready}{pure function implementations $\mathcal{F}$} from the rest of the program for two reasons: to simplify reasoning about variable usage and whole-program rewrites, and 
\textcolor{cameraready}{to highlight the fact that our DSL is agnostic of the hidden functions domain.}
In our implementation of \systemname, we consider a minimal language of hidden functions, which includes a JSONPath as well as other basic operations over strings, numbers, and Booleans. However our approach can be generalized to any other language for hidden functions, as long as expressions in that language can be synthesized.

\begin{example}
The script in \autoref{lst:stop-intances-cond} is an example of a script written in our DSL. The script takes a single input, \dsl{instanceIDs} and defines one data operation \dsl{extractInstanceStatus}. The visible functions are the API methods \dsl{StopInstances} (called twice), and \dsl{DescribeInstanceStatus}.
\end{example}

\subsection{Program Traces}\label{sec:traces-definition}

Our synthesis starts from \emph{observable} traces that can be produced by the program's execution. The \emph{observable} traces contain only records of the visible function calls made by the program and their results; the hidden function calls do not appear in traces 
.
Formally, a trace is a (possibly empty) finite sequence of records of all visible calls that the run of the program makes:
\[ \tau := \emptytrace \mid (\mathbb{A}(\overline{v}), e) :: \tau \quad \text{(traces)}, \]
where $\emptytrace$ is the empty trace, operator $\bullet$ performs concatenation, and $\emptytrace \bullet \tau = \tau \bullet \emptytrace = \tau$ for any trace~$\tau$. 
Each record is a pair \((\mathbb{A}(\overline{v}), e)\), where the first element states the name of the function \(\mathbb{A}\) and the inputs to the call \(\overline{v}\); the second element, \(e\), is the response to the call. \(e\) is an expression in the same language as the hidden functions. 

\subsection{DSL Semantics}\label{sec:semantics}

Next, we define the semantics of our language as the relation~$\eval$, presented in Figure~\ref{fig:small-step-sem-noret}. The relation~$\eval$ maps a pair of a program body and state to a triple of \textcolor{cameraready}{local} state, trace and continuation token. \textcolor{cameraready}{In our DSL semantics, we refer to two different notions of state. The local state, \(\sigma\), stores the bindings of every variable assigned in the program at a given point. The global state, \(G\), represents external resources accessed by the visible functions in the trace. The notion of global state is necessary because visible functions are not pure functions of their inputs; depending on the resources they access, two calls to the same function with the same inputs might return different outputs.} The continuation token is either $\continue$, indicating that evaluation must continue, or~$\returns$, indicating that the evaluation must stop.

\begin{figure}
\footnotesize
\begin{mathpar}
\inferrule[Seq]{
(s, \sigma) \eval (\sigma', \tau, \continue) \\ (S', \sigma') \eval (\sigma'', \tau', cr)
}{
(s\ S', \sigma) \eval (\sigma'', \tau \bullet \tau', cr)
}
\and
\inferrule[Seq-S-Term]{
(s, \sigma) \eval (\sigma', \tau, \returns)
}{
(s\ S', \sigma) \eval (\sigma', \tau, \returns)
}
\and
\inferrule[Ret]{
}{
    (\dslend, \sigma) \eval (\sigma, \emptytrace, \returns)
}
\and 
\inferrule[Emp]{
}{
    (\epsilon, \sigma) \eval (\sigma, \emptytrace, \continue)
}

\and

\inferrule[Ite-$\top$]{
\sigma \models b \\ (S_\top, \sigma) \eval (\sigma', \tau, cr)
}{
(\dslite{b}{S_\top}{S_\bot}, \sigma) \eval (\sigma', \tau, cr)
}

\and

\inferrule[Ite-$\bot$]{
\sigma \not\models b \\ (S_\bot, \sigma) \eval (\sigma'', \tau', cr)
}{
(\dslite{b}{S_\top}{S_\bot}, \sigma) \eval (\sigma'', \tau', cr)
}
\and
\inferrule[Retry-Until-Continue]{
    (S, \sigma) \eval (\sigma', \tau', cr) \\ \sigma' \not\models b \wedge \#_\iota < K \\
\dslretryuntil{S}{b}, \sigma'[\#_\iota \to \#_\iota + 1]) \eval (\sigma'', \tau'', cr) 
}{
    ({}^{\#\iota}\dslretryuntil{S}{b}, \sigma) \eval (\sigma'', \tau' \bullet \tau'', cr)
}
\and
\inferrule[Retry-Until-Stop]{
    (S, \sigma) \eval (\sigma', \tau', cr) \\ \sigma' \models b \vee \#_\iota \geq K
}{
    ({}^{\#\iota}\dslretryuntil{S}{b}, \sigma) \eval (\sigma'[\#_\iota \to 0], \tau', cr)
}
\and
\inferrule[Retry-S-Term]{
    (S, \sigma) \eval (\sigma', \tau', \returns)
}{
    ({}^{\#\iota}\dslretryuntil{S}{b}, \sigma) \eval (\sigma', \tau', \returns)
}
\and
\inferrule[For-Continue]{
    (S, \sigma[x \to L[\#_\iota]]) \eval (\sigma', \tau', cr) \\ \sigma' \models \#_{\iota} < |L| \\
    ({}^{\#\iota}\dslforeach{x}{L}{S}, \sigma'[\#_\iota \to \#_\iota + 1 ]) \eval (\sigma'', \tau'', cr) 
}{
    ({}^{\#\iota}\dslforeach{x}{L}{S}, \sigma) \eval (\sigma'', \tau' \bullet \tau'', cr)
}
\and
\inferrule[For-Stop]{
    (S, \sigma) \eval (\sigma', \tau', cr) \\ \sigma' \#_\iota \geq |L|
}{
    ({}^{\#\iota}\dslforeach{x}{L}{S}, \sigma) \eval (\sigma'[\#_\iota \to 0], \tau', cr)
}
\and
\inferrule[For-S-Term]{
    (S, \sigma) \eval (\sigma', \tau', \returns)
}{
    ({}^{\#\iota}\dslforeach{x}{L}{S}, \sigma) \eval (\sigma', \tau', \returns)
}
\and
\inferrule[Hidden]{
    \sigma \models f:= \mathcal{F} \wedge \exists e \cdot \mathcal{F}(\overline{y}) = e
}{
    (\dsllet\ x = f(\overline{y}) \dslin, \sigma) \eval (\sigma[x \to e], \emptytrace, \continue) 
}
\and
\inferrule[Visible]{
    \sigma \models \exists \overline{v} \cdot \overline{y} = \overline{v} \\ \mathbb{A}(G, \overline{v}) \downarrow e
}{
    (\dsllet\ x = \mathbb{A}(\overline{y}) \dslin, \sigma) \eval (\sigma[x \to e], (\mathbb{A}(G, \overline{v}), e), \continue) 
}
\end{mathpar}
\caption{Big-step semantics}\label{fig:small-step-sem-noret}
\Description{This figure contains a set of rules that define Syren's language's big step semantics. The rules are described in detail in the main body of the paper.}
\vspace{-1.3em}
\end{figure}

The rule \textsc{Seq} specifies how a statement $s$ followed by instructions $S'$ is evaluated sequentially and traces are concatenated. The rule \textsc{Seq-S-Term} handles the case where the first statement \emph{terminates} the evaluation of the program.
\textsc{Ret} states that the statement $\dslend$ always terminates early with an empty trace.
\textsc{Emp} generates an empty trace, does not change the \textcolor{cameraready}{local} state, and always continues evaluation.
In \textsc{Ite-$\top$} and \textsc{Ite-$\bot$} for conditionals, either the branch with instructions~$S_\bot$ or $S_\top$ are evaluated depending on whether the \textcolor{cameraready}{local} state entails $b$ or not. The continuation or termination token~$cr$ of the if-then-else is the same as the token in the evaluation of the branch, in particular, the statement terminates the evaluation when the branch terminates evaluation.

The rules \textsc{Retry-Until-Continue}, \textsc{Retry-Until-Stop}, and \textsc{Retry-S-Term} define how the retry-until statements are evaluated. 
Note that retry-until does not have the same semantics as a while loop: it will always terminate, and the predicate $b$ is not guaranteed to hold when the loop ends. We assume a constant $K$ that bounds the number of times the body of a retry-until statement can be "retried". For any set of traces, we can select a $K$ higher than the longest trace. This ensures that $K$ is high enough that any trace can be regenerated by the transformed program without timing out. We discuss this more in \autoref{sec:evaluation}.

Each retry-until statement is given a unique identifier, $\#_\iota$, and each of those identifiers is assigned $0$ in the initial \textcolor{cameraready}{local} state. The rule \textsc{Retry-Until-Continue} states that a retry-until statement with identifier $\#_\iota$ evaluates to state~$\sigma''$ and trace $\tau' \bullet \tau''$  when one iteration results in $\sigma'$ and $\tau'$, $b \wedge \#_\iota < K$ holds in state $\sigma'$,  and evaluating again the retry-until statement with the state $\sigma'$ where $\#_\iota$ is incremented results in $\sigma''$ and $\tau''$. The rule \textsc{Retry-Until-Stop} handles the case where $b \wedge \#_\iota < K$ does \emph{not} hold after evaluating the body of the loop. The rule \textsc{Retry-S-Term} handles the case where the body of the retry loop returns, and therefore the entire program returns. The rules for for-loops (\textsc{For-Continue}, \textsc{For-Stop} and \textsc{For-S-Term}) are similar, the main difference being that the variable~$x$ is bound at each new iteration and the stop condition depends on the size of the list~$L$, not the value of the Boolean $b$.

We differentiate binding on the type of expression they bind. If it is a call to a hidden function \textcolor{cameraready}{$f$ (rule \textsc{Hidden}) then the local state is modified by binding $x$ to the value $\mathcal{F}(\overline{y})$ evaluates to in the current state~$\sigma$, according to the semantics of the data-transformation domain and assuming $\mathcal{F}$ is $f$'s implementation}. The trace is unchanged \textcolor{cameraready}{by the hidden function}. If it is a call to a visible function (rule \textsc{Visible}), the arguments of the call are evaluated in $\sigma$, the result of the call is bound to $x$ in the \textcolor{cameraready}{local} state, and the call to $\mathbb{A}$ with the input values is recorded in the trace. 

The relation $\mathbb{A}(G, \overline{v}) \downarrow e$ means that the call to externally defined function $\mathbb{A}$ with input $\overline{v}$ in global state $G$ returns a response $e$. We implicitly update the global state as a function of each visible call and transfer it through sequences. This means that two programs that start in the same global state and execute the same visible calls receive the same responses to those calls. This formulation allows us to reason about the semantics of the program given existing pairs of input-output examples of calls (the traces) without actually executing any of the calls; we only need to assume the initial global state is the same as in the traces. A limitation of this approach is that time is not considered, so our approach will be unsound in situations where responses implicitly depend on time as opposed to ordering.

Finally, we introduce notation for relating traces with programs and states.

\begin{definition}[Program Evaluation]
Let $P := \lambda \overline{x}. S \ \dslwhere \ \overline{f := \nonterm{F}}$ a program and $\sigma$ a state mapping every variable in $\overline{x}$ to some value and every $f$ into the corresponding $\nonterm{F}$. 
Then, given a starting global state, there might exist exactly one trace $\tau$ and termination token \(cr\) such that $(S, \sigma) \eval (\sigma', \tau, cr)$. If and only if the trace and termination token exist, we say that $\tau$ is a \emph{trace of} $P$ with input $\sigma$ and write $P(\sigma) = \tau$. While not all syntactically valid programs fully evaluate, all \emph{synthesized} programs evaluate by construction. Given that all programs we discuss are synthesized, we no longer need to consider programs that do not successfully evaluate.
\end{definition}
\section{Synthesis Problem}\label{sec:problem}

As we illustrated with our motivating example, the synthesis problem solved in this paper consists in finding a program \(P\) in the language described in Figure~\ref{fig:core-dsl} that reproduces a set of input traces~\( T_{in} \). In this section, we formalize this intuitive correctness constraint and define our synthesis problem as a combination of the correctness constraint and another constraint on the quality of the program.

\subsection{Correct Solutions}
We use the notion of \emph{trace subsumption} to describe one program that can generate at least the same traces as another:

\begin{definition}[$\sqsupseteq$]\label{def:subsumes}
A program $P'$ subsumes a program $P$ ($P' \sqsupseteq P$) if and only if for every state $\sigma$ and trace $\tau$ such that $P(\sigma) = \tau$, there exists a state $\sigma'$ such that $P'(\sigma') = \tau$.
\end{definition}

Note that $\sqsupseteq$ is a partial order on programs. If $P \sqsupseteq P'$ and $P' \sqsupseteq P$ then $P$ and $P'$ are \emph{trace equivalent}. In general, we are interested in transformations that preserve subsumption (i.e. $P \rightsquigarrow P'$ only if $P' \sqsupseteq P$, where $\rightsquigarrow$ is a transformation), not just trace equivalence.
Formally, our correctness constraint \( \Psi \) is
\begin{equation}\label{eq:correctness}
\Psi(P, T_{in}) \equiv \forall \tau_i \in T_{in} \cdot \exists \sigma \cdot P(\sigma) = \tau_i.
\end{equation}

The set \(T_{in}\) is a set of finite input traces \(\tau_1, \tau_2, ..., \tau_t\)\footnote{We always assume that \(|T_{in}| > 1\).}.
There is always a trivial solution to \( \Psi \) for a given set of traces \( T_{in} \). It can be constructed using a single integer parameter \dummybranch and \( |T_{in}|\) branches, where each branch can be selected with a value for \dummybranch, and the branch makes the API calls contained in the \dummybranch-th trace of \( T_{in} \). We show in \autoref{sec:motivating-example} how this solution is built directly from the traces; it simply replays each of the traces, and the set of possible traces of the program is exactly \( T_{in} \).

Another less trivial solution would consist of combining visible function calls when possible but leaving all the inputs of the visible calls as parameters of the program, thus always discarding the output of the visible function call. This is also not an acceptable solution. Although it generalizes to other inputs, the generalization only comes from the program being entirely parameterized.

\subsection{Quality Constraint}

Although trivial solutions exist, they usually will not be what a user would expect as output; there is an expectation that the solution is a generalization of the traces. There are also infinitely many correct (solutions to \( \Psi \)) programs that are very general; consider, for example, a program listing all possible syntactic productions that satisfy the correctness constraint in branches. However, those programs are also not typically what the user expects.

To address this challenge, we assume the existence of a \emph{program cost function} \( \chi \) that, given a program \( P \) and a set of traces \( T_{in} \), returns a positive number. This program cost function reflects what the user expects; a good program is one with a low cost. For example, the program cost function could return the count of branches and the count of parameters of the program, indicating that the user desires a program with low complexity that is likely to generalize well.
The goal of the synthesizer is to find a program that minimizes the cost function. Formally, the goal is to solve, given a fixed set of traces \( T_{in}\),
\begin{equation}\label{eq:problem}
    \min_{\forall P \cdot \Psi(P,T_{in})} \chi(P, T_{in}).
\end{equation}

In this paper, we describe a generic algorithm that is parametric in \(\chi\), first by describing our rewrites in \autoref{sec:rewrite}, and then by describing the search approach in \autoref{sec:synthesis}.

\section{Rewriting Programs}\label{sec:rewrite}

Our synthesis algorithm applies a succession of rewrite rules to transform an initial trivial program into a more general and user-friendly one.
Each of these rewrite rules provably maintains the program’s correctness constraint, \(\Psi\), so that all intermediate programs can generate all input traces in \(T_{in}\).  
We split our rules into two categories, \emph{synthesis} rules and \emph{refinement} rules, depending on how they maintain correctness.

\begin{lemma}
\label{lem:sub-preserves}
Subsumption preserves correctness: 
\(\Psi(P, T_{in}) \bigwedge P' \sqsupseteq P \implies \Psi(P', T_{in})\). 
\end{lemma}
\begin{proof}
If \(P\) can generate all traces in \(T_{in}\), and \(P'\) can generate all traces that \(P\) can generate, then \(P'\) can generate all traces in \(T_{in}\).
\end{proof}

\emph{Refinement rewrite rules} preserve subsumption: for all \(P\) and \(P'\), if a refinement rule rewrites \(P\) into \(P'\), \(P \rightsquigarrow P'\), then \(P' \sqsupseteq P\). By \autoref{lem:sub-preserves}, these rules preserve the correctness when applied to a correct program.
The following is an example of the application of a refinement rule that extracts an identical instruction \(\color{red}\mathcal{R}\) from both branches of an if-then-else statement. This rewrite does not change the semantics of the program but improves its readability by reducing its number of instructions.

\begin{center} \vspace{-1.3em}
\begin{tabular}{m{0.375\textwidth} m{.04\textwidth} m{0.42\textwidth}}
\begin{lstlisting}[language=Syren,numbers=none,basicstyle=\small\linespread{.9}\ttfamily]
lambda |\(\overline{x}\)|.
|\(\mathcal{U}\)| if |\(\mathcal{C}\)| { |\(\color{red}\mathcal{R}\)| |\(\mathcal{S}\)| } else { |\(\color{red}\mathcal{R}\)| |\(\mathcal{T}\)| } |\(\mathcal{V}\)|
\end{lstlisting} 
& \(\rightsquigarrow\)& 
\begin{lstlisting}[language=Syren,numbers=none,basicstyle=\small\linespread{.9}\ttfamily]
lambda |\(\overline{x}\)|.
|\(\mathcal{U}\)| |\(\color{red}\mathcal{R}\)| if |\(\mathcal{C}\)| { |\(\mathcal{S}\)| } else { |\(\mathcal{T}\)| } |\(\mathcal{V}\)|
\end{lstlisting} 
\\
\end{tabular}
\vspace{-1.2em}
\end{center}

\emph{Synthesis rewrite rules} all follow the same pattern: replace an expression \(e\) with the output of a call to a to-be-synthesized pure function \(\phi\). \(\phi\) is not visible in the traces, so we refer to it as a \emph{hidden} function. All the bound variables available at the location are used as arguments to \(\phi\), except when the rule is trying to eliminate a parameter.
The correctness of a synthesis rewrite rule is conditioned by the existence of a solution for the hidden function calls they introduce. Formally, we denote by \(\Lambda \overline{\phi} \cdot P\) a program \(P\) parametric on a set of hidden functions~\(\overline{\phi}\). For a given set of implementations \(\overline{f}\), \( (\Lambda \overline{\phi} \cdot P)(\overline{f})\) is a valid program in our DSL.
A~single synthesis rule \(\rightsquigarrow\) rewrites \(P\) to a program \(P' := \Lambda \phi \cdot P_s\) parametric on some hidden function~\(\phi\). The rewrite rule \(P \rightsquigarrow P'(f)\) is correct for some function \(f\) only if \(P'(f) \sqsupseteq P_{in}\) where \(P_{in}\) is the initial program. By \autoref{lem:sub-preserves}, if \(\Psi(P_{in}, T_{in})\) and  $P'(f) \sqsupseteq P_{in}$, then \(\Psi(P', T_{in})\). Note that we can chain multiple synthesis rules together and check for correctness only later, i.e. rewrite \( P \rightsquigarrow \Lambda \phi \cdot P_s \rightsquigarrow \Lambda \phi,\phi' \cdot P_s' \) and then find \(f\) and \(f'\) later to instantiate \( \phi\) and \(\phi'\). 
We explain how to find the implementation of hidden functions \(f\) in \autoref{sec:data-transform-synthesis}.
For a list of \systemname's rewrite rules, the reader can refer to 
\iffullversion
\autoref{sec:appendix-rewrite-rules}.
\fi
\begin{example}\label{ex:rewrites}
We illustrate below how we apply a sequence of rewrite rules to generalize programs and produce an acceptable solution. Suppose that we have constants \dsl{c1,c2,c3}, visible functions \(\dsl{A},\dsl{B}\)  and some initial program \(P_{in}\) as shown below:

\vspace{.5em}
\noindent
\begin{tabular}{m{0.18\textwidth} m{.04\textwidth} m{0.30\textwidth} m{.04\textwidth} m{0.29\textwidth}}
\(P_{in}:\) & & \(P_1:\) & & \(P_2:\) \\[-1.3em]
\begin{lstlisting}[language=Syren,numbers=none,basicstyle=\small\linespread{.9}\ttfamily]
lambda br.
if br = 1 {
  let x1 = A(c1)
  let y = B(c2)
} else {
  let x2 = A(c3)
}
\end{lstlisting} 
& \(\rightsquigarrow\)& 
\begin{lstlisting}[language=Syren,numbers=none,basicstyle=\small\linespread{.9}\ttfamily]
lambda br.
let a = (br = 1) ? c1 : c3
if br = 1 {
  let x1 = A(a)
  let y = B(c2)
} else { let x2 = A(a) }
\end{lstlisting}
& \(\rightsquigarrow\)& 
\begin{lstlisting}[language=Syren,numbers=none,basicstyle=\small\linespread{.9}\ttfamily]
lambda br.
let a = (br = 1) ? c1 : c3
let x = A(a)
if br = 1 { let y = B(c2) }
\end{lstlisting}
\end{tabular}

\vspace{-1.1em}
\noindent
\begin{tabular}{m{.05\textwidth} m{0.18\textwidth} m{.05\textwidth} m{0.25\textwidth} m{.05\textwidth} m{0.25\textwidth}}
& \(P_3:\) & & \(P_4:\) & & \(P_5:\) \\[-1.3em]
\(\rightsquigarrow\)& 
\begin{lstlisting}[language=Syren,numbers=none,basicstyle=\small\linespread{.9}\ttfamily]
lambda br, d.
let x = A(d)
if br = 1 {
  let y = B(c2)
}
\end{lstlisting}
& \(\rightsquigarrow\)& 
\begin{lstlisting}[language=Syren,numbers=none,basicstyle=\small\linespread{.9}\ttfamily]
LAMBDA phi lambda br, d.
let x = A(d)
let c = phi(d, x)
if c { let y = B(c2) }
\end{lstlisting}
& \(\rightsquigarrow\)& 
\begin{lstlisting}[language=Syren,numbers=none,basicstyle=\small\linespread{.9}\ttfamily]
LAMBDA phi lambda br, d.
let x = A(d)
let c = phi(d, x)
if c { let y = B(c2) }
\end{lstlisting}
\end{tabular}
\vspace{-.9em}

\noindent
We rewrite \(P_{in}\) using a refinement rule that introduces a new variable \dsl{a}, which is bound to the constants \dsl{c1} or \dsl{c3} in the conditional, and then used as argument to the calls to \dsl{A}. \(P_1\) is the resulting program. Then, we apply to \(P_1\) the refinement rule shown in \autoref{sec:rewrite}, which factors the calls to \dsl{A} out of the conditional, resulting in \(P_2\).
The third rewrite eliminates the expression \lstinline[language=Syren]^let a = if br = 1 {c1} {c3}^ which depends on \dsl{br} and introduces a parameter \dsl{d} that takes its value.
This rewrite provably maintains correctness and produces a program, \(P_3\), generalized to any input \dsl{d}. 
A fourth rewrite introduces a function parameter \(\phi\) (to be synthesized) to eliminate \dsl{br} from the conditional, resulting in~\(P_4\).
The final rewrite eliminates the unused parameter~\dsl{br}.
 
\end{example}

\subsection{Trace Valuation}\label{sec:synthesis-rewrite-rules}


%

Rewrites maintain correctness by ensuring that, given an implementation for the hidden function introduced, the program can still generate the initial set of traces. While refinement rewrite rules are correct for all inputs of the program, synthesis rewrite rules require more attention.
To keep track of this correctness constraint, we maintain an augmented \textcolor{cameraready}{local} state \(\sigma\), the trace valuation of the program, which relates variables in the program with a specific trace and a concrete value (and an iteration number for variables in loops).
The trace valuation stores the relationships necessary for the rewritten program to reproduce each trace \(\tau \in T_{in}\), and each rewrite rule application modifies that state to maintain the invariant. This ensures that the value of all expressions of the program for a certain trace is always known, either because the variable's value is known, or the expression's value can be computed from those known values. Given an expression \(e\), trace valuations \( \sigma\) and trace \( \tau \) we denote the value of \( e \) in trace \( \tau \) and state \( \sigma \) by \( \valuation{e}{\sigma}{\tau} \).

The initial program has only one variable \dsl{br}, and initially  \(\valuation{\dsl{br}}{\sigma}{\tau_i}= i\) for each trace \(\tau_i \in T_{in}\). 
Then, each rewrite rule \( \rightsquigarrow \) in our system is accompanied with a trace valuation transformation \(t\), which we denote by \(\rightsquigarrow_t\).
We introduce a new function $I$, which extracts the program parameters from a state, and overload $\rightsquigarrow$ to apply to sequences of instructions as well, instead of entire programs only.
Each \(t\) and associated rewrite $\rightsquigarrow_t$ is correct when for each input trace, when a rewrite and corresponding trace valuation transformation are applied, if the program runs with the inputs of the updated state then it produces the same input trace:

\vspace{-1.5em}
\[\tau_i \in T_{in} \wedge t(\sigma_0') = \sigma_1' \wedge (S, \sigma_0) \eval (\sigma_0', \tau_i, cr) \wedge S \rightsquigarrow_t S' \implies (S', I(\sigma_1')) \eval (\sigma_1', \tau_i, cr)\]

\vspace{-.2em}
The function $t$ will encapsulate the parameter updates of the rewrite and any functions introduced. This rule requires the correctness of~$S$ (in the third conjunct of the hypothesis), and the conclusion directly implies the correctness of~$S'$, where $I(\psi_1')$ witnesses the existential needed by the correctness statement.

Synthesis rules replace an expression or number of expressions with a hidden function.

\begin{example}\label{ex:trace-state}
The refinement rule of Example~\ref{ex:rewrites} can specified with its transformation \(t_3\):
\begin{center}
\begin{tabular}{m{0.42\textwidth} m{.06\textwidth} m{0.33\textwidth}}
\(P_2\): &  & \(P_3\): \\[-1.2em]
\begin{lstlisting}[language=Syren,numbers=none,basicstyle=\small\linespread{.9}\ttfamily]
lambda br.
let a = if br = 1 { c1 } else { c3 }
let x = A(a)
if br = 1 { let y = B(c2) }
\end{lstlisting}
& \(\rightsquigarrow_{\color{red} t_3}\)  & 
\begin{lstlisting}[language=Syren,numbers=none,basicstyle=\small\linespread{.9}\ttfamily]
lambda br, d.
let x = A(d)
if br = 1 { let y = B(c2) }
\end{lstlisting}
\end{tabular}
\vspace{-1em}

{\small \raggedleft where \({\color{red} t_3}(\sigma) = \sigma[(d,\tau) \mapsto \valuation{\dsl{if br = 1 {c1} else {c3}}}{\sigma}{\tau}]\)}
\end{center}

\noindent
That is, the trace valuation transformation \(t_3\) corresponding to this rewrite assigns the resulting value of evaluating the eliminated expression \(\dsl{if br = 1 {c1} else {c3}}\) to the new parameter \dsl{d}. 
Concretely, if the program above is synthesized from the two traces:
\[
\tau_1 = (\dsl{A(c1)}, o_{1}) :: (\dsl{B(c2)}, o_{2})
\quad \text{ and } \quad
\tau_2 = (\dsl{A(c3)}, o_{3})\\
\]

Given that \(\valuation{\dsl{br}}{\sigma}{\tau_1} = 1\), we have \(\valuation{\dsl{if br = 1 {c1} else {c3}}}{\sigma}{\tau_1} = \dsl{c1}\), and therefore 
\(\valuation{d}{t_3(\sigma)}{\tau_1} = \dsl{c1}\). For the second trace, we would have \(\valuation{d}{t_3(\sigma)}{\tau_1} = \dsl{c3}\).

In a following step in Example~\ref{ex:rewrites}, we apply the following synthesis rule to the program:
\begin{center}
\vspace{-.8em}
\begin{tabular}{m{0.31\textwidth} m{.01\textwidth} m{0.22\textwidth} m{0.38\textwidth}}
\(P_3\): &  & \(P_4\): &  \\[-1.2em]
\begin{lstlisting}[language=Syren,numbers=none,basicstyle=\small\linespread{.9}\ttfamily]
lambda br, d.
let x = A(d)
if br = 1 { let y = B(c2) }
\end{lstlisting}
& \hspace{-1.5em} \(\rightsquigarrow_{\color{red}{t_4}}\)  & 
\begin{lstlisting}[language=Syren,numbers=none,basicstyle=\small\linespread{.9}\ttfamily]
LAMBDA phi lambda br, d.
let x = A(d)
let c = phi(d, x)
if c { let y = B(c2) }
\end{lstlisting}

&
{\small where \({\color{red}t_4}(\sigma) = \sigma[(c,\tau) \mapsto \valuation{\dummybranch = 1}{\sigma}{\tau}] \)}
\end{tabular}
\vspace{-1.5em}
\end{center}

The trace valuation for program \(P_4\) will map \dsl{d} to the correct boolean value in each trace, that is, \(\valuation{c}{\sigma}{\tau_1} = \dsl{true}\) and \(\valuation{c}{\sigma}{\tau_2}=\dsl{false}\). Additionally, the values for the inputs of \(c\) and \(x\) will also be known from the state, for example for trace 1 \(\valuation{d}{\sigma}{\tau_1} = \dsl{c1}\) and \(\valuation{x}{\sigma}{\tau_1} = o_1\) (see trace \(\tau_1\)).



\end{example}

\begin{example}
\color{cameraready}
The rewrite rules that introduce retry loops are synthesis rules because the condition on which to stop the loop needs to be synthesized. Syntactically, the rewrite identifies a sequence of statements, possibly with conditionals, and rolls them into a loop. The following is an example of a loop introduction rewrite:

\vspace{-1em}
\begin{center}
\begin{tabular}{m{0.26\textwidth} m{.05\textwidth} m{0.32\textwidth}}
\begin{lstlisting}[language=Syren,numbers=none,basicstyle=\small\linespread{.9}\color{black}\ttfamily]
lambda br, |$\overline{y}$|.
let a = A(c1)
let b1 = B(c2)
let b2 = B(c2)
if br=1 { let b3 = B(c2) }
\end{lstlisting} 
& \(\rightsquigarrow_t\)  
& \begin{lstlisting}[language=Syren,numbers=none,basicstyle=\small\linespread{.9}\color{black}\ttfamily] 
LAMBDA phi lambda br, |$\overline{y}$|.
let a = A(c1)
retry {
  let b = B(c2)
  let s = phi(b, a, |$\overline{y}$|)
} until s
\end{lstlisting}
\end{tabular}
\vspace{-.9em}

{\small \raggedleft where $t(\sigma) = \sigma[
(\dsl{(b,b,b)}, \tau) \mapsto \valuation{\dsl{(b1,b2,b3)}}{\sigma}{\tau}
] \cup [
(\dsl{(s,s,s)}, \tau) \mapsto \dsl{(false,}\valuation{\dsl{br != 1}}{\sigma}{\tau}\dsl{,true)}
]$
}
\end{center}
\vspace{-.2em}

\noindent

In the rewritten program syntax, a new variable \dsl{s} is bound to the result of the hidden function \dsl{phi} and used as a stopping condition for the retry loop. 
The key in ensuring this is a correct rewrite is in the valuation transformation $t$. 
The new trace valuation maps \emph{iterations} of \dsl{b} to the values of each statement that has been captured in the loop, represented by the vector \dsl{(b1,b2,b3)}. When evaluating the program for a trace, the variable \dsl{b3} will not be defined for the traces where \dummybranch \(\ne\) 1, in which case the value is null. The valuation of condition \dsl{s} is also a vector \dsl{(s,s,s)} that is computed by assigning the truth value of whether the statement in the trace should be the last one; at the end of the second iteration, \dsl{s} is \dsl{true} for the trace where \dsl{br!=1}. 

\end{example}

\subsection{Synthesizing Hidden Functions}\label{sec:data-transform-synthesis}

\begin{wrapfigure}{r}{0.52\textwidth}
\vspace{-1.8em}
\centering
\small
\begin{grammar}
  \firstcase{B}{\nonterm{J} == \nonterm{V}}{string or integer equality}
    \otherform{empty(\nonterm{J})}{emptiness check}
    \otherform{! \nonterm{B}}{negation}
  \firstcase{J}{\$}{input}
    \otherform{\nonterm{J} .\nonterm{K}}{select child by name}
    \otherform{\nonterm{J} ..\nonterm{K}}{select descendants by name}
    \otherform{\nonterm{J}[\nonterm{I}]}{select by index}
    \otherform{\nonterm{J}[\nonterm{I}:\nonterm{I}]}{slice by index}
    \otherform{length(\nonterm{J})}{length}
    \otherform{\nonterm{V} + \nonterm{J}}{numerical addition}
    \otherform{\nonterm{V} \bullet \nonterm{J}}{/ string concatenation}
    \firstcase{K}{k \in \textit{keys}}{}
    \firstcase{V}{v \in \textit{values}}{}
    \firstcase{I}{i \in \textit{indices}}{}
\end{grammar}
\caption{Hidden functions synthesis DSL.}
\label{fig:data-transform-synthesis-dsl}
\end{wrapfigure}

The correctness of the result of applying a synthesis rewrite rule \( P \rightsquigarrow P'(f) \) depends on satisfying a set of constraints imposed on \( f \) by the condition  \(\exists f. P'(f) \sqsupseteq P_{in}\). 
As we explained in the previous section, all rewrite rules update an extended state that keeps track of the valuations of the variables in a correct program.
Given the value in the extended state, the synthesis of \( f\) is reducible to a standard programming-by-example (PBE) synthesis problem, deducible from \(\sigma\) only. 
Those constraints are solved with an off-the-shelf synthesizer to produce either an implementation for \(f\) or an unsatisfiability result. In the latter case, the synthesis rule cannot be applied while maintaining correctness.

\paragraph{Generating input/output examples} 
Synthesis rewrite rules replace an expression \(e\) in the program with a function call \(\phi(\overline{x})\) whose result is bound to a new variable \dsl{y}. Before the rewrite, for each trace \( \tau \) we had some value for \(e\), i.e., \(\valuation{e}{\sigma}{\tau} = v_\tau\). To maintain trace subsumption, the transformation~\(t\) ensures \(\valuation{\dsl{y}}{t(\sigma)}{\tau} = v_\tau\) by mapping the new variable \(\dsl{y}\) to the appropriate value. A correct implementation for \(\phi\) must satisfy for each trace \(\tau\) the input-output constraint \( \phi(\valuation{\overline{x}}{\sigma}{\tau}) = v_\tau\).

\begin{example}
Recall Example~\ref{ex:trace-state}. Program \(P_4\) is parametric on \(\phi\), which appears in the statement 
\lstinline[language=Syren]^let c = phi(d,x)^, 
and is correct for a specific implementation of \(\phi\) iff for all traces \(\tau\), \(\phi(\valuation{d}{\sigma}{\tau},\valuation{x}{\sigma}{\tau}) = \valuation{c}{\sigma}{\tau}\). Since we have two traces we have the constraints
\(\phi(\dsl{c1}, o_1) =\dsl{true}\) and \( \phi(\dsl{c3}, o_3) =\dsl{false}\).
\end{example}

\paragraph{Synthesizing solutions} The input-output pairs for each hidden function are used to synthesize the expression for that hidden function. To achieve this, we encode the problems of our synthesis domain into an existing example-based program synthesizer. The hidden functions synthesis can be done in any domain, as long as it is supported by the example-based synthesizer.
In this section, we illustrate using the domain of our running example: cloud automation scripts. The visible functions are typically APIs that accept parameters and return responses in JSON format \cite{bray2014rfc}. Thus, our domain targets JSON data manipulation scripts, including small predicates for generating conditions. This domain covers the majority of use cases of hidden functions between visible function calls in the automation scripts we observed. 
JSON is a lightweight, language-agnostic data interchange format widely used in web applications and APIs. 
The main (recursive) datatypes are lists and dictionaries, which map unique string keys to other JSON objects. The base datatypes are booleans, strings and numbers. Our synthesis domain is summarized in \autoref{fig:data-transform-synthesis-dsl}, which presents a grammar that includes basic comparison between objects and values, and JSONPath~\cite{jsonpath} operations. The non-terminal \(\mathcal{B}\) in the grammar symbolizes the boolean expressions we consider in our DSL, and \(\mathcal{J}\) the JSONPath expressions. Those operations allow the selection of specific indices, members, or descendants of JSON data structures. For example, the path \lstinline^$.element[0]^ selects the \lstinline^element^ field of the object, and then the first element in that list. 

To the best of our knowledge there is no synthesizer that targets this domain, despite the ubiquity of JSON to represent data in applications. Solving it requires encoding our problem into a domain supported by a general purpose synthesizer that allows specifications using input-output examples.
In \systemname, we encode the JSONPath synthesis problem grammar into Rosette~\cite{rosette}, a solver-aided programming language with synthesis constructs. Rosette does not support symbolic strings, thus in our encoding, all strings are constant values extracted from the input and output examples.
The string values are used for \textit{keys} and \textit{values} in the grammar in \autoref{fig:data-transform-synthesis-dsl}, and are the result of enumerating all keys in the dictionaries in input-output constraints, and all values, respectively. In some problems, the size of this set of constants becomes a bottleneck because the objects returned by API calls contain hundreds of keys and values. 
We parallelize the search for a solution by producing sub-grammars for the problem, using different sets of keys and values, and different grammar sizes \cite{abagnale, cubes}.
Rosette was able to synthesize most Boolean expressions in our benchmarks in the grammar including a subset of JSONPath and string operations, shown in \autoref{fig:data-transform-synthesis-dsl}.

We considered an alternative approach to using Rosette: encoding the synthesis problem into SMT theories, and using a SyGuS solver supporting those theories. The SyGuS language~\cite{sygus-lang} allows users to specify synthesis problems with input/output pairs as specifications. We encoded JSON data structures and JSONPath operations using a combination of list and user-defined datatypes for dictionaries and lists, and string and integer theories for the base types. We tested this encoding on CVC5~\cite{cvc5} alongside Rosette~\cite{rosette} on our benchmarks, and found that Rosette consistently outperforms CVC5. CVC5 was unable to solve the problems in our JSONPath benchmarks in reasonable time. We also experimented with PBE problems in the domain of arithmetic operations, and CVC5 and the SyGuS encoding outperformed the grammar defined in Rosette. We conclude that the performance of the PBE solver to which \systemname offloads the hidden function synthesis depends very highly on the domain. \systemname is agnostic to it, and we provide support for using either CVC5 or Rosette, as well as for parallel portfolio solving. 
\section{Rewrite Strategies}\label{sec:synthesis}

The synthesis algorithm is a rewrite process that starts with an initial program \( P_{in} \) that trivially satisfies the correctness criterion \( \Psi \), but is not likely to minimize the cost function \( \chi \). The goal is to transform the initial program by applying refinement and synthesis rules until a program minimizing \( \chi \) is found. Naturally, a naive solution would be to enumerate all possible ways of rewriting \( P_{in} \). However, depending on the program and the set of rewrites available, there may not be a finite set of programs. We consider different strategies to explore the search space of all rewrites efficiently, with the goal of optimizing for \( \chi \).

\begin{wrapfigure}{r}{0.55\textwidth}
    \centering
    \includegraphics[scale=0.7]{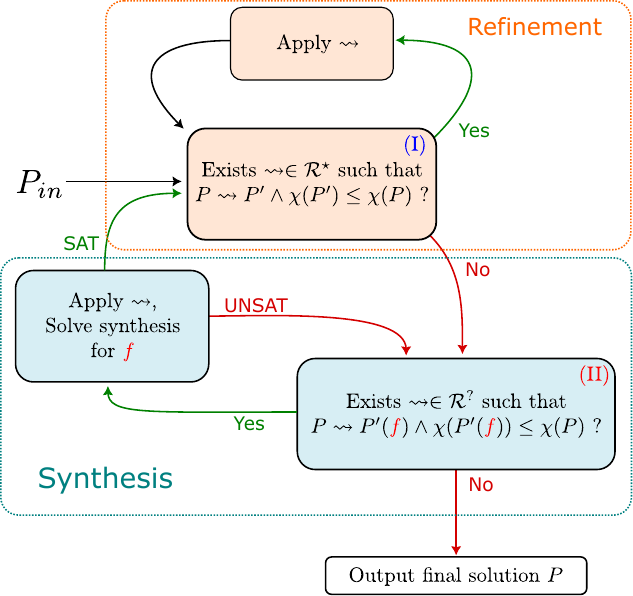}
    \caption{Cost-directed alternating rewrite rule application. }
    \label{fig:cost-directed-search-algo}
    \vspace{1em}
\end{wrapfigure}

\paragraph{Initial Program}
We start with a trivially correct program \( P_{in} \) that provably generates all input traces \( T_{in} \). This program is constructed by introducing a single parameter \dummybranch and a program body that consists of \( |T_{in}| \) branches. Each branch is guarded by a condition \( \dummybranch\ \dsl{==}\ i\),  with \( 0 \leq i < |T_{in}| \). The statements in the then-branch are the API calls of trace \( \tau_i \), written as API call bindings to fresh local variables. The else-branch contains the other branches. \autoref{lst:first-synthesis-program} shows the constructed program.

\begin{figure}
\vspace{-1em}
\begin{lstlisting}[language=Syren, escapeinside=||, numbers=none, basicstyle=\small\linespread{.9}\ttfamily]
lambda br
if br == 1 {
    let y_1 = |\(A^1_1(x^1_1)\)|  let y_2 = |\(A^1_2(x^1_2)\)|  ...  let y_|\(N_1\)| = |\(A^1_{N_1} (x^1_{N_1})\)| (** Replays trace |$\tau_1$| *)
} else if br == 2 {
    let y_1 = |\(A^2_1(x^2_1)\)|  let y_2 = |\(A^2_2(x^2_2)\)|  ...  let y_|\(N_2\)| = |\(A^2_{N_2}(x^2_{N_2})\)| (** Replays trace |$\tau_2$| *)
} else if br == 3 {
    let y_1= |\(A^3_1(x^3_1)\)|  let y_2 = |\(A^3_2(x^3_2)\)|  ...  let y_|\(N_3\)| = |\(A^3_{N_3}(x^3_{N_3})\)|  (** Replays trace |$\tau_3$| *)
} else ...
\end{lstlisting}
\vspace{-1.1em}
\caption{The initial program \( P_{in} \) takes a single integer parameter \(\dsldummybranch\), and has \( |T_{in}|\) branches, where each branch \(i\) simply replays the visible function calls in trace \(\tau_i \in T_{in}\).
\(A_i^q\) is the \(i\)-th function call in trace \#\(q\), and \(x_i^q\) its corresponding~input. 
}
\label{lst:first-synthesis-program}
\Description{A figure showing how the initial program is constructed. The logic of the program is described in the caption.}
\vspace{-.6em}
\end{figure}

In Section~\ref{sec:rewrite}, we distinguish two types of rewrite rules: \emph{refinement rules} \( \mathcal{R}^\star \), which simply rewrite the program maintaining trace subsumption, and \emph{synthesis rules} \( \mathcal{R}^?\) which introduce a data transform synthesis constraint and are correct by construction of the solution of these constraints. Intuitively, refinement rules use less memory and computation, whereas synthesis rules should be applied more carefully. For scalability, one should use refinement rules as much as possible until applying synthesis rules is necessary. 

\paragraph{Alternating Refinement and Synthesis} Figure~\ref{fig:cost-directed-search-algo} illustrates our main algorithm, which alternates between refinement rule and synthesis rule applications until no rewrite rule is applicable. We start with the initial program \( P_{in}\) and look for refinement rules (in \( \mathcal{R}^\star \)) to apply in a way that reduces the cost of the program (step {\color{blue} (I)}). The rule that yields the lowest cost is selected first. If such a rule can be found, we apply it to the current program. We repeat these two steps until no refinement rule can be found. In that case, the algorithm moves inside the bottom loop. It searches for a synthesis rule in \( \mathcal{R}^? \) that reduces the most the cost of the program (step {\color{red}(II)}). If no rule can be found, the synthesis terminates with the current program. Otherwise, the algorithm attempts to apply the rewrite rule \( \exists f \cdot P \rightsquigarrow P'(f)\) and solve for \( f\) using the synthesis process described in Section~\ref{sec:data-transform-synthesis}. There are two possible answers: either a solution is found (\(\mathrm{SAT}\)) or the synthesizer returned \(\mathrm{UNSAT}\). In the first case, the algorithm returns to the upper loop and repeats the entire process. In the other case, the algorithm backtracks on the synthesis rewrite and attempts to find another synthesis rule to apply. When no synthesis rules apply, the algorithm returns the final~program.

This algorithm applies synthesis rules parsimoniously compared to refinement rules. A synthesis problem is solved only when no refinement rules can further lower the program cost, and as soon as a synthesis rule is applied, the algorithm attempts to use more refinement rules.

\subsection{Baselines}\label{sec:baselines}

We give a brief overview of baseline algorithms we implemented as a basis for testing our hypothesis, starting with the observation that motivates them.

\paragraph{Refine-then-Synthesize} Experimentally, we observe that, in many cases, it may be sufficient to apply rules in \emph{only three phases}. First, apply all possible refinement rules to simplify the trivial program \( P_{in} \). Then, we apply every possible synthesis rule that reduces the program's cost. Finally, a final round of refinement rewriting is necessary to clean up the program with the new data transformations. The intuition is that refinement rules operate mostly on the control flow of the program, while synthesis rules operate on the data flow, and interaction between the two is~minimal. 

With that insight in mind, the \emph{refine-then-synthesize} algorithm (denoted \textsc{RTS}) first applies refinement rules, reducing the cost of the program until no refinement rule is applicable, and then synthesis rules until no rule can be found, and finally another round of refinement rules. In other words, it is a modification of the algorithm in Figure~\ref{fig:cost-directed-search-algo} where the \( \mathrm{SAT} \) arrow instead points to step {\color{red}(II)} and updates the program, and the \textcolor{red}{No} arrow returns to refinement for one round.

\paragraph{\(k\)-Bounded Search} One problem of the two previous algorithms is that they \emph{may get stuck on a local minimum of the cost function}. A completely different approach that does not have this problem is a bounded exhaustive search starting from \(P_{in}\). In the \(k\)-bounded search algorithm (denoted by \(k\)-search), rewrite rules from both the refinement set \(\mathcal{R}^\star\) and synthesis set \(\mathcal{R}^?\) are applied, independently of their effect on program cost. Rules are applied again to the resulting programs until all programs resulting from applying \(k\) rules (where \(k\) is a constant) are obtained.

Once all possible rewrites are enumerated, the algorithm ranks all rewrites by increasing cost and attempts to find the program with the lowest cost whose underlying synthesis constraints are satisfiable. Note that in this version of the algorithm, we do not solve the synthesis problem when a synthesis rule is applied. The enumeration is done without a single call to the synthesis solver, which is used only for the programs with low scores. 
\section{Evaluation}\label{sec:evaluation}

We implemented our synthesis approach in a tool, \systemname{}, and evaluated the different algorithms and cost functions against a set of benchmarks, showing a promising approach for synthesizing real-world API composing functions.
Since no existing tool can solve the problem out of the box, we compare against the baseline algorithms introduced in \autoref{sec:synthesis}: \emph{refine-then-synthesize} (RTS) and \(k\)-bounded search (\(k\)-search). Although comparison against a monolithic syntax-guided synthesis approach may be possible (e.g., encoding the problem in Rosette), the limitations in the scalability of Rosette to solve even only the subproblems indicate that it would not scale to the entire problem.

\subsection{Implementation}

All experiments were run on a 2022 Macbook Pro with an M1Pro processor (10 physical cores) and 32GB memory.
\systemname{} is implemented in OCaml and Python 3.12 and uses the Rosette~\cite{rosette} solver-aided language (version 4.1 running on Racket version 8.11) with its default solver Z3 (version 4.12) \cite{z3-solver} to synthesize data transformations. The implementation uses a synthesis constraint cache to avoid repeated calls to the solver with the same constraints. This is especially useful since the algorithm will attempt many synthesis rewrites that will not have a solution.

\subsection{Benchmarks}
We test \systemname{} on a set of \numTotalBenchmarks{} benchmarks that implement various tasks that require branching and looping (in the form of retries) using various APIs. The full description of each task is available in 
\iffullversion
\autoref{sec:detailed-benchmarks}.
\fi
We grouped our benchmarks into four different categories to indicate their origin. The first category is a set of \emph{custom benchmarks} that we wrote to perform some tasks on cloud infrastructure, some shell scripts, and SVG manipulation scripts. We then collected tasks from the Blink automation library \cite{blink-automation}, where various APIs are interfaced. A similar set of tasks comes from AWS Systems Manager Automation Runbooks \cite{aws-runbooks}. Our final category consists of tasks adapted from previous literature; we adapt the nested loop-free benchmarks from \textsc{ApiPhany} \cite{apiphany} that use Stripe and Slack APIs\footnote{We collect benchmarks from \textsc{ApiPhany} \cite{apiphany} by simulating traces from their solutions that do not have list comprehensions. However, we cannot make a direct comparison since our specifications are traces and theirs are types.}. \textcolor{cameraready}{
In general, in the benchmarks used to evaluate Syren, there is a clear separation between the visible function calls (cloud API calls, system calls or library calls) and the local operations (which can be encoded into some solver's theory). When constructing the set of traces, two important parameters have to be considered: whether the different sequences of visible calls exemplify the desired program's control flow paths, and whether the various input values to the calls are sufficient to infer the hidden function's implementation in a given domain (e.g. JSON transformations requires fewer examples than arithmetic).
}

\textcolor{cameraready}{
We wrote programs for each benchmark and collected the inputs for the synthesizer by simulating the traces those programs would produce.
We ran each program for enough different inputs that the produced traces exercise all program paths; we manually inspected the synthesized program and added more traces when it did not exemplify all the behaviors of the target benchmark program. We collected between 2 and 10 traces (median 4) for each benchmark. This is not the smallest number of traces necessary to describe the task, but a reasonable amount that the user could provide. 
}

Our benchmarks include synthesis tasks of varying complexity so we can gauge how well \systemname{} scales. Although we cannot predict how complex a given task is to synthesize, we can estimate it by the complexity of the smallest program that performs that task. We do so by considering the number of conditionals, loops, and hidden functions in the program.

\subsection{Cost Functions}
We ran experiments with two different cost functions to evaluate the flexibility of our approach with respect to different user-defined notions of "best program." We follow the general idea that good programs are simple programs that generalize well. The rewrite rules, especially the refinement rules, are generally geared towards syntactic simplification of the program.
The cost functions match this high-level goal and generally assign a lower cost to simpler programs; the exact meaning of simplicity depends on the user.

\paragraph{Syntactic Complexity} The first method we use, denoted by \(\scoreFuncOne\), is a straightforward cost function that describes the syntactic complexity of the program. This intuitively corresponds to a user who desires a program that is syntactically as simple as possible. This function computes a weighted sum of the number of conditionals, loops, and parameters in the program and a penalty for using the dummy branching variable \(\dummybranch{}\) introduced in the initial program. We use a simplified version of the function in the running example of \autoref{sec:motivating-example}. This function can easily be customized, for example, by modifying the weights of each characteristic in the summation. Typically, we prioritize fewer statements, then fewer parameters, and set the weights accordingly. The penalty for the \(\dummybranch{}\) variable ensures that the algorithm will prioritize eliminating this parameter over all else.

\paragraph{Reuse Across Traces} Our second cost function, denoted by \( \scoreFuncTwo \), measures how many times API call statements are reused with respect to the input set of traces. For each trace in \( T_{in} \), we count how many times each API call statement in the program must be called to produce the trace. The cost is the total number of API calls in the traces plus the number of statements minus the sum of the counts for all statements and all traces. Intuitively, a program with a lower cost means that the API call statements are reused more often; for example, unrolling a loop would increase the cost. We also add a penalty for using the \(\dummybranch{}\) variable.
This is another measure of the program's simplicity. This cost function is coarser in that it assigns the same cost to many different programs.

\subsection{Results}

\begin{figure}
 \centering
 \begin{subfigure}[t]{0.37\linewidth}
     \centering
     \includegraphics[width=.95\linewidth]{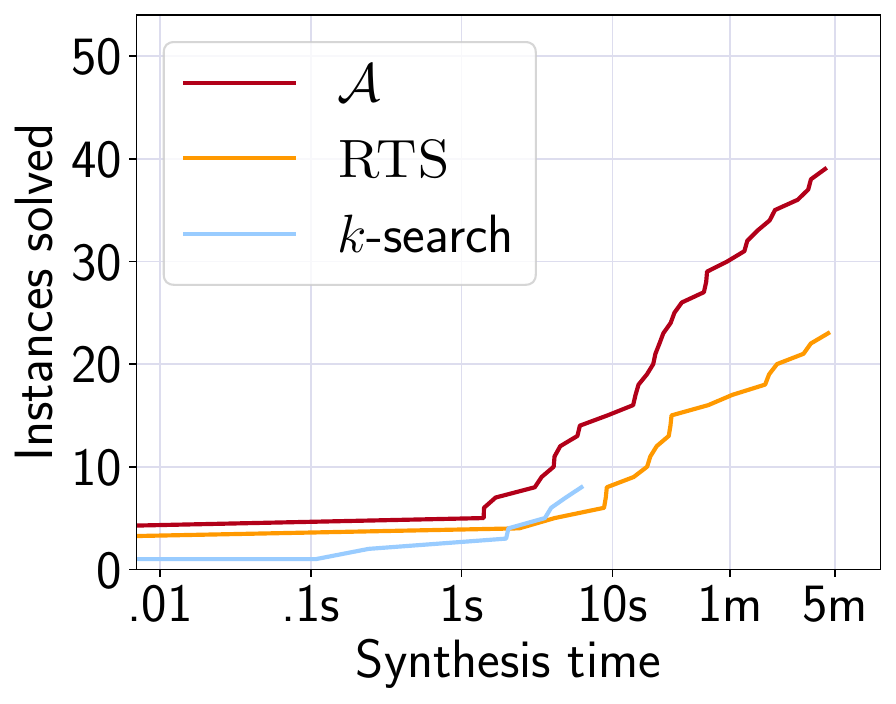}
     \caption{\footnotesize\color{cameraready} 
     Number of benchmarks solved (optimal solution synthesized) for each algorithm of \systemname{} with score function \scoreFuncOne{} against synthesis time.
     }
     \label{fig:solved-per-algo}     
 \end{subfigure}
 \hfill
 \begin{subfigure}[t]{0.27\textwidth}
     \centering
     \includegraphics[scale=.29]{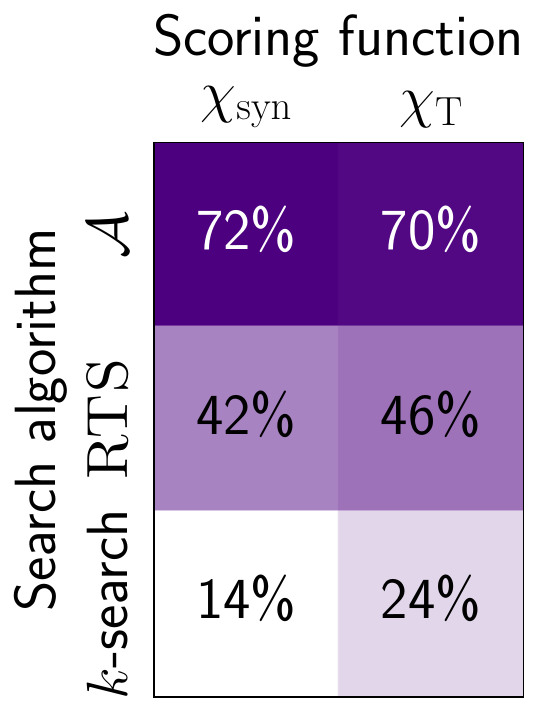}
     \caption{\footnotesize\color{cameraready} Percentage of benchmarks for which \systemname synthesized the optimal solution by search algorithm and scoring function.}
     \label{fig:percentage-smallest}
 \end{subfigure}
 \hfill
 \begin{subfigure}[t]{0.32\textwidth}
     \centering
     \includegraphics[scale=.3]{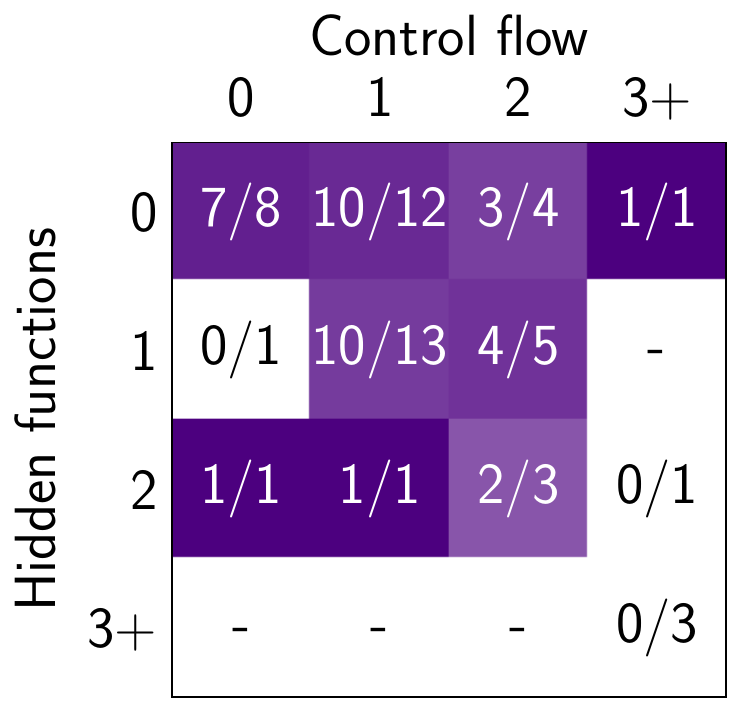}
     \caption{\footnotesize\color{cameraready} Count of optimally synthesized~/~total benchmarks, grouped by complexity measures (\# hidden functions, \# control flow statements) of the program.}
     \label{fig:solved-complexity-amt}
 \end{subfigure}
 \vspace{-.8em}
\caption{\color{cameraready} Comparison of synthesis times and quality of synthesized programs using different search algorithms and cost functions. The background color in the heat map in \ref{fig:percentage-smallest} reflects the same data as the labels. In \ref{fig:solved-complexity-amt}, we show \systemname's ability to scale to complex benchmarks. The darker the color the the better \systemname preformed for benchmarks of that complexity.
}
\vspace{-1em}
\label{fig:solved-instances-times}
\Description{Three plots showing different measurements of Syren's performance. The first plot shows a comparison of the number of benchmarks solved over time using each of the described search strategies: Alternating, RTS, and k-search. The plot shows that the default strategy, alternating, is able to solve more benchmarks per time than any of the other search strategies. The second shows the total percentage of instances each search algorithm was able to solve, with either of the cost functions. The results are the following. With the syntactic complexity cost function, the alternating search solves 72\% of all benchmarks, RTS solves 42\%, and k-search solves 14\%. With the number of traces cost function, the alternating search solves 70\% of all instances, RTS solves 46\%, and k-search solves 24\%. The third plot shows the number of benchmarks solved for different measures of complexity of the desired program. The benchmarks are split by number of hidden functions (0, 1, 2, or 3+) and by number of control flow statements (0, 1, 2, 3+), which can be conditional branching (i.e., if-then-else statements) or loops. For benchmarks with no hidden functions, Syren solves 7 out of 8 benchmarks with no control flow structures, 10 out of 12 for a single control flow structure, 3 out of 4 with 2 control flow structures, and the single benchmark with 3+ control flow structures.
For benchmarks with 1 hidden function, Syren does not solve the single benchmark with no control flow, it solves 10 out of 13 with 1 control flow structure, 4 out of 5 with 2 control flow structures, and there are no benchmarks for 3+ control flow.
At 2 hidden functions, Syren solves the single benchmark with no control flow, the single benchmark with 1 control flow structure, 2 out of 3 with 2 control flow structures, and it does not solve the benchmark with more than 3 control flow structures.
At 3+ hidden functions, there are no benchmarks for 0, 1, or 2 control flow, and Syren solves 0 none of 3 benchmarks with 3+ control flow structures.}
\end{figure}

\textcolor{cameraready}{
The goal is to synthesize a program that is equivalent to the one used to collect traces from; when we report synthesis success, this means the synthesized program is syntactically equivalent to the desired program (modulo variable renaming). Since \systemname{} may terminate with a correct solution that is not optimal (i.e. has a larger cost than our desired program), we also report when the tool terminated but did not solve the benchmark.
}

\autoref{fig:solved-instances-times} plots the synthesis time required for each benchmark and combination of cost function (\(\scoreFuncOne\) or \(\scoreFuncTwo\)) and algorithms (\(\mathcal{A}\) for our algorithm, RTS and \(k\)-bounded for the baselines). Each experiment runs 10 times with a timeout of 10 minutes. We use a fixed \(k = 6\). 

Our proposed algorithm \(\mathcal{A}\) finds the most solutions across benchmarks and cost functions (39 or 72\% \textcolor{cameraready}{optimally synthesized} for syntactic cost, and 38 or 70\% for trace reuse cost out of \numTotalBenchmarks{} benchmarks). The simplified algorithm, RTS, produces fewer \textcolor{cameraready}{optimal} solutions, 23 (42\%) for \(\chi_{\text{syn}}\) and 25 (46\%) for \(\chi_T\)); in many cases the solution produced is not optimal because the algorithm did not attempt enough rewrites. 
This is especially the case for the more complex benchmarks containing loops. We found that the choice of scoring function has little impact on solving time. However, \(\chi_T\) is more coarse in the sense that more programs may have the same score, and when manually inspecting solutions, we found that they are usually farther from the ideal solution. \( \chi_{\text{syn}}\) is better at characterizing the ideal solution.
The solutions obtained required from 4 to 36 rewrites.
The tables in the 
\iffullversion
\autoref{app:tables}
\fi
list detailed synthesis times for all benchmarks.

The bounded search performs poorly across the benchmarks (13 optimally solved for \(\chi_T\)), either yielding a poor solution (because of a small \(k\)) or timing out. The size of the search space for the bounded search is a combination of both the number of rewrite rules and the complexity of the initial program. In general, we observed that scaling \(k\) to the same number used to find a solution using the other algorithms would produce an intractable search space. Interestingly, we observe many timeouts for the \(\chi_T\) cost function; there are many programs with the same cost, but most of them have unrealizable synthesis subproblems.

Even when the best synthesizable solution does not require a data transformation, the rewrite process might still attempt to find some. Synthesis time is dominated by the number of synthesis rewrites attempted rather than the number that ends up being used, so time often does not correspond to the total number of rewrites. For example, synthesizing a solution for the CreateTable benchmark takes 78s, despite not requiring any data transformation. However, since there are two API calls with many arguments, the algorithm must ensure, by attempting to synthesize data transformations, that none of the arguments of the second API call can be computed from the results of the first API~call.

\colorlet{saved}{.}
\color{cameraready}
\paragraph{Comparison against Large Language Models} Large Language Models (LLMs) allow users to generate code from natural language specifications. One can speculate whether specifications could also be given as a list of traces, as in our problem.  To compare \systemname{} against LLMs, we queried Claude 3.5 Sonnet \cite{claude35sonnet} with a prompt explaining the problem to solve, followed by the same traces used by \systemname, in JSON format (see the full prompt in
\iffullversion
\autoref{app:prompt}).
\fi
The output is expected to be a Python script that should be correct (in the sense of \autoref{eq:correctness}) and semantically close to the ideal program. With that success criterion in mind, the LLM synthesizes a correct and optimal solution for 29 benchmarks (53\% of total). 

Note that the evaluation of correctness and equivalence is manual work on our part. This highlights a crucial difference between the two methods. When \systemname{} succeeds and generates a program, the user trusts it is correct, in the sense that it will reproduce the traces. If \systemname{} cannot generate a good program, it will fail or otherwise generate a program that is not general enough, but the output is always safe to use. Specifically, it will never make API calls different from those in the traces, or calls with new inputs. On the other hand, LLMs can fail to synthesize a correct program silently; their output can be subtly incorrect, and verifying this output would require effort (that Syren avoids by generating a program correct by construction). For example, in the task described in \S\ref{sec:motivating-example}, Claude synthesizes the condition \dsl{status == "running"} instead of \dsl{status != "stopped"}, which does not satisfy traces where the instance status is \dsl{"stopping"}.


\color{saved}

\paragraph{Limitations} 

\textcolor{cameraready}{Our evaluation compares \systemname's default search strategy to our own  baselines. To the best of our knowledge, there is no other tool that currently solves this problem, in which the specification is a set of \emph{partial} traces. 
}

Our diverse range of benchmarks includes a few examples where \systemname{} times out or does not return a good solution. Some of those are due to the limited expressive power of the syntax-guided synthesis approach used to create the data transformations. For example, the benchmark ReportLongRunningInstancesToSlack requires reasoning about dates and time intervals, and our solver cannot currently synthesize a solution. The objective program cannot be represented in our synthesis DSL, so we cannot indicate the number of ifs or loops such a program would have. Similarly, the CreateImage benchmark requires reasoning about string operations. \textcolor{cameraready}{Improving the underlying PBE solver to handle more and more complex date, integer, and string operations is an interesting and promising direction of future research. Such improvements would result in an improvement of \systemname{}'s performance without further modification.} The rewrite system also has limitations when considering benchmarks with complex control flow. For example, our tool fails to synthesize the automation AWSSupport-CopyEC2Instance, which has a loop, 5 conditionals, and 16 data transformations to synthesize. Our approach also cannot synthesize any control flow that happens before the first API call in the automation to be synthesized. For example, the benchmark AWS-ConfigureS3BucketVersioning contains empty traces that correspond to situations where the user ran the script, but it terminated before performing any API call. Our approach is unable to infer this behavior, since it cannot observe what the potential inputs to the decision are.

Our current representation of loops limits the programs we can synthesize. We are unable to synthesize programs that loop a fixed number of times and then time out. If \systemname{} is given the traces that result from this pattern, it will instead attempt to discover a condition that is consistent across all the final API calls, which will not exist in general. We omitted the benchmarks from ApiPhany\cite{apiphany} that could not be generated because either they contain nested loops or the loop iterates on parameter of the script. We believe those limitations could be lifted in future work, \textcolor{cameraready}{allowing \systemname to synthesize programs with more complex structures}.







\section{Related Work}

In this section, we provide a more in-depth overview of previous work related to ours. Although we found strong work on problems related to the one \systemname solves, there is no previous solution to solve the same problem. We compare our work against techniques that target similar outputs, i.e. programs with API calls. Then, we look at work that consider similar trace-based specifications, which we can mainly classify in programming-by-demonstration. Finally, we look at work similar in their approach: rewriting techniques and syntax-guided synthesis.

\paragraph{Synthesis with API calls}

\textsc{SyPet}~\cite{sypet}, TYGAR~\cite{tygar}, \textsc{RbSyn}~\cite{rbsyn} and \textsc{ApiPhany}~\cite{apiphany} propose type-guided approaches to synthesis of programs containing sequences of API calls. All three address component-based synthesis, which focuses on finding a composition of components (API calls or library functions) that implements some desired task. These systems differ from \systemname in more than one aspect. 
First,  \textsc{SyPet}, TYGAR, \textsc{RbSyn}, and \textsc{ApiPhany} take as a specification the desired input and output types of a program. \systemname, on the other hand, synthesizes a program using logs of the desired task executed manually.
The contrasts between their approach and ours go beyond the input specification: in their approach, the main challenge \textsc{SyPet} tackles is the large search space of API calls that it has to consider, which results from the ever-increasing expressivity and size of API libraries. TYGAR and \textsc{ApiPhany} propose additional techniques to better represent and understand API calls, thus effectively reducing this search space.
The challenge in our problem is not to discover which API calls need to be made because they are present in the traces, and thus, the API size does not impact our problem. However, we consider more complex interactions between API calls and synthesize data operations that correspond to the non-observable part of the traces.
Conversely, due to the nature of our approach, where we synthesize scripts based on logs, we already know which API calls are made, so the complexity of our synthesis procedure is not affected by the size of the underlying APIs.

DemoMatch \cite{demomatch} discovers code snippets explaining API usage, but in the authors’ own words ``While DemoMatch produces code, it is not a program synthesis tool; it is an API discovery tool.''

\paragraph{Programming-by-demonstration (PBD)}

There are examples of work in PBD that target the automation of web tasks, one of the domains to which we applied \systemname, but there are significant differences in the setting and method that justify our claim of novelty. 
Ringer \cite{ringer} explores ways to represent one recording of a user interacting with a web page and outputs a script that reproduces a single execution of the demonstrated behavior. Our work attempts to generalize multiple executions into a program from logs; the generalization over data required in our setting is not a problem for Ringer.
Approaches such as Konure \cite{konure}, DemoMatch \cite{demomatch}, PUMICE \cite{pumice}, and \textsc{Sketch-n-Sketch} \cite{sketch-n-sketch} are distinct from ours because we synthesize programs only from the logs of the API calls they make, \emph{without} examples of the local operations of the program to be synthesized. 
None of the works cited above attempts to synthesize hidden functions, especially with conditions that depend on the outputs of visible function calls.
Furthermore, some of the approaches rely on interactive demonstrations to generalize their data (WebRobot \cite{webrobot}) or queries (Konure \cite{konure}), or on an existing program to extract dynamic traces from (\textsc{Chisel} \cite{chisel}). Instead, we rely only on a fixed set of examples. Inferring conditionals and loops when a user or algorithm can test the program paths through demonstrations is a different task from inferring them only by optimizing for a cost function.
Like \systemname, WebRobot \cite{webrobot} and Arborist \cite{arborist} use a rewriting strategy to synthesize programs, with speculative rewrites to generalize patterns. However, their method is interactive and relies on the user to validate the heuristic speculations made by the algorithm. In contrast, we rely on data transformation synthesis to validate the applicability of a synthesis rewrite rule (for conditions and loops) and on the cost function to direct the application of rewrite rules. Neither attempts to synthesize decision-making control flow, i.e., if-statements, or loop conditionals (they synthesize for-each loops and while(true) loops, whereas we focus on loops that are stopped by a condition). The synthesis and generalization of these structures is the main challenge in our~approach.
Our work differentiates itself from Konure \cite{konure} in two significant ways. First, Konure is an active learning system that requires querying of the system that needs modeling. Our system is completely passive (closer to PBE than to PBD) and relies only on a fixed set of examples. Second, Konure has full observability over what it generates. For example, the predicates in their DSL are SQL queries, which are observable in the trace. The authors mention that to support more sophisticated implementations in their DSL they suppose that they would need to have a more fine-grained observation of the application. Instead, we propose that fine-grained data operations can be synthesized using a synthesis solver for small expressions. Our paper shows that this difficulty requires carefully exploring the search space of rewrite rules.

\paragraph{Program Rewriting} There is a rich history of work in applying rewriting strategies for program analysis, refactoring, or optimization~\cite{visser2001survey, szalinski}. 
Superoptimization~\cite{superoptimizer-denali, superoptimizer-stochastic, superoptimization-cryptopt} problems, where programs must be rewritten to optimize for a given cost, are classic applications of such techniques. Our application is different in that many of our rewrites are whole-program rewrites that need to consider the state of the program, as opposed to small local rewrites. Note that our approach would not scale when applying many rewrites in an exploratory search: we also need to solve the underlying syntax-guided synthesis problems, which are much more computationally expensive than syntactic rewrite rules. To the best of our knowledge, this work is the first to introduce a system where the application of a rewrite rule is conditional on solving an input-/output-based synthesis problem.
Rewrite approaches have also been used to solve other synthesis problems, such as automatic parallelization in \textsc{Mold}~\cite{translating-mapreduce-rr}. Although they also consider a two-phase approach with refinement and exploration, they use solely fixed rewrite rules instead of our synthesis rules, and the exploratory phase does not have the same expressive power as our synthesis-based solution. 
\textcolor{cameraready}{Szalinski \cite{szalinski} describes an approach that takes a flat, hard-to-read program and introduces map and fold operators using optimizing rewrites. Although its goal is similar to the rewrite part of our approach, Szalinski has no way to synthesize hidden logic in control flow. }

\paragraph{Off-the-shelf SyGuS solver}
The problem \systemname tackles could, in theory, be encoded into Syntax-Guided Synthesis (SyGuS) \cite{sygus} and solved using an off-the-shelf solver, such as CVC5 \cite{cvc5}. We attempted to implement this approach and ran into two problems. First, it is challenging to encode side-effecting functions as uninterpreted functions because the output is not a direct function of the inputs alone. Their output depends on some global state, updated by other functions in the program. More complex encodings can be designed, but they would be a novel contribution on their own, in particular, if they could be made to perform well. With any straightforward encoding, the resulting formula would be too complex to be solved in a reasonable time by CVC5. As explained in \autoref{sec:data-transform-synthesis}, we were unable to use CVC5 for even the JSON synthesis subproblems. 





\section{Conclusion and Future Work}

In this paper we described a novel approach for the synthesis of scripts with calls to side-effecting methods, such as web API method calls, from partial execution traces containing only those calls. The described approach combines a search through a space of rewrite rules that mutate the program without jeopardizing its correctness, and syntax-guided synthesis, which allows us to fill in the expressions that we cannot observed in the traces. We implement our approach in \systemname, and test it on \numTotalBenchmarks{} benchmarks that combine hand-written examples based on common visual interface tasks, scripts from publicly available libraries, and benchmarks from previous work. We show that our approach can successfully synthesize 39/\numTotalBenchmarks{} scripts in under 5 minutes.

The main bottleneck in \systemname's results is in the data-transformations synthesis, so performance improvements should focus on syntax-based synthesis within the domain of \textit{JSONPath}. Another improvement would be writing more specific grammars for the data transforms synthesis, by leveraging additional information about the types of the inputs and outputs, or even their semantic types, as suggested by \textsc{ApiPhany} \cite{apiphany}. Another interesting orthogonal line of work would be to allow traces which show errors on the calls, and synthesize scripts with error handling. \systemname could also be extended to cover more looping behavior, particularly timeouts.

\textcolor{cameraready}{%
We believe the problem \systemname solves is an important and very general problem, yet we had to construct a new set of benchmarks to evaluate our approach.
In related work, the benchmarks included either types, synthesized programs or complete traces, where a clear separation between API-level and local operations was not clear.
 Our benchmark set could be extended with additional use cases where that separation is clear, and the collection of traces is plausible because the system collecting the traces is external to the hidden program. For example, one can collect traces of library calls, system calls or various cloud API calls. The quantity and quality required to synthesize a given program will depend on the control-flow complexity of that program, and the ability of an underlying hidden function synthesizer to synthesize data transformations from few examples.
}

\section*{Acknowledgments}
We are grateful to Ruben Martins and Ricardo Brancas for their feedback on multiple drafts of this paper. We also thank Ricardo for his very insightful suggestions that improved \systemname's evaluation.
Finally, we thank the anonymous reviewers and shepherd for their comments and guidance. 
This work was partially done during an internship at Amazon. Margarida Ferreira is supported by the Portuguese Foundation for Science and Technology under the Carnegie Mellon Portugal PhD fellowship SFRH/BD/151467/2021.

\section*{Data-Availability Statement}
The source code for \systemname is available on Zenodo \cite{artifact}, along with all the data we used to test it. The artifact contains comprehensive instructions and scripts to reproduce the experiments reported in this paper.

\bibliography{bibliography}

\iffullversion
\newpage
\appendix 
\begin{minipage}{\textwidth}

\section{Benchmarks and Detailed Results}\label{app:tables}

\vspace{-2ex}
\thispagestyle{empty} 
\begin{table}[H]
    \centering
    \footnotesize
    \renewcommand{\arraystretch}{.9}
    \setlength{\tabcolsep}{3pt}
    \caption{
      List of benchmarks, measures of complexity, and outcome of synthesis using \systemname (\scoreFuncOne and $\mathcal{A}$) or an
      LLM.
      \#Traces indicates the number of traces used to generate the program.
      \#If (\#Loop, \#Hidden-\(f\)) report the number of conditionals (resp. loops, hidden functions)
      in the best synthesizable solution. The outcome columns report the synthesis outcome for either
      \systemname (Timeout, Terminated or Optimal) or the LLM (Success or Failure).
    }
    \label{tab:outcomes-table}
    \vspace{-1.3em}
    \begin{tabularx}{\linewidth}{|X|c|c|c|c||c|c|}
    \hline
        \rowcolor{lightgray} \small Name &
        \small \rotatebox{0}{\#Traces} &
        \small \rotatebox{0}{\#If} &
        \small \rotatebox{0}{\#Loop} &
        \small \rotatebox{0}{\#Hidden-\(f\)} &
        \small \systemname outcome &
        \small LLM outcome
    \\

        \rowcolor{lightestgray} \multicolumn{7}{| c|}{\small Custom Benchmarks} \\
         Start Instances & 7 & 0 & 0 & 0 &Optimal & Success \\
        \rowcolor{lblue} Stop Instances & 9 & 0 & 0 & 0 &Optimal & Success \\
         Create Image & 4 & 0 & 0 & 1 &Timeout & Failure \\
        \rowcolor{lblue} Create Image From Ssm Parameter And Log & 3 & 0 & 0 & 4 &Optimal & Success \\
         Create Table & 7 & 0 & 0 & 0 &Optimal & Failure \\
        \rowcolor{lblue} Delete Table & 7 & 0 & 0 & 0 &Optimal & Failure \\
         Create Bucket Then Folder & 3 & 0 & 0 & 1 &Optimal & Success \\
        \rowcolor{lblue} Put Object If Not Present & 6 & 1 & 0 & 1 &Optimal & Failure \\
         Stop Instances Cond & 6 & 1 & 0 & 1 &Optimal & Failure \\
        \rowcolor{lblue} Create Table Insert Item & 8 & 0 & 1 & 1 &Timeout & Success \\
         Backup Then Delete Table & 6 & 0 & 1 & 2 &Optimal & Success \\
        \rowcolor{lblue} Start Instances With Tags & 7 & 0 & 0 & 1 &Optimal & Success \\
         Stop All Running Instances & 4 & 0 & 0 & 1 &Optimal & Success \\
        \rowcolor{lblue} Move Ddb Item & 7 & 0 & 0 & 1 &Optimal & Failure \\
         Move Ddb Item If Present & 8 & 1 & 0 & 2 &Optimal & Success \\
        \rowcolor{lblue} Copy S3Objects & 3 & 1 & 1 & 2 &Optimal & Failure \\
         Tag Instances With Dry Run & 4 & 1 & 0 & 1 &Optimal & Failure \\
        \rowcolor{lblue} Send Email On Input & 3 & 0 & 1 & 1 &Optimal & Success \\
         Retrieve Channel Members & 3 & 0 & 1 & 2 &Optimal & Failure \\
        \rowcolor{lblue} List and move files & 3 & 1 & 1 & 0 &Optimal & Failure \\
         Clean current dir & 2 & 1 & 1 & 1 &Optimal & Success \\
        \rowcolor{lblue} Clean regular files in dir & 2 & 1 & 1 & 2 &Optimal & Success \\
         Conversation members & 4 & 1 & 1 & 3 &Timeout & Failure \\
        \rowcolor{lblue} SVG-Increase Circle Radius-json & 3 & 0 & 0 & 2 &Terminated & Failure \\
         SVG-Increase Circle Radius & 3 & 0 & 0 & 1 &Optimal & Success \\
        \rowcolor{lightestgray} \multicolumn{7}{| c|}{\small Blink Automation} \\
         Copy Instance To New Region & 4 & 0 & 1 & 2 &Timeout & Success \\
        \rowcolor{lblue} Report Long Running Instances & 8 & ?? & ?? & -1 &Timeout & Failure \\
         Create Iam User And Notify & 3 & 0 & 0 & 1 &Optimal & Failure \\
        \rowcolor{lightestgray} \multicolumn{7}{| c|}{\small AWS Automation Runbooks} \\
         Stop EC2Instance & 5 & 0 & 0 & 0 &Optimal & Success \\
        \rowcolor{lblue} Start EC2Instance & 5 & 1 & 0 & 2 &Optimal & Failure \\
         Configure S3Bucket Versioning & 9 & 0 & 0 & 1 &Terminated & Success \\
        \rowcolor{lblue} Configure Cloud Watch On EC2 & 8 & 1 & 0 & 0 &Terminated & Failure \\
         Copy EC2Instance & 3 & 5 & 1 & 16 &Timeout & Failure \\
        \rowcolor{lblue} Resize Instance & 8 & 3 & 2 & 3 &Timeout & Failure \\
         Set Required Tags & 4 & 0 & 0 & 0 &Optimal & Success \\
        \rowcolor{lightestgray} \multicolumn{7}{| c|}{\small Literature benchmarks} \\
         Api Phany Ex. 1.1 & 4 & 1 & 2 & 5 &Timeout & Success \\
        \rowcolor{lblue} Api Phany Ex. 1.2 & 3 & 0 & 0 & 2 &Optimal & Failure \\
         Api Phany Ex. 1.6 & 3 & 0 & 0 & 1 &Optimal & Success \\
        \rowcolor{lblue} Api Phany Ex. 1.8 & 3 & 0 & 0 & 1 &Optimal & Success \\
         Api Phany Ex. 2.1 & 3 & 0 & 1 & 1 &Optimal & Success \\
        \rowcolor{lblue} Api Phany Ex. 2.3 & 3 & 0 & 0 & 2 &Optimal & Success \\
         Api Phany Ex. 2.5 & 3 & 0 & 1 & 1 &Optimal & Success \\
        \rowcolor{lblue} Api Phany Ex. 2.6 & 3 & 0 & 0 & 1 &Optimal & Failure \\
         Api Phany Ex. 2.7 & 4 & 0 & 1 & 1 &Optimal & Success \\
        \rowcolor{lblue} Api Phany Ex. 2.10 & 2 & 0 & 1 & 1 &Timeout & Success \\
         Api Phany Ex. 2.11 & 3 & 0 & 0 & 1 &Optimal & Failure \\
        \rowcolor{lblue} Api Phany Ex. 2.13 & 3 & 0 & 0 & 2 &Optimal & Success \\
         Api Phany Ex. 3.1 & 2 & 0 & 0 & 0 &Optimal & Failure \\
        \rowcolor{lblue} Api Phany Ex. 3.3 & 3 & 0 & 0 & 0 &Terminated & Failure \\
         Api Phany Ex. 3.4* & 2 & 0 & 1 & 1 &Optimal & Failure \\
        \rowcolor{lblue} Api Phany Ex. 3.6 & 3 & 0 & 1 & 1 &Optimal & Success \\
         Api Phany Ex. 3.7 & 2 & 0 & 1 & 1 &Optimal & Success \\
        \rowcolor{lblue} Api Phany Ex. 3.8 & 2 & 0 & 2 & 2 &Timeout & Success \\
         Api Phany Ex. 3.9 & 3 & 0 & 1 & 1 &Timeout & Failure \\
    \hline
    \end{tabularx}
    
    \vspace{-2em}
\end{table}

\end{minipage}

\begin{table}[H]
    \centering
    \footnotesize
    \renewcommand{\arraystretch}{1.2}
    \setlength{\tabcolsep}{3pt}
    \caption{List of Custom benchmarks, with synthesis time and number of SyGuS solver calls.
    For each algorithm (\(\mathcal{A}\), \textsc{RTS} and \(k\)-search)
    we report the synthesis time in seconds, the number of total SyGuS solver calls,
    and the subset of those calls that returned SAT, in parenthesis.
    When \systemname times out (TO), it does not report the number of solver calls (N/A).}
    \label{tab:synthesis-times-custom}
    \begin{tabularx}{\linewidth}{|X| |c|c||c|c||c|c| |c|c||c|c||c|c|}
    \hline
        \rowcolor{lightestgray} \multicolumn{13}{| c|}{\small Custom Benchmarks} \\
    \hline
        \rowcolor{lightgray}
        \ &
        \multicolumn{6}{c|}{ \(\scoreFuncOne\) } &
        \multicolumn{6}{c|}{ \(\scoreFuncTwo\) }
    \\
    \hline
        \rowcolor{lightgray} &

        \multicolumn{2}{c|}{\rotatebox{0}{\(\mathcal{A}\)}} &
        \multicolumn{2}{c|}{\rotatebox{0}{\textsc{RTS}}} &
        \multicolumn{2}{c|}{\rotatebox{0}{\(k\)-search}} &

        \multicolumn{2}{c|}{\rotatebox{0}{\(\mathcal{A}\)}} &
        \multicolumn{2}{c|}{\rotatebox{0}{\textsc{RTS}}} &
        \multicolumn{2}{c|}{\rotatebox{0}{\(k\)-search}}
    \\
    \hline
    \rowcolor{lightgray} Name &

        \rotatebox{90}{\systemname runtime\phantom{x}} &
        \rotatebox{90}{SyGuS calls (sat)\phantom{x}} &

        \rotatebox{90}{\systemname runtime} &
        \rotatebox{90}{SyGuS calls (sat)} &

        \rotatebox{90}{\systemname runtime} &
        \rotatebox{90}{SyGuS calls (sat)} &

        \rotatebox{90}{\systemname runtime} &
        \rotatebox{90}{SyGuS calls (sat)} &

        \rotatebox{90}{\systemname runtime} &
        \rotatebox{90}{SyGuS calls (sat)} &

        \rotatebox{90}{\systemname runtime} &
        \rotatebox{90}{SyGuS calls (sat)}
    \\
    \hline

         Start Instances &\scriptsize \text{<0.01} & \scriptsize  0 (0) & \scriptsize \text{<0.01} & \scriptsize  0 (0) & \scriptsize \textcolor{darkgray}{TO} & \scriptsize \text{\textcolor{darkgray}{N/A}} & \scriptsize \text{<0.01} & \scriptsize  0 (0) & \scriptsize \text{<0.01} & \scriptsize  0 (0) & \scriptsize \textcolor{darkgray}{TO} & \scriptsize \text{\textcolor{darkgray}{N/A}}  \\
        \rowcolor{lblue} Stop Instances &\scriptsize \text{<0.01} & \scriptsize  0 (0) & \scriptsize \text{<0.01} & \scriptsize  0 (0) & \scriptsize \textcolor{darkgray}{TO} & \scriptsize \text{\textcolor{darkgray}{N/A}} & \scriptsize \text{<0.01} & \scriptsize  0 (0) & \scriptsize \text{<0.01} & \scriptsize  0 (0) & \scriptsize \textcolor{darkgray}{TO} & \scriptsize \text{\textcolor{darkgray}{N/A}}  \\
         Create Image &\scriptsize \textcolor{darkgray}{TO} & \scriptsize \text{\textcolor{darkgray}{N/A}} & \scriptsize \textcolor{darkgray}{TO} & \scriptsize \text{\textcolor{darkgray}{N/A}} & \scriptsize \color{darkgray} 0.16 & \scriptsize \textcolor{darkgray}{ 0} (\textcolor{darkgray}{ 0}) & \scriptsize \textcolor{darkgray}{TO} & \scriptsize \text{\textcolor{darkgray}{N/A}} & \scriptsize \textcolor{darkgray}{TO} & \scriptsize \text{\textcolor{darkgray}{N/A}} & \scriptsize \textcolor{darkgray}{TO} & \scriptsize \text{\textcolor{darkgray}{N/A}}  \\
        \rowcolor{lblue} Create Image From Ssm Parameter And Log &\scriptsize 207.02 & \scriptsize  8 (4) & \scriptsize 206.53 & \scriptsize  8 (4) & \scriptsize \color{darkgray} 201.87 & \scriptsize \textcolor{darkgray}{ 0} (\textcolor{darkgray}{ 0}) & \scriptsize 210.21 & \scriptsize  8 (4) & \scriptsize 205.58 & \scriptsize  8 (4) & \scriptsize \textcolor{darkgray}{TO} & \scriptsize \text{\textcolor{darkgray}{N/A}}  \\
         Create Table &\scriptsize 78.33 & \scriptsize  4 (0) & \scriptsize 123.40 & \scriptsize  7 (0) & \scriptsize \color{darkgray} 0.44 & \scriptsize \textcolor{darkgray}{ 0} (\textcolor{darkgray}{ 0}) & \scriptsize 77.97 & \scriptsize  4 (0) & \scriptsize 123.41 & \scriptsize  7 (0) & \scriptsize \textcolor{darkgray}{TO} & \scriptsize \text{\textcolor{darkgray}{N/A}}  \\
        \rowcolor{lblue} Delete Table &\scriptsize \text{<0.01} & \scriptsize  0 (0) & \scriptsize \text{<0.01} & \scriptsize  0 (0) & \scriptsize \textcolor{darkgray}{TO} & \scriptsize \text{\textcolor{darkgray}{N/A}} & \scriptsize \text{<0.01} & \scriptsize  0 (0) & \scriptsize \text{<0.01} & \scriptsize  0 (0) & \scriptsize \textcolor{darkgray}{TO} & \scriptsize \text{\textcolor{darkgray}{N/A}}  \\
         Create Bucket Then Folder &\scriptsize 25.76 & \scriptsize  3 (2) & \scriptsize 4.14 & \scriptsize  2 (3) & \scriptsize 0.11 & \scriptsize  6 (61) & \scriptsize 25.85 & \scriptsize  3 (2) & \scriptsize 4.13 & \scriptsize  2 (3) & \scriptsize 0.06 & \scriptsize  1 (1)  \\
        \rowcolor{lblue} Put Object If Not Present &\scriptsize 256.88 & \scriptsize  5 (1) & \scriptsize 269.27 & \scriptsize  8 (1) & \scriptsize \color{darkgray} 0.74 & \scriptsize \textcolor{darkgray}{ 0} (\textcolor{darkgray}{ 0}) & \scriptsize 13.65 & \scriptsize  3 (1) & \scriptsize 268.51 & \scriptsize  8 (1) & \scriptsize \color{darkgray} 175.23 & \scriptsize \textcolor{darkgray}{ 10} (\textcolor{darkgray}{ 1})  \\
         Stop Instances Cond &\scriptsize 57.80 & \scriptsize  7 (7) & \scriptsize \color{darkgray} 51.84 & \scriptsize \textcolor{darkgray}{ 5} (\textcolor{darkgray}{ 6}) & \scriptsize \color{darkgray} 0.37 & \scriptsize \textcolor{darkgray}{ 0} (\textcolor{darkgray}{ 0}) & \scriptsize 50.71 & \scriptsize  3 (3) & \scriptsize 54.64 & \scriptsize  5 (7) & \scriptsize \color{darkgray} 104.34 & \scriptsize \textcolor{darkgray}{ 7} (\textcolor{darkgray}{ 2})  \\
        \rowcolor{lblue} Create Table Insert Item &\scriptsize \textcolor{darkgray}{TO} & \scriptsize \text{\textcolor{darkgray}{N/A}} & \scriptsize \textcolor{darkgray}{TO} & \scriptsize \text{\textcolor{darkgray}{N/A}} & \scriptsize \color{darkgray} 310.20 & \scriptsize \textcolor{darkgray}{6} (\textcolor{darkgray}{3}) & \scriptsize \textcolor{darkgray}{TO} & \scriptsize \text{\textcolor{darkgray}{N/A}} & \scriptsize \textcolor{darkgray}{TO} & \scriptsize \text{\textcolor{darkgray}{N/A}} & \scriptsize \color{darkgray} 310.29 & \scriptsize \textcolor{darkgray}{ 6} (\textcolor{darkgray}{ 1})  \\
         Backup Then Delete Table &\scriptsize 198.66 & \scriptsize  11 (4) & \scriptsize \color{darkgray} 200.20 & \scriptsize \textcolor{darkgray}{ 9} (\textcolor{darkgray}{ 4}) & \scriptsize \color{darkgray} 2.08 & \scriptsize \textcolor{darkgray}{0} (\textcolor{darkgray}{0}) & \scriptsize \color{darkgray} 196.25 & \scriptsize \textcolor{darkgray}{ 8} (\textcolor{darkgray}{ 3}) & \scriptsize \textcolor{darkgray}{TO} & \scriptsize \text{\textcolor{darkgray}{N/A}} & \scriptsize \color{darkgray} 2.08 & \scriptsize \textcolor{darkgray}{ 0} (\textcolor{darkgray}{ 0})  \\
        \rowcolor{lblue} Start Instances With Tags &\scriptsize 14.22 & \scriptsize  2 (1) & \scriptsize 17.83 & \scriptsize  2 (1) & \scriptsize \color{darkgray} 0.35 & \scriptsize \textcolor{darkgray}{ 1} (\textcolor{darkgray}{ 1}) & \scriptsize 19.32 & \scriptsize  2 (1) & \scriptsize 20.56 & \scriptsize  2 (1) & \scriptsize 0.43 & \scriptsize  1 (0)  \\
         Stop All Running Instances &\scriptsize 24.30 & \scriptsize  1 (1) & \scriptsize 24.24 & \scriptsize  1 (1) & \scriptsize \textcolor{darkgray}{TO} & \scriptsize \text{\textcolor{darkgray}{N/A}} & \scriptsize 16.29 & \scriptsize  1 (1) & \scriptsize 21.72 & \scriptsize  1 (1) & \scriptsize \textcolor{darkgray}{TO} & \scriptsize \text{\textcolor{darkgray}{N/A}}  \\
        \rowcolor{lblue} Move Ddb Item &\scriptsize 119.49 & \scriptsize  6 (1) & \scriptsize 109.31 & \scriptsize  7 (1) & \scriptsize \color{darkgray} 0.38 & \scriptsize \textcolor{darkgray}{ 0} (\textcolor{darkgray}{ 0}) & \scriptsize 89.38 & \scriptsize  6 (1) & \scriptsize 108.40 & \scriptsize  7 (1) & \scriptsize \textcolor{darkgray}{TO} & \scriptsize \text{\textcolor{darkgray}{N/A}}  \\
         Move Ddb Item If Present &\scriptsize 91.50 & \scriptsize  7 (2) & \scriptsize 102.89 & \scriptsize  8 (2) & \scriptsize \color{darkgray} 2.18 & \scriptsize \textcolor{darkgray}{ 0} (\textcolor{darkgray}{ 0}) & \scriptsize 32.67 & \scriptsize  3 (1) & \scriptsize 102.27 & \scriptsize  8 (3) & \scriptsize \textcolor{darkgray}{TO} & \scriptsize \text{\textcolor{darkgray}{N/A}}  \\
        \rowcolor{lblue} Copy S3Objects &\scriptsize 110.38 & \scriptsize  4 (2) & \scriptsize \textcolor{darkgray}{TO} & \scriptsize \text{\textcolor{darkgray}{N/A}} & \scriptsize \color{darkgray} 3.32 & \scriptsize \textcolor{darkgray}{ 2} (\textcolor{darkgray}{ 2}) & \scriptsize 110.17 & \scriptsize  4 (2) & \scriptsize \textcolor{darkgray}{TO} & \scriptsize \text{\textcolor{darkgray}{N/A}} & \scriptsize \color{darkgray} 3.18 & \scriptsize \textcolor{darkgray}{ 2} (\textcolor{darkgray}{ 2})  \\
         Tag Instances With Dry Run &\scriptsize 74.95 & \scriptsize  9 (3) & \scriptsize 62.35 & \scriptsize  8 (3) & \scriptsize \color{darkgray} 0.10 & \scriptsize \textcolor{darkgray}{ 0} (\textcolor{darkgray}{ 0}) & \scriptsize 28.36 & \scriptsize  3 (1) & \scriptsize 62.51 & \scriptsize  8 (4) & \scriptsize \color{darkgray} 61.31 & \scriptsize \textcolor{darkgray}{ 12} (\textcolor{darkgray}{ 263})  \\
        \rowcolor{lblue} Send Email On Input &\scriptsize 5.86 & \scriptsize  3 (2) & \scriptsize \color{darkgray} 8.88 & \scriptsize \textcolor{darkgray}{ 4} (\textcolor{darkgray}{ 2}) & \scriptsize \color{darkgray} 10.14 & \scriptsize \textcolor{darkgray}{ 1} (\textcolor{darkgray}{ 1}) & \scriptsize 5.75 & \scriptsize  3 (2) & \scriptsize \color{darkgray} 8.83 & \scriptsize \textcolor{darkgray}{ 4} (\textcolor{darkgray}{ 2}) & \scriptsize 10.14 & \scriptsize  1 (1)  \\
         Retrieve Channel Members &\scriptsize 6.08 & \scriptsize  2 (2) & \scriptsize \color{darkgray} 0.67 & \scriptsize \textcolor{darkgray}{ 7} (\textcolor{darkgray}{ 4}) & \scriptsize \textcolor{darkgray}{TO} & \scriptsize \text{\textcolor{darkgray}{N/A}} & \scriptsize 6.02 & \scriptsize  2 (2) & \scriptsize \textcolor{darkgray}{TO} & \scriptsize \text{\textcolor{darkgray}{N/A}} & \scriptsize \textcolor{darkgray}{TO} & \scriptsize \text{\textcolor{darkgray}{N/A}}  \\
        \rowcolor{lblue} List and move files &\scriptsize 14.89 & \scriptsize  9 (3) & \scriptsize \color{darkgray} 199.48 & \scriptsize \textcolor{darkgray}{ 11} (\textcolor{darkgray}{ 3}) & \scriptsize \textcolor{darkgray}{TO} & \scriptsize \text{\textcolor{darkgray}{N/A}} & \scriptsize 14.92 & \scriptsize  10 (3) & \scriptsize \color{darkgray} 199.76 & \scriptsize \textcolor{darkgray}{ 11} (\textcolor{darkgray}{ 3}) & \scriptsize \color{darkgray} 6.49 & \scriptsize \textcolor{darkgray}{ 0} (\textcolor{darkgray}{ 0})  \\
         Clean current dir &\scriptsize 1.68 & \scriptsize  1 (1) & \scriptsize \textcolor{darkgray}{TO} & \scriptsize \text{\textcolor{darkgray}{N/A}} & \scriptsize \textcolor{darkgray}{TO} & \scriptsize \text{\textcolor{darkgray}{N/A}} & \scriptsize 1.70 & \scriptsize  1 (1) & \scriptsize \textcolor{darkgray}{TO} & \scriptsize \text{\textcolor{darkgray}{N/A}} & \scriptsize \textcolor{darkgray}{TO} & \scriptsize \text{\textcolor{darkgray}{N/A}}  \\
        \rowcolor{lblue} Clean regular files in dir &\scriptsize 9.27 & \scriptsize  2 (2) & \scriptsize \color{darkgray} 6.53 & \scriptsize \textcolor{darkgray}{ 3} (\textcolor{darkgray}{ 3}) & \scriptsize \textcolor{darkgray}{TO} & \scriptsize \text{\textcolor{darkgray}{N/A}} & \scriptsize 11.32 & \scriptsize  2 (2) & \scriptsize \color{darkgray} 8.60 & \scriptsize \textcolor{darkgray}{ 3} (\textcolor{darkgray}{ 3}) & \scriptsize \textcolor{darkgray}{TO} & \scriptsize \text{\textcolor{darkgray}{N/A}}  \\
         Conversation members &\scriptsize \textcolor{darkgray}{TO} & \scriptsize \text{\textcolor{darkgray}{N/A}} & \scriptsize \textcolor{darkgray}{TO} & \scriptsize \text{\textcolor{darkgray}{N/A}} & \scriptsize \textcolor{darkgray}{TO} & \scriptsize \text{\textcolor{darkgray}{N/A}} & \scriptsize \textcolor{darkgray}{TO} & \scriptsize \text{\textcolor{darkgray}{N/A}} & \scriptsize \textcolor{darkgray}{TO} & \scriptsize \text{\textcolor{darkgray}{N/A}} & \scriptsize \textcolor{darkgray}{TO} & \scriptsize \text{\textcolor{darkgray}{N/A}}  \\
        \rowcolor{lblue} SVG-Increase Circle Radius-json &\scriptsize \color{darkgray} 165.69 & \scriptsize \textcolor{darkgray}{ 6} (\textcolor{darkgray}{ 4}) & \scriptsize \color{darkgray} 207.81 & \scriptsize \textcolor{darkgray}{ 6} (\textcolor{darkgray}{ 4}) & \scriptsize \color{darkgray} 0.09 & \scriptsize \textcolor{darkgray}{ 0} (\textcolor{darkgray}{ 0}) & \scriptsize 165.94 & \scriptsize  6 (4) & \scriptsize 207.65 & \scriptsize  6 (4) & \scriptsize \color{darkgray} 8.33 & \scriptsize \textcolor{darkgray}{ 3} (\textcolor{darkgray}{ 4})  \\
         SVG-Increase Circle Radius &\scriptsize 42.38 & \scriptsize  4 (3) & \scriptsize 2.41 & \scriptsize  4 (4) & \scriptsize \color{darkgray} 0.09 & \scriptsize \textcolor{darkgray}{ 1} (\textcolor{darkgray}{ 1}) & \scriptsize 42.44 & \scriptsize  4 (3) & \scriptsize 2.40 & \scriptsize  4 (4) & \scriptsize 0.09 & \scriptsize  14 (8)  \\
    \hline
    \end{tabularx}
\end{table}

\begin{table}[H]
    \centering
    \footnotesize
    \renewcommand{\arraystretch}{1.2}
    \setlength{\tabcolsep}{3pt}
    \caption{List of Blink Automation benchmarks, with synthesis time and number of SyGuS solver calls.
    For each algorithm (\(\mathcal{A}\), \textsc{RTS} and \(k\)-search)
    we report the synthesis time in seconds, the number of total SyGuS solver calls,
    and the subset of those calls that returned SAT, in parenthesis.
    When \systemname times out (TO), it does not report the number of solver calls (N/A).}
    \label{tab:synthesis-times-blink-automation}
    \begin{tabularx}{\linewidth}{|X| |c|c||c|c||c|c| |c|c||c|c||c|c|}
    \hline
        \rowcolor{lightestgray} \multicolumn{13}{| c|}{\small Blink Automation} \\
    \hline
        \rowcolor{lightgray}
        \ &
        \multicolumn{6}{c|}{ \(\scoreFuncOne\) } &
        \multicolumn{6}{c|}{ \(\scoreFuncTwo\) }
    \\
    \hline
        \rowcolor{lightgray} &

        \multicolumn{2}{c|}{\rotatebox{0}{\(\mathcal{A}\)}} &
        \multicolumn{2}{c|}{\rotatebox{0}{\textsc{RTS}}} &
        \multicolumn{2}{c|}{\rotatebox{0}{\(k\)-search}} &

        \multicolumn{2}{c|}{\rotatebox{0}{\(\mathcal{A}\)}} &
        \multicolumn{2}{c|}{\rotatebox{0}{\textsc{RTS}}} &
        \multicolumn{2}{c|}{\rotatebox{0}{\(k\)-search}}
    \\
    \hline
    \rowcolor{lightgray} Name &

        \rotatebox{90}{\systemname runtime\phantom{x}} &
        \rotatebox{90}{SyGuS calls (sat)\phantom{x}} &

        \rotatebox{90}{\systemname runtime} &
        \rotatebox{90}{SyGuS calls (sat)} &

        \rotatebox{90}{\systemname runtime} &
        \rotatebox{90}{SyGuS calls (sat)} &

        \rotatebox{90}{\systemname runtime} &
        \rotatebox{90}{SyGuS calls (sat)} &

        \rotatebox{90}{\systemname runtime} &
        \rotatebox{90}{SyGuS calls (sat)} &

        \rotatebox{90}{\systemname runtime} &
        \rotatebox{90}{SyGuS calls (sat)}
    \\
    \hline

         Copy Instance To New Region &\scriptsize \textcolor{darkgray}{TO} & \scriptsize \text{\textcolor{darkgray}{N/A}} & \scriptsize \textcolor{darkgray}{TO} & \scriptsize \text{\textcolor{darkgray}{N/A}} & \scriptsize \textcolor{darkgray}{TO} & \scriptsize \text{\textcolor{darkgray}{N/A}} & \scriptsize \textcolor{darkgray}{TO} & \scriptsize \text{\textcolor{darkgray}{N/A}} & \scriptsize \textcolor{darkgray}{TO} & \scriptsize \text{\textcolor{darkgray}{N/A}} & \scriptsize \textcolor{darkgray}{TO} & \scriptsize \text{\textcolor{darkgray}{N/A}}  \\
        \rowcolor{lblue} Report Long Running Instances &\scriptsize \textcolor{darkgray}{TO} & \scriptsize \text{\textcolor{darkgray}{N/A}} & \scriptsize \textcolor{darkgray}{TO} & \scriptsize \text{\textcolor{darkgray}{N/A}} & \scriptsize \textcolor{darkgray}{TO} & \scriptsize \text{\textcolor{darkgray}{N/A}} & \scriptsize \textcolor{darkgray}{TO} & \scriptsize \text{\textcolor{darkgray}{N/A}} & \scriptsize \textcolor{darkgray}{TO} & \scriptsize \text{\textcolor{darkgray}{N/A}} & \scriptsize \textcolor{darkgray}{TO} & \scriptsize \text{\textcolor{darkgray}{N/A}}  \\
         Create Iam User And Notify &\scriptsize 41.83 & \scriptsize  2 (1) & \scriptsize 43.17 & \scriptsize  2 (1) & \scriptsize \color{darkgray} 0.05 & \scriptsize \textcolor{darkgray}{ 2} (\textcolor{darkgray}{ 1}) & \scriptsize 40.75 & \scriptsize  2 (1) & \scriptsize 41.26 & \scriptsize  2 (1) & \scriptsize 0.05 & \scriptsize  2 (2)  \\
    \hline
    \end{tabularx}
\end{table}

\begin{table}[H]
    \centering
    \footnotesize
    \renewcommand{\arraystretch}{1.2}
    \setlength{\tabcolsep}{3pt}
    \caption{List of AWS Automation Runbooks benchmarks, with synthesis time and number of SyGuS solver calls.
    For each algorithm (\(\mathcal{A}\), \textsc{RTS} and \(k\)-search)
    we report the synthesis time in seconds, the number of total SyGuS solver calls,
    and the subset of those calls that returned SAT, in parenthesis.
    When \systemname times out (TO), it does not report the number of solver calls (N/A).}
    \label{tab:synthesis-times-aws-automation-runbooks}
    \begin{tabularx}{\linewidth}{|X| |c|c||c|c||c|c| |c|c||c|c||c|c|}
    \hline
        \rowcolor{lightestgray} \multicolumn{13}{| c|}{\small AWS Automation Runbooks} \\
    \hline
        \rowcolor{lightgray}
        \ &
        \multicolumn{6}{c|}{ \(\scoreFuncOne\) } &
        \multicolumn{6}{c|}{ \(\scoreFuncTwo\) }
    \\
    \hline
        \rowcolor{lightgray} &

        \multicolumn{2}{c|}{\rotatebox{0}{\(\mathcal{A}\)}} &
        \multicolumn{2}{c|}{\rotatebox{0}{\textsc{RTS}}} &
        \multicolumn{2}{c|}{\rotatebox{0}{\(k\)-search}} &

        \multicolumn{2}{c|}{\rotatebox{0}{\(\mathcal{A}\)}} &
        \multicolumn{2}{c|}{\rotatebox{0}{\textsc{RTS}}} &
        \multicolumn{2}{c|}{\rotatebox{0}{\(k\)-search}}
    \\
    \hline
    \rowcolor{lightgray} Name &

        \rotatebox{90}{\systemname runtime\phantom{x}} &
        \rotatebox{90}{SyGuS calls (sat)\phantom{x}} &

        \rotatebox{90}{\systemname runtime} &
        \rotatebox{90}{SyGuS calls (sat)} &

        \rotatebox{90}{\systemname runtime} &
        \rotatebox{90}{SyGuS calls (sat)} &

        \rotatebox{90}{\systemname runtime} &
        \rotatebox{90}{SyGuS calls (sat)} &

        \rotatebox{90}{\systemname runtime} &
        \rotatebox{90}{SyGuS calls (sat)} &

        \rotatebox{90}{\systemname runtime} &
        \rotatebox{90}{SyGuS calls (sat)}
    \\
    \hline

         Stop EC2Instance &\scriptsize 4.51 & \scriptsize  0 (1) & \scriptsize 9.03 & \scriptsize  0 (1) & \scriptsize 0.24 & \scriptsize  0 (1) & \scriptsize 4.55 & \scriptsize  0 (1) & \scriptsize 9.05 & \scriptsize  0 (1) & \scriptsize 0.18 & \scriptsize  0 (1)  \\
        \rowcolor{lblue} Start EC2Instance &\scriptsize 19.26 & \scriptsize  4 (3) & \scriptsize 13.83 & \scriptsize  5 (2) & \scriptsize \color{darkgray} 0.30 & \scriptsize \textcolor{darkgray}{ 1} (\textcolor{darkgray}{ 1}) & \scriptsize \color{darkgray} 11.50 & \scriptsize \textcolor{darkgray}{ 2} (\textcolor{darkgray}{ 1}) & \scriptsize \color{darkgray} 15.33 & \scriptsize \textcolor{darkgray}{ 4} (\textcolor{darkgray}{ 2}) & \scriptsize \color{darkgray} 351.02 & \scriptsize \textcolor{darkgray}{ 14} (\textcolor{darkgray}{ 58})  \\
         Configure S3Bucket Versioning &\scriptsize \color{darkgray} 71.48 & \scriptsize \textcolor{darkgray}{ 3} (\textcolor{darkgray}{ 0}) & \scriptsize \color{darkgray} 80.84 & \scriptsize \textcolor{darkgray}{3} (\textcolor{darkgray}{0}) & \scriptsize \color{darkgray} 77.89 & \scriptsize \textcolor{darkgray}{3} (\textcolor{darkgray}{0}) & \scriptsize \color{darkgray} 50.56 & \scriptsize \textcolor{darkgray}{ 1} (\textcolor{darkgray}{ 0}) & \scriptsize \color{darkgray} 77.11 & \scriptsize \textcolor{darkgray}{1} (\textcolor{darkgray}{0}) & \scriptsize \color{darkgray} 56.80 & \scriptsize \textcolor{darkgray}{1} (\textcolor{darkgray}{0})  \\
        \rowcolor{lblue} Configure Cloud Watch On EC2 &\scriptsize \color{darkgray} 13.72 & \scriptsize \textcolor{darkgray}{ 2} (\textcolor{darkgray}{ 0}) & \scriptsize \color{darkgray} 13.72 & \scriptsize \textcolor{darkgray}{ 2} (\textcolor{darkgray}{ 0}) & \scriptsize \textcolor{darkgray}{TO} & \scriptsize \text{\textcolor{darkgray}{N/A}} & \scriptsize \color{darkgray} 8.77 & \scriptsize \textcolor{darkgray}{ 1} (\textcolor{darkgray}{ 0}) & \scriptsize \color{darkgray} 13.73 & \scriptsize \textcolor{darkgray}{ 2} (\textcolor{darkgray}{ 0}) & \scriptsize \textcolor{darkgray}{TO} & \scriptsize \text{\textcolor{darkgray}{N/A}}  \\
         Copy EC2Instance &\scriptsize \textcolor{darkgray}{TO} & \scriptsize \text{\textcolor{darkgray}{N/A}} & \scriptsize \textcolor{darkgray}{TO} & \scriptsize \text{\textcolor{darkgray}{N/A}} & \scriptsize \textcolor{darkgray}{TO} & \scriptsize \text{\textcolor{darkgray}{N/A}} & \scriptsize \textcolor{darkgray}{TO} & \scriptsize \text{\textcolor{darkgray}{N/A}} & \scriptsize \textcolor{darkgray}{TO} & \scriptsize \text{\textcolor{darkgray}{N/A}} & \scriptsize \textcolor{darkgray}{TO} & \scriptsize \text{\textcolor{darkgray}{N/A}}  \\
        \rowcolor{lblue} Resize Instance &\scriptsize \textcolor{darkgray}{TO} & \scriptsize \text{\textcolor{darkgray}{N/A}} & \scriptsize \textcolor{darkgray}{TO} & \scriptsize \text{\textcolor{darkgray}{N/A}} & \scriptsize \textcolor{darkgray}{TO} & \scriptsize \text{\textcolor{darkgray}{N/A}} & \scriptsize \textcolor{darkgray}{TO} & \scriptsize \text{\textcolor{darkgray}{N/A}} & \scriptsize \textcolor{darkgray}{TO} & \scriptsize \text{\textcolor{darkgray}{N/A}} & \scriptsize \textcolor{darkgray}{TO} & \scriptsize \text{\textcolor{darkgray}{N/A}}  \\
         Set Required Tags &\scriptsize 16.93 & \scriptsize  2 (0) & \scriptsize 16.97 & \scriptsize  2 (0) & \scriptsize 4.91 & \scriptsize  13 (43) & \scriptsize 16.97 & \scriptsize  2 (0) & \scriptsize 16.97 & \scriptsize  2 (0) & \scriptsize 243.93 & \scriptsize  3 (0)  \\
    \hline
    \end{tabularx}
\end{table}

\begin{table}[H]
    \centering
    \footnotesize
    \renewcommand{\arraystretch}{1.2}
    \setlength{\tabcolsep}{3pt}
    \caption{List of Literature benchmarks, with synthesis time and number of SyGuS solver calls.
    For each algorithm (\(\mathcal{A}\), \textsc{RTS} and \(k\)-search)
    we report the synthesis time in seconds, the number of total SyGuS solver calls,
    and the subset of those calls that returned SAT, in parenthesis.
    When \systemname times out (TO), it does not report the number of solver calls (N/A).}
    \label{tab:synthesis-times-literature}
    \begin{tabularx}{\linewidth}{|X| |c|c||c|c||c|c| |c|c||c|c||c|c|}
    \hline
        \rowcolor{lightestgray} \multicolumn{13}{| c|}{\small Literature benchmarks} \\
    \hline
        \rowcolor{lightgray}
        \ &
        \multicolumn{6}{c|}{ \(\scoreFuncOne\) } &
        \multicolumn{6}{c|}{ \(\scoreFuncTwo\) }
    \\
    \hline
        \rowcolor{lightgray} &

        \multicolumn{2}{c|}{\rotatebox{0}{\(\mathcal{A}\)}} &
        \multicolumn{2}{c|}{\rotatebox{0}{\textsc{RTS}}} &
        \multicolumn{2}{c|}{\rotatebox{0}{\(k\)-search}} &

        \multicolumn{2}{c|}{\rotatebox{0}{\(\mathcal{A}\)}} &
        \multicolumn{2}{c|}{\rotatebox{0}{\textsc{RTS}}} &
        \multicolumn{2}{c|}{\rotatebox{0}{\(k\)-search}}
    \\
    \hline
    \rowcolor{lightgray} Name &

        \rotatebox{90}{\systemname runtime\phantom{x}} &
        \rotatebox{90}{SyGuS calls (sat)\phantom{x}} &

        \rotatebox{90}{\systemname runtime} &
        \rotatebox{90}{SyGuS calls (sat)} &

        \rotatebox{90}{\systemname runtime} &
        \rotatebox{90}{SyGuS calls (sat)} &

        \rotatebox{90}{\systemname runtime} &
        \rotatebox{90}{SyGuS calls (sat)} &

        \rotatebox{90}{\systemname runtime} &
        \rotatebox{90}{SyGuS calls (sat)} &

        \rotatebox{90}{\systemname runtime} &
        \rotatebox{90}{SyGuS calls (sat)}
    \\
    \hline

         Api Phany Ex. 1.1 &\scriptsize \textcolor{darkgray}{TO} & \scriptsize \text{\textcolor{darkgray}{N/A}} & \scriptsize \textcolor{darkgray}{TO} & \scriptsize \text{\textcolor{darkgray}{N/A}} & \scriptsize \textcolor{darkgray}{TO} & \scriptsize \text{\textcolor{darkgray}{N/A}} & \scriptsize \textcolor{darkgray}{TO} & \scriptsize \text{\textcolor{darkgray}{N/A}} & \scriptsize \textcolor{darkgray}{TO} & \scriptsize \text{\textcolor{darkgray}{N/A}} & \scriptsize \textcolor{darkgray}{TO} & \scriptsize \text{\textcolor{darkgray}{N/A}}  \\
        \rowcolor{lblue} Api Phany Ex. 1.2 &\scriptsize 18.64 & \scriptsize  4 (2) & \scriptsize 19.64 & \scriptsize  4 (2) & \scriptsize \color{darkgray} 0.07 & \scriptsize \textcolor{darkgray}{ 0} (\textcolor{darkgray}{ 0}) & \scriptsize 17.72 & \scriptsize  4 (2) & \scriptsize 19.43 & \scriptsize  4 (2) & \scriptsize \color{darkgray} 84.93 & \scriptsize \textcolor{darkgray}{ 1} (\textcolor{darkgray}{ 1})  \\
         Api Phany Ex. 1.6 &\scriptsize 28.84 & \scriptsize  6 (2) & \scriptsize 24.64 & \scriptsize  6 (3) & \scriptsize \color{darkgray} 0.09 & \scriptsize \textcolor{darkgray}{ 1} (\textcolor{darkgray}{ 1}) & \scriptsize 27.49 & \scriptsize  6 (2) & \scriptsize 24.45 & \scriptsize  6 (3) & \scriptsize \color{darkgray} 174.70 & \scriptsize \textcolor{darkgray}{ 1} (\textcolor{darkgray}{ 1})  \\
        \rowcolor{lblue} Api Phany Ex. 1.8 &\scriptsize 21.74 & \scriptsize  4 (3) & \scriptsize 9.18 & \scriptsize  3 (3) & \scriptsize \color{darkgray} 0.16 & \scriptsize \textcolor{darkgray}{ 1} (\textcolor{darkgray}{ 1}) & \scriptsize 23.76 & \scriptsize  4 (3) & \scriptsize 11.60 & \scriptsize  3 (3) & \scriptsize 0.18 & \scriptsize  1 (1)  \\
         Api Phany Ex. 2.1 &\scriptsize 1.41 & \scriptsize  1 (1) & \scriptsize \textcolor{darkgray}{TO} & \scriptsize \text{\textcolor{darkgray}{N/A}} & \scriptsize 2.04 & \scriptsize  2 (2) & \scriptsize 1.42 & \scriptsize  1 (1) & \scriptsize \textcolor{darkgray}{TO} & \scriptsize \text{\textcolor{darkgray}{N/A}} & \scriptsize 1.68 & \scriptsize  2 (2)  \\
        \rowcolor{lblue} Api Phany Ex. 2.3 &\scriptsize 168.98 & \scriptsize  12 (2) & \scriptsize 184.82 & \scriptsize  15 (2) & \scriptsize \color{darkgray} 0.04 & \scriptsize \textcolor{darkgray}{ 0} (\textcolor{darkgray}{ 0}) & \scriptsize 161.13 & \scriptsize  12 (2) & \scriptsize 184.17 & \scriptsize  15 (2) & \scriptsize \color{darkgray} 161.41 & \scriptsize \textcolor{darkgray}{ 10} (\textcolor{darkgray}{1})  \\
         Api Phany Ex. 2.5 &\scriptsize 1.40 & \scriptsize  1 (1) & \scriptsize \textcolor{darkgray}{TO} & \scriptsize \text{\textcolor{darkgray}{N/A}} & \scriptsize 6.23 & \scriptsize  2 (2) & \scriptsize 1.33 & \scriptsize  1 (1) & \scriptsize \textcolor{darkgray}{TO} & \scriptsize \text{\textcolor{darkgray}{N/A}} & \scriptsize 6.07 & \scriptsize  2 (2)  \\
        \rowcolor{lblue} Api Phany Ex. 2.6 &\scriptsize 20.50 & \scriptsize  4 (2) & \scriptsize 23.62 & \scriptsize  4 (2) & \scriptsize \color{darkgray} 0.14 & \scriptsize \textcolor{darkgray}{ 0} (\textcolor{darkgray}{ 0}) & \scriptsize 30.28 & \scriptsize  4 (2) & \scriptsize 31.88 & \scriptsize  4 (2) & \scriptsize \textcolor{darkgray}{TO} & \scriptsize \text{\textcolor{darkgray}{N/A}}  \\
         Api Phany Ex. 2.7 &\scriptsize 4.13 & \scriptsize  1 (1) & \scriptsize \textcolor{darkgray}{TO} & \scriptsize \text{\textcolor{darkgray}{N/A}} & \scriptsize \textcolor{darkgray}{TO} & \scriptsize \text{\textcolor{darkgray}{N/A}} & \scriptsize 4.38 & \scriptsize  1 (1) & \scriptsize \textcolor{darkgray}{TO} & \scriptsize \text{\textcolor{darkgray}{N/A}} & \scriptsize \textcolor{darkgray}{TO} & \scriptsize \text{\textcolor{darkgray}{N/A}}  \\
        \rowcolor{lblue} Api Phany Ex. 2.10 &\scriptsize \textcolor{darkgray}{TO} & \scriptsize \text{\textcolor{darkgray}{N/A}} & \scriptsize \textcolor{darkgray}{TO} & \scriptsize \text{\textcolor{darkgray}{N/A}} & \scriptsize \textcolor{darkgray}{TO} & \scriptsize \text{\textcolor{darkgray}{N/A}} & \scriptsize \textcolor{darkgray}{TO} & \scriptsize \text{\textcolor{darkgray}{N/A}} & \scriptsize \textcolor{darkgray}{TO} & \scriptsize \text{\textcolor{darkgray}{N/A}} & \scriptsize \textcolor{darkgray}{TO} & \scriptsize \text{\textcolor{darkgray}{N/A}}  \\
         Api Phany Ex. 2.11 &\scriptsize 13.72 & \scriptsize  4 (3) & \scriptsize 8.79 & \scriptsize  3 (3) & \scriptsize \color{darkgray} 0.11 & \scriptsize \textcolor{darkgray}{ 1} (\textcolor{darkgray}{ 1}) & \scriptsize 13.90 & \scriptsize  4 (3) & \scriptsize 8.21 & \scriptsize  3 (3) & \scriptsize 0.12 & \scriptsize  1 (1)  \\
        \rowcolor{lblue} Api Phany Ex. 2.13 &\scriptsize 40.39 & \scriptsize  9 (3) & \scriptsize \color{darkgray} 31.81 & \scriptsize \textcolor{darkgray}{ 7} (\textcolor{darkgray}{ 2}) & \scriptsize \color{darkgray} 0.08 & \scriptsize \textcolor{darkgray}{ 0} (\textcolor{darkgray}{ 0}) & \scriptsize 40.53 & \scriptsize  9 (3) & \scriptsize 31.15 & \scriptsize  7 (2) & \scriptsize \color{darkgray} 28.94 & \scriptsize \textcolor{darkgray}{ 0} (\textcolor{darkgray}{ 0})  \\
         Api Phany Ex. 3.1 &\scriptsize \text{<0.01} & \scriptsize  0 (0) & \scriptsize \textcolor{darkgray}{TO} & \scriptsize \text{\textcolor{darkgray}{N/A}} & \scriptsize 1.97 & \scriptsize  1 (1) & \scriptsize \text{<0.01} & \scriptsize  0 (0) & \scriptsize \textcolor{darkgray}{TO} & \scriptsize \text{\textcolor{darkgray}{N/A}} & \scriptsize 1.61 & \scriptsize  1 (1)  \\
        \rowcolor{lblue} Api Phany Ex. 3.3 &\scriptsize \color{darkgray} 162.15 & \scriptsize \textcolor{darkgray}{ 3} (\textcolor{darkgray}{ 1}) & \scriptsize \color{darkgray} 165.80 & \scriptsize \textcolor{darkgray}{ 2} (\textcolor{darkgray}{ 1}) & \scriptsize \color{darkgray} 181.45 & \scriptsize \textcolor{darkgray}{2} (\textcolor{darkgray}{0}) & \scriptsize \color{darkgray} 162.83 & \scriptsize \textcolor{darkgray}{ 3} (\textcolor{darkgray}{ 1}) & \scriptsize \color{darkgray} 164.49 & \scriptsize \textcolor{darkgray}{ 2} (\textcolor{darkgray}{ 1}) & \scriptsize \color{darkgray} 183.63 & \scriptsize \textcolor{darkgray}{3} (\textcolor{darkgray}{0})  \\
         Api Phany Ex. 3.4* &\scriptsize 4.08 & \scriptsize  1 (1) & \scriptsize \textcolor{darkgray}{TO} & \scriptsize \text{\textcolor{darkgray}{N/A}} & \scriptsize 3.92 & \scriptsize  1 (1) & \scriptsize 3.86 & \scriptsize  1 (1) & \scriptsize \textcolor{darkgray}{TO} & \scriptsize \text{\textcolor{darkgray}{N/A}} & \scriptsize 4.03 & \scriptsize  5 (6)  \\
        \rowcolor{lblue} Api Phany Ex. 3.6 &\scriptsize 3.39 & \scriptsize  1 (1) & \scriptsize \textcolor{darkgray}{TO} & \scriptsize \text{\textcolor{darkgray}{N/A}} & \scriptsize \color{darkgray} 2.85 & \scriptsize \textcolor{darkgray}{ 3} (\textcolor{darkgray}{ 3}) & \scriptsize 3.57 & \scriptsize  1 (1) & \scriptsize \textcolor{darkgray}{TO} & \scriptsize \text{\textcolor{darkgray}{N/A}} & \scriptsize \color{darkgray} 3.02 & \scriptsize \textcolor{darkgray}{ 3} (\textcolor{darkgray}{ 3})  \\
         Api Phany Ex. 3.7 &\scriptsize 3.06 & \scriptsize  1 (1) & \scriptsize \color{darkgray} 106.23 & \scriptsize \textcolor{darkgray}{ 2} (\textcolor{darkgray}{ 1}) & \scriptsize 3.57 & \scriptsize  1 (1) & \scriptsize 3.07 & \scriptsize  1 (1) & \scriptsize \color{darkgray} 105.35 & \scriptsize \textcolor{darkgray}{ 2} (\textcolor{darkgray}{ 1}) & \scriptsize \textcolor{darkgray}{TO} & \scriptsize \text{\textcolor{darkgray}{N/A}}  \\
        \rowcolor{lblue} Api Phany Ex. 3.8 &\scriptsize \textcolor{darkgray}{TO} & \scriptsize \text{\textcolor{darkgray}{N/A}} & \scriptsize \textcolor{darkgray}{TO} & \scriptsize \text{\textcolor{darkgray}{N/A}} & \scriptsize \textcolor{darkgray}{TO} & \scriptsize \text{\textcolor{darkgray}{N/A}} & \scriptsize \textcolor{darkgray}{TO} & \scriptsize \text{\textcolor{darkgray}{N/A}} & \scriptsize \textcolor{darkgray}{TO} & \scriptsize \text{\textcolor{darkgray}{N/A}} & \scriptsize \textcolor{darkgray}{TO} & \scriptsize \text{\textcolor{darkgray}{N/A}}  \\
         Api Phany Ex. 3.9 &\scriptsize \textcolor{darkgray}{TO} & \scriptsize \text{\textcolor{darkgray}{N/A}} & \scriptsize \textcolor{darkgray}{TO} & \scriptsize \text{\textcolor{darkgray}{N/A}} & \scriptsize \textcolor{darkgray}{TO} & \scriptsize \text{\textcolor{darkgray}{N/A}} & \scriptsize \textcolor{darkgray}{TO} & \scriptsize \text{\textcolor{darkgray}{N/A}} & \scriptsize \textcolor{darkgray}{TO} & \scriptsize \text{\textcolor{darkgray}{N/A}} & \scriptsize \textcolor{darkgray}{TO} & \scriptsize \text{\textcolor{darkgray}{N/A}}  \\
    \hline
    \end{tabularx}
\end{table}

\section{Rewrite Rules}\label{sec:appendix-rewrite-rules}
\subsection{Refinement Rewrite Rules}\label{sec:refinement-rewrite-rules}
Recall our two sets of rewrite rules: refinement rewrite rules 
and synthesis rewrite rules 
. 
We ignore the data transformation definitions associated with the program, since they do not intervene in the rewrite.

\paragraph{Control Flow Manipulation}
The following rewrite rule pulls calls to the same API out of a conditional:
\begin{multline}\label{rule:pull-same-calls-of-ite}
\dsllambda \overline{z} \ \mathcal{S}\ 
\dslite{c}{
\dsllet\ x = \mathbb{A}(\overline{y})\ \mathcal{S}'}{\dsllet\ x' = \mathbb{A}(\overline{y}')\ \mathcal{S}''}\ \mathcal{R} \\
\rightsquigarrow_t 
\dsllambda \overline{z}\ \mathcal{S}\ \dsllet\ x = \mathbb{A}(\dsltern{c}{\overline{y}}{\overline{y}'})
\; \dslite{c}{\mathcal{S}'}{\left.\mathcal{S}''  \left[x' \to x\right] \right.}\ 
\left. \mathcal{R}\left[x' \to x\right] \right.\\
\text{ with } t(\sigma) = \sigma[x \to \dsltern{c}{\sigma(x)}{\sigma(x')}] 
\end{multline}
Note that in this transformation, the condition is used both in transforming the arguments of the API call moved out of the branches (in \( \dsltern{c}{\overline{y}}{\overline{y}'}\)) and in the update to the state (i.e. the mapping to \(x\) becomes the evaluated expression \(\dsltern{c}{\sigma(x)}{\sigma(x')}\).

A symmetric rewrite rule pushes calls out of branches:
\begin{multline}\label{rule:push-same-calls-of-ite}
    \dsllambda \overline{z}\ \mathcal{S}\ 
\dslite{c}{
\mathcal{S}' \
\dsllet\ x = \mathbb{A}(\overline{y})
}{
\mathcal{S}''\
\dsllet\ x' = \mathbb{A}(\overline{y}')
}\ \mathcal{R} 
\\
\rightsquigarrow_t 
\dsllambda \overline{z}\ \mathcal{S}\ 
\dslite{c}{\mathcal{S}'}{\mathcal{S}'' }\ 
\dsllet\ x = \mathbb{A}(\dsltern{c}{\overline{y}}{\overline{y}'})
\left. \mathcal{R}\left[x' \to x\right] \right.\\
\text{ with } t(\sigma) = \sigma[x \to \sigma(x) \cup \sigma(x')] 
\end{multline}

The following rewrite rule eliminates empty conditional statements:
\begin{multline}\label{rule:eliminate-empty-ifs}
\dsllambda \overline{z}\
\mathcal{S}\ \dslite{c}{}{}\ \mathcal{R}
\rightsquigarrow_t
\dsllambda \overline{z}\
\mathcal{S}\ \mathcal{R} \\
\text{ with  } t(\sigma) = \sigma
\end{multline}
And this rewrite rule inverts the branches of a conditional statement when the then branch is~empty:
\begin{multline}
\dsllambda \overline{z}\ 
\mathcal{S}\ \dslite{c}{}{\mathcal{S}'}\ \mathcal{R}
\rightsquigarrow_t
\mathcal{S}\ \dslite{\neg c}{\mathcal{S}'}{} \mathcal{R}
\\ \text{ with } t(\sigma) = \sigma
\end{multline}

More complex rewrite rules manipulate the control flow of the program to combine conditional statements in order to group API calls:
\begin{multline}
\dsllambda \overline{z}\ 
\mathcal{S}\
\dslite{c}{
    \dsllet\ x = \mathbb{A}(\overline{y})
}{
    \dslite{c'}{
        \dsllet\ x' = \mathbb{A}(\overline{y}')
    }{
    \mathcal{S'}
    }
}\
\mathcal{R}
\\ 
\rightsquigarrow_t
\dsllambda \overline{z}\ \mathcal{S}\
\dslite{c \vee c'}{
\dsllet\ x = \mathbb{A}\left(\dsltern{c}{\overline{y}}{\overline{y}'}\right)
}{
\mathcal{S'}
}\
\left. \mathcal{R}\left[x' \to x\right] \right.
\\ 
\text{ with } t(\sigma) = \sigma[x \to \dsltern{c}{\sigma(x)}{(\dsltern{c'}{\sigma(x')}{\emptyset})}]
\end{multline}

A symmetric rule handles statements where the branches in the nested conditionals are reversed:
\begin{multline}
\dsllambda \overline{z}\ 
\mathcal{S}\
\dslite{c}{
    \dsllet\ x = \mathbb{A}(\overline{y})
}{
    \dslite{c'}{ \mathcal{S'} }{ \dsllet\ x' = \mathbb{A}(\overline{y}') }
}\
\mathcal{R}
\\ 
\rightsquigarrow
\dsllambda \overline{z}\ \mathcal{S}\
\dslite{c \vee \neg c'}{
\dsllet\ x = \mathbb{A}\left(\dsltern{c}{\overline{y}}{\overline{y}'}\right)
}{
\mathcal{S'}
}\
\left. \mathcal{R}\left[x' \to x\right] \right.
\\ 
\text{ with } t(\sigma) = \sigma[x \to \dsltern{c}{\sigma(x)}{(\dsltern{c'}{\emptyset}{\sigma(x')})}]
\end{multline}
In general, all rules have similar variation that consider symmetries in the program. We do not list all of those.

Nested conditionals can be sequenced if the appropriate branches are empty:
\begin{multline}
\dsllambda \overline{z}\ \mathcal{S}\
\dslite{c}{
    \dsllet\ x = \mathbb{A}(\overline{y})\
    \dslite{c'}{
        \dsllet\ x' = \mathbb{A}(\overline{y}')
    }{}
}{}
\\
\rightsquigarrow_t
\dsllambda \overline{z}\ \mathcal{S}\
\dslite{c}{
    \dsllet\ x = \mathbb{A}(\overline{y})
}{}\ 
\dslite{c' \wedge c}{
    \dsllet\ x' = \mathbb{A}(\overline{y}')
}{}\ 
\mathcal{R}
\\ 
\text{ with } t(\sigma) = \sigma
\end{multline}
And similarly nested conditionals can be simplified:
\begin{multline}
\dsllambda \overline{z}\ \mathcal{S}\
\dslite{c}{
    \dslite{c'}{
        \mathcal{S'}
    }{}
}{}
\\
\rightsquigarrow_t
\dsllambda \overline{z}\ \mathcal{S}\
\dslite{c' \wedge c}{
    \mathcal{S'}
}{}\ 
\mathcal{R}
\\ 
\text{ with } t(\sigma) = \sigma
\end{multline}

\paragraph{Parameter Elimination}
A rewrite rule eliminates a parameter \(x\) if it is not used in the body of the program:
\[ 
    x \not\in FV(\mathcal{S}) \implies \dsllambda x,\overline{y}\ \mathcal{S} \\
    \rightsquigarrow_t \dsllambda \overline{y}\ \mathcal{S} 
    \\
    \text{ where } t(\sigma) = \sigma
\]
where \(x \not\in FV(\mathcal{S})\) means that \(x\) is not in the free variables of \(\mathcal{S}\), i.e. not used in \(\mathcal{S}\)

\paragraph{Data Transform Elimination} Simple data transformations can be eliminated by inlining then in the program. In general, we inline data transformations that return constants or return one of their arguments.
For example, a program \(\dsllambda \overline{x}\ \mathcal{S}\ \dslwhere\ f := (x,y) \to y\) is transformed into a program \( \dsllambda \overline{x}\ \mathcal{S}'\) , in which \(\mathcal{S}'\) is the same as \(\mathcal{S}\) modulo the inlining of \(f\). Inlining consists in substituting all variables bound to the result of \(f\) by the second argument of the binding.

\paragraph{Expression Simplification}
There are additional refinement rules that operate only on the expressions that appear in the program and have no impact on the control flow (and their state transformation functions are identity).

\subsection{Synthesis Rewrite rules}

The following synthesis rewrite rule replaces the branching condition of an if-then-else instruction whose condition \(\mathcal{C}(\dummybranch)\) depends on \dummybranch with a hidden function \(\phi\): in place of \(\mathcal{C}(\dummybranch)\), the branching condition becomes the output of a function \(\phi\). This function may take as input all parameters besides \dummybranch, or any variable bound in the previous instructions of the program.

\begin{center}
\begin{tabular}{m{0.14\textwidth} m{.05\textwidth} m{0.2\textwidth}}
\begin{lstlisting}[language=Syren,numbers=none]
|\textcolor{purple}{$\dsllambda$} $\dummybranch, \overline{y}$|.
|$\mathcal{S}$|
if |\(\mathcal{C}(\dummybranch)\)| {|\(\mathcal{T}\)|}
else {|\(\mathcal{R}\)|}
\end{lstlisting} 
& \(\rightsquigarrow_t\)  & \begin{lstlisting}[language=Syren,numbers=none]
|\textcolor{purple}{$\dsllambda$} $\phi$. \textcolor{purple}{$\dsllambda$} $\dummybranch, \overline{y}$|.
|$\mathcal{S}$|
let b = |$\phi$|(|$\overline{y}$|, |BV|(|$\mathcal{S}$|))
if b {|\(\mathcal{T}\)|}
else {|\(\mathcal{R}\)|}
\end{lstlisting}
\label{example-of-synthesis-rule-remove-br}
\end{tabular}

\vspace{-2em}
{\small \raggedleft where \(t(\sigma) = \sigma \cup [b \to C(\dummybranch{})]\).

}
\end{center}

\noindent
In the rewrite above, \(BV(\mathcal{S})\) denotes the set of bound variables in the statements of \(\mathcal{S}\). The rewrite introduces a fresh variable \dsl{b}, which is assigned the value of the function \(\phi\) and is used as the if-condition. 
To maintain correctness, \(t\) updates the augmented program state to ensure the value of \dsl{b} is the same as \(\mathcal{C}(\dummybranch)\) for every input trace.

Retry loops are introduced by synthesis rules because the condition of the loop needs to be synthesized in order for the loop to be valid.
Retry loop introduction rewrites have the following form:

\begin{center}
\begin{tabular}{m{0.12\textwidth} m{.05\textwidth} m{0.3\textwidth}}
\begin{lstlisting}[language=Syren,numbers=none]
|\textcolor{purple}{$\dsllambda$} $\overline{y}$|.
|$\mathcal{R}$|
|$\mathcal{S}$|
|$\mathcal{S'}$|
if |$\mathcal{C}$| {|$\mathcal{S''}$|}
|$\mathcal{T}$|
\end{lstlisting} 
& \(\rightsquigarrow_t\)  & \begin{lstlisting}[language=Syren,numbers=none] 
|\textcolor{purple}{$\dsllambda$} $\phi$ \textcolor{purple}{$\dsllambda$} $\overline{y}$|.
|$\mathcal{R}$|
retry {
  |$\mathcal{S}$|
  let b = |$\phi$|(|BV($\mathcal{S}$),BV($\mathcal{R}$),$\overline{y}$|)
} until b
|$\mathcal{T}$|
\end{lstlisting}
\end{tabular}
\vspace{-.7em}

{\small \raggedleft where $t(\sigma) = \sigma[\overline{BV(\mathcal{S}),BV(\mathcal{S}),BV(\mathcal{S})} \to \overline{BV(\mathcal{S}),BV(\mathcal{S'}),BV(\mathcal{S''})}] \,\cup\, [\overline{b} \to \overline{?\mathcal{S}, ?\mathcal{S'}, ?\mathcal{S''}}]$
}

\end{center}

\noindent
A new variable \dsl{b} is introduced, bound to the result of the hidden function \(\phi\), and then used as a stopping condition for the retry loop. Statements \(\mathcal{S}\), \(\mathcal{S'}\), and \(\mathcal{S''}\) must all be calls to the same API method. The state transformation now maps \emph{iterations} of the bound variables of \(\mathcal{S}\) to the values of each statement that has been captured in the loop, represented by the vector \(\overline{BV(\mathcal{S}),BV(\mathcal{S'}),BV(\mathcal{S''})}\). When evaluating the program for a trace, the variables \( BV(\mathcal{S''})\) will not be defined for the traces where \(\mathcal{C}\) is false, in which case the value is null. The valuation of condition \dsl{b} is also a vector \(\overline{?\mathcal{S}, ?\mathcal{S'}, ?\mathcal{S''}}\) that is computed by assigning the truth value of whether the statement replaced by the iteration is the last statement in the trace (denoted by \(\overline{?S}\)).
In practice, we generalize this rule by considering patterns where matching statements are in sequence, and the matching can happen within conditionals, as is the case with the last statement in this example~rule.

\section{LLM Experiment Prompt}\label{app:prompt}

\colorlet{saved}{.}
\color{cameraready}

The following is the prompt used in every call to Claude 3.5 Sonnet in our comparison against LLMs experiment described in \S\ref{sec:evaluation}. We describe the task that \systemname solves, including the goals of the synthesized program and lightweight restrictions on the output language. We choose to let the LLM express the script in any language to improve its chances of success.
We also emphasize that the desired program must be able to replicate every provided trace, as that is our formal definition of correctness. We also added text that attempts to counteract a tendency to include unnecessary loops in the program.

\begin{lstlisting}[breaklines, numbers=none, breakautoindent=false,breakindent=0ex, basicstyle=\linespread{.9}\small\ttfamily]
You are a helpful assistant. You are given a set of traces from a program execution, and you need to generate a script that the user can then use in some automation.
The traces are in the format of a list of json events, where each event is of the form:
```
{
    "api": <THE API NAME>,
    "request": {
        <ARGUMENT 1>: <VALUE 1>,
        <ARGUMENT 2>: <VALUE 2>,
        ...
    },
    "response": <RESPONSE VALUE>
}
```
That is, the event contains an API name in the "api" field, the request parameters in the "request" field and the response value in the "response" field. The "request" is a JSON object, where each member is one of the arguments of the API call.

The script you generate can use retry loops, list maps and conditionals. Each statement in the script should either be a control flow statement (conditional or loop), an API call (with the same name as in the traces and same arguments) or a data transformation function. Data transformation functions can transform the output of an API call into another value that can be used in a conditional or in another API call's arguments. You need to write the script as part of a function with clearly identified parameters.

The script you generate MUST reproduce at least the same traces that it was given, for some value of the parameters of the script.
The script SHOULD be also more general, meaning that different parameter values can be used, and they would generate similar traces as the original ones.
You SHOULD NOT add loops when the traces do not show the need for loops, for example, if each trace calls an API once, then the script should only call once.
\end{lstlisting}

\color{saved}
\section{Detailed Benchmarks}\label{sec:detailed-benchmarks}

This sections shows our detailed benchmarks of Cloud automation scripts. These scripts were manually written to reflect Cloud automation tasks, and were used to evaluate \systemname.

\subsection{Custom Benchmarks}
Hand-written benchmarks that perform common cloud automation tasks.

\subsubsection{BackupThenDeleteTable} \phantom{break}

\begin{syrenColorListing}
lambda tableName, backupName
let backupCreation = dynamodb.CreateBackup(TableName=tableName, BackupName=backupName)
let backupArns = getBackupArn(backupCreation)
let backupArn = first(backupArns)
retry {
   let backupDesc = dynamodb.DescribeBackup(BackupArn=backupArn)
   let backupStatus = getBackupStatus(backupDesc)
} until (backupStatus[0] == "AVAILABLE")

let tableDeletion = dynamodb.DeleteTable(TableName=tableName)

where
getBackupArn := "$.BackupDetails.BackupArn"
getBackupStatus := "$.BackupDescription.BackupDetails.BackupStatus"
first := (x) -> x[0]
\end{syrenColorListing}

\pagebreak[1] 

\subsubsection{CreateBucketThenFolder}\phantom{break}

\begin{syrenColorListing}
lambda bucketName
let bucketCreation = s3.CreateBucket(bucket=bucketName)
let folderCreation = s3.PutObject(bucket=bucketName, key="custom-logs/")
\end{syrenColorListing}

\pagebreak[1] 

\subsubsection{CreateImage}\phantom{break}

\begin{syrenColorListing}
lambda instanceId, noReboot
let instanceIds = list(instanceId)
let description = ec2.DescribeInstances(instanceIds=instanceIds)
let requestId = getRequestId(description)
let imageName = buildName(instanceId, requestId)
let result = ec2.CreateImage(instanceId=instanceId, noReboot=noReboot, name=imageName)
where
buildName := (X, Y) -> (X + "_auto_") + Y[0]
getRequestId := "$..x-amzn-requestid"
\end{syrenColorListing}

\pagebreak[1] 

\subsubsection{CreateImageFromSsmParameterAndLog}\phantom{break}

\begin{syrenBreakableColorListing}
lambda instances, x_3
let instanceDesc = ec2.DescribeInstances(InstanceIds=instances)
let parameter = ssm.GetParameter(Name="image-builder")
let instanceId = getExistingInstance(instances)
let imageName = getParameterValue(parameter)
let mustReboot = checkReboot(x_3)
let imageCreationRes = ec2.CreateImage(
  InstanceId=instanceId, Name=imageName, NoReboot=mustReboot)
let imageId = getImageId(imageCreationRes)
let api_res_2_3 = console.Log(message=imageId)
where
getExistingInstance := (x) -> x[0]
getParameterValue := (x) -> x..Value[0]
checkReboot := (x) -> x == "ami-e27e9b0896ea30tb2"
getImageId := (x) -> x.ImageId
    
\end{syrenBreakableColorListing}

\pagebreak[1]

\subsubsection{CreateTable}\phantom{break}
\begin{syrenColorListing}
lambda tableName, definitions, keys, provisioning
let tableCreating = dynamodb.CreateTable(tableName=tableName,
    attributeDefinitions=definitions,
    keySchema=keys,
    provisionedThroughput=provisioning)
\end{syrenColorListing}

\pagebreak[1] 

\subsubsection{CreateTableThenInsertItem}\phantom{break}

\begin{syrenColorListing}
lambda tableName, definitions, keys, item
let provisionedThroughput = throughputSettings()
let tableCreating = dynamodb.CreateTable(TableName=tableName,
    AttributeDefinitions=definitions,
    KeySchema=keys,
    ProvisionedThroughput=provisionedThroughput)

retry {
   let tableDescription = dynamodb.DescribeTable(TableName=tableName)
   let tableStatus = extractStatus(tableDescription)
} until (tableStatus[0] == "ACTIVE")
let tableInserting = dynamodb.PutItem(TableName=tableName,Item=item)

where
extractStatus := "$.Table.TableStatus"
throughputSettings := () -> { "ReadCapacityUnits" : 5, "WriteCapacityUnits" : 5 }
\end{syrenColorListing}

\pagebreak[1] 

\subsubsection{DeleteTable}\phantom{break}

\begin{syrenColorListing}
lambda tableName
let tableDeletion = dynamodb.DeleteTable(TableName=tableName)
\end{syrenColorListing}

\pagebreak[1] 

\subsubsection{MoveDdbItem}\phantom{break}

\begin{syrenColorListing}
lambda originTableName, destinationTableName, keyValue, keyId
let itemKey = composeKey(keyValue, keyId)
let itemResponse = dynamodb.GetItem(TableName=originTableName,Key=itemKey)
let items = selectItem(itemResponse)
let item = first(items)
let response = dynamodb.PutItem(TableName=destinationTableName,Item=item)
where
composeKey := (x, y) -> {y : {"S": x}}
first := (x) -> x[0]
selectItem := "$.Item"
\end{syrenColorListing}

\pagebreak[1] 

\subsubsection{MoveDdbItemIfPresent}\phantom{break}

\begin{syrenColorListing}
lambda originTableName, destinationTableName, keyValue, keyId
let itemKey = composeKey(keyValue, keyId)
let itemResponse = dynamodb.GetItem(TableName=originTableName,Key=itemKey)
let items = selectItem(itemResponse)
if len(items) > 0 {
  let item = first(items)
  let response = dynamodb.PutItem(TableName=destinationTableName,Item=item)
}
where
composeKey := (x, y) -> {y : {"S": x}}
first := (x) -> x[0]
selectItem := "$.Item"
\end{syrenColorListing}

\pagebreak[1] 

\subsubsection{PutObjectIfNotPresent}\phantom{break}

\begin{syrenColorListing}
lambda bucketName, objectName
let listObjectsResult = s3.ListObjectsV2(Bucket=bucketName,Prefix=objectName)
let counts = getCount(listObjectsResult)
let count = fst(counts)
if count == 1 {
  let contentKeys = getContentKey(listObjectsResult)
  let contentKey = fst(contentKeys)
  if contentKey == objectName {
    return
  }
}
let createObject = s3.PutObject(Bucket=bucketName,Key=objectName)
where
getCount := "$.KeyCount"
getContentKey := "$.Contents[0].Key"
fst := (X) -> X[0]
\end{syrenColorListing}

\pagebreak[1] 

\subsubsection{CopyS3Objects}\phantom{break}

\begin{syrenColorListing}
lambda originBucket, dest
let list_objects_result = s3.ListObjects(Bucket=originBucket, Prefix="*")
let keys = getObjectsKeys(list_objects_result)
for (key) in keys {
  let objectInfo = s3.HeadObject(Bucket=dest, Key=key)
  let c = getObjectVersion(objectInfo)
  if isFile {
    let moved = s3.Copy(
      SourceBucket=originBucket, 
      SourceKey=key, 
      DestinationBucket=dest, 
      DestinationKey=key
    )
  }
}
where
 getObjectsKeys := (x) -> x..key
 getObjectVersion := (x) -> x..version == "v1"    
\end{syrenColorListing}

\pagebreak[1] 

\subsubsection{StartInstances}\phantom{break}

\begin{syrenColorListing}
lambda instanceIds
let x = ec2.StartInstances(InstanceIds=instanceIds)

\end{syrenColorListing}

\pagebreak[1] 

\subsubsection{StartInstancesWithTags}\phantom{break}

\begin{syrenColorListing}
lambda key, value
let filters = makeFilterTag(key,value)
let x = ec2.DescribeInstances(Filters=filters)
let iids = extractIids(x)
let resp = ec2.StartInstances(InstanceIds=iids)
where
extractIids := "$..InstanceId"
makeFilterTag := (k,v) -> [{"Name": "tag:" + k, "Values" : [v]}]

\end{syrenColorListing}

\pagebreak[1] 

\subsubsection{StopAllRunningInstances}\phantom{break}

\begin{syrenColorListing}
lambda 
let filters = runningInstancesFilter()
let runningInstances = ec2.DescribeInstances(Filters=filters)
let instanceIds = extractInstanceIds(runningInstances)
let res = ec2.StopInstances(InstanceIds=instanceIds)
where
runningInstancesFilter := () -> [dict(Name="instance-state-name",Values=["running"])]
extractInstanceIds := "$..InstanceId"
\end{syrenColorListing}

\pagebreak[1] 

\subsubsection{StopInstances}\phantom{break}

\begin{syrenColorListing}
lambda instanceIds
let x = ec2.StopInstances(InstanceIds=instanceIds)
\end{syrenColorListing}

\pagebreak[1] 

\subsubsection{StopInstancesCond}\phantom{break}

\begin{syrenColorListing}
lambda instanceId
let ids = list(instanceId)
let _ = ec2.StopInstances(InstanceIds=ids, Force=False)
let s = ec2.DescribeInstanceStatus(InstanceIds=ids, IncludeAllInstances=True)
let statuses = extractInstanceStatus(s)
let status = first(statuses)
if status != "stopped" {
    let _ = ec2.StopInstances(InstanceIds=ids, Force=True)
}
where
first := (X) -> X[0]
extractInstanceStatus := "$.InstanceStatuses[0].InstanceState.Name"
\end{syrenColorListing}

\pagebreak[1] 

\subsubsection{TagInstancesWithDryRun}\phantom{break}

\begin{syrenColorListing}
lambda instanceIds, key, value
let tags = makeTags(key,value)
let x = ec2.CreateTags(Resources=instanceIds, Tags=tags, DryRun=True)
if x == "ok" {
  let _ = ec2.CreateTags(Resources=instanceIds, Tags=tags)
}
where
makeTags := (k,v) -> [dict(Key=k, Value=v)]
\end{syrenColorListing}

\pagebreak[1] 

\subsubsection{WaitForInputThenSendEmail}\phantom{break}
\begin{syrenColorListing}
lambda  user_email
retry {
  let payload = os.QueryInput()
  let b = isNonEmpty(_r)
} until ( b )
let _ = messages.SendEmail(payload=payload, userId=isNonEmpty)
where
isNonEmpty := (x) -> ! (x == "")
\end{syrenColorListing}

\pagebreak[1] 

\subsubsection{RetrieveChannelMembers}\phantom{break}

\begin{syrenColorListing}
lambda
let slackResponse = slack.GetConversationsList()
let channelId = getFirstChannelId(slackResponse)
let conversationMembersResponse = slack.GetConversationMembers(channel=channelId)
let responseUserList = getMembersFunction(conversationMembersResponse)
for (memberUser) in responseUserList {
  let user_info = slack.GetUserInfo(user=memberUser)
}
 where
  getFirstChannelId := (x) -> x..id[0]
  getMembersFunction := (x) -> x.members
\end{syrenColorListing}

\pagebreak[1] 

\subsubsection{SVG-IncreaseCircleRadius}\phantom{break}

\begin{syrenColorListing}
lambda elementId
let elt = js.getElementAttributeById(Attribute=0, Id=elementId)
let incr = f(api_res_2_0)
let api_res_2_1 = js.setElementAttribute(Attribute=0, Id=elementId, Value=incr)
 where
  f := (x) -> 1 + x
\end{syrenColorListing}

\pagebreak[1] 

\subsubsection{SVG-IncreaseCircleRadius-json}\phantom{break}

\begin{syrenColorListing}
lambda document
let circle = js.getElementById(Document=document,Id="circle1")
let radius = js.getElementAttribute(Element=circle, Name="r")
let result = js.setElementAttribute(Element="circle1", Name="r", Value=radius)
where
add10 := (circle) -> circle + 10
\end{syrenColorListing}

\pagebreak[1] 

\subsubsection{SVG-ScaleAndMove}\phantom{break}

\begin{syrenColorListing}
lambda  x_1
let api_res_2_0 = js.getElementAttributeById(Attribute=0, Id=x_1)
let param_elim_12 = f?_22(api_res_2_0)
let api_res_2_1 = js.setElementAttribute(Attribute=0, Id=x_1, Value=param_elim_12)
 where
  f?_22 := (arg0_2) -> 1 + arg0_2
\end{syrenColorListing}

\pagebreak[1] 

\subsubsection{List and Move Files}\phantom{break}

\begin{syrenColorListing}
lambda source, dest
let paths = os.Ls(Dir=source, Pattern="/*")
for (path) in paths {
  let isFile = os.IsFile(Path=path)
  if isFile {
    let moved = os.Mv(Path=path, Dest=dest)
  }
}
\end{syrenColorListing}

\pagebreak[1] 

\subsubsection{Clean Current Dir}\phantom{break}
\begin{syrenColorListing}
lambda dirName
let files = os.OpenAndGetdents(Dir=dirName)
let names = f1(files)
for (path) in (names) {
  let isfile_res = os.IsFile(Path=path) 
  if isfile_res {
    let rm_res = os.Rm(Path=path)
  }
}
where:
f1 : x -> $..d_name
\end{syrenColorListing}

\pagebreak[1] 

\subsubsection{Clean Regular Files in Dir}\phantom{break}
\begin{syrenColorListing}
lambda dirName
let getdents_res = os.OpenAndGetdents(Dir=dirName)
let f1_res = f1(getdents_res)
for (path) in (f1_res) {
  let fstat_res = os.OpenAndGetdentsFstat(Path=path) 
  let f2_res = f2(fstat_res)
  if f2_res {
    let mv_res = os.Rm(Path=path)
  }
}
where:
f1 : x -> $..d_name
f2 : x -> $.st_type == "S_IFREG"
\end{syrenColorListing}

\pagebreak[1] 

\subsection{Blink automation benchmarks}
Tasks from the Blink automation library \cite{blink-automation}, where various APIs are interfaced.

\subsubsection{CopyEC2InstanceToNewRegion}\phantom{break}

\begin{syrenColorListing}
lambda  instanceId, imageName, sourceRegion, destRegion, destName
let createImageResponse = ec2.CreateImage(InstanceId=instanceId, Name=imageName)
let imageId = extractImageId(createImageResponse)
let imageIds = list(imageId)
retry {
  let describeImagesResponse = ec2.DescribeImages(ImageIds=imageIds)
  let state = extractImageStatus(describeImagesResponse)
} until (state == "available")
let copyResponse = ec2.CopyImage(
    Name=destName, SourceImageId=imageId, SourceRegion=sourceRegion, Region=destRegion)

where
extractImageId := "$.ImageId"
extractImageStatus := "$..State"
\end{syrenColorListing}

\pagebreak[1] 

\subsubsection{CreateIamUserAndNotify}\phantom{break}

\begin{syrenColorListing}
lambda userName, notificationEmail
let createUserResponse = iam.CreateUser(UserName=userName)
let user = extractUser(createUserResponse)
let sendEmail = blink.SendEmail(Content=)
where
extractUser := "$.User"
\end{syrenColorListing}

\pagebreak[1] 

\subsection{AWS Automation Runbooks Benchmarks}
Scripts collected from AWS Systems Manager Automation Runbooks.

\subsubsection{AWS-ConfigureCloudWatchOnEC2Instance}\phantom{break}

\begin{syrenColorListing}
lambda instanceId, status, properties
let instanceIds = list(instanceId)
if status == "Enabled" {
  let _ = ec2.MonitorInstances(InstanceIds=instanceIds)
} else {
  let _ = ec2.UnmonitorInstances(InstanceIds=instanceIds)
}
\end{syrenColorListing}

\pagebreak[1] 

\subsubsection{AWS-ConfigureS3BucketVersioning}\phantom{break}

\begin{syrenColorListing}
lambda bucketName, versioningState
if ((versioningState != "Enabled") && (versioningState != "Suspended")) {
  return
}
let config = makeConfig(versioningState)
let result = s3.PutBucketVersioning(Bucket=bucketName, VersioningConfiguration=config)
where
makeConfig := (x) -> dict(MFADelete="Disabled",Status=x)
\end{syrenColorListing}

\pagebreak[1] 

\subsubsection{AWS-ResizeInstance}\phantom{break}

\begin{syrenBreakableColorListing}
lambda instanceId, instanceType
let instanceIds = list(instanceId)
let description = ec2.DescribeInstances(InstanceIds=instanceIds)
let actualType = getInstanceType(description)
if (actualType != instanceType) {
  let _ = ec2.StopInstances(InstanceIds=instanceIds)
  retry {
    let x = ec2.DescribeInstanceStatus()
    let status = extractInstanceStatus(x)
  } until (status == "stopped")
  let v = makeValue(instanceType)
  let _ = ec2.ModifyInstanceAttribute(InstanceId=instanceId, InstanceType=v)
  let _ = ec2.StartInstances(InstanceIds=instanceIds)
  retry {
    let x = ec2.DescribeInstanceStatus()
    let status = extractInstanceStatus(x)
  } until (status == "running")
}
where
extractInstanceStatus := "$.InstanceStatuses[0].InstanceState.Name"
getInstanceType := "$.Reservations[0].Instances[0].InstanceType"
makeValue := (x) -> {'Value' : x}
\end{syrenBreakableColorListing}

\pagebreak[1] 

\subsubsection{AWS-SetRequiredTags}\phantom{break}

\begin{syrenColorListing}
lambda tags, resources
let taggedResourcesResponse = resourcegroupstaggingapi.TagResources(
  resourceARNList=resources,
   tags=tags)
\end{syrenColorListing}

\pagebreak[1] 

\subsubsection{AWS-StartEC2Instance}\phantom{break}

\begin{syrenColorListing}
lambda instanceId
let ids = list(instanceId)
let s = ec2.DescribeInstances(InstanceIds=ids)
let status = extractInstanceStatus(s)
if status != "running" {
  let _ = ec2.StartInstances(InstanceIds=ids)
}
where
first := (X) -> X[0]
extractInstanceStatus := "$.Reservations[0].Instances[0].State.Name"
\end{syrenColorListing}

\pagebreak[1] 

\subsubsection{AWS-StopEC2Instance}\phantom{break}

\begin{syrenColorListing}
lambda instanceId
let ids = list(instanceId)
let _ = ec2.StopInstances(InstanceIds=ids, Force=False)
let _ = ec2.StopInstances(InstanceIds=ids, Force=True)
\end{syrenColorListing}

\pagebreak[1] 

\subsubsection{AWSSupport-CopyEC2Instance}\phantom{break}

\begin{syrenBreakableColorListing}
lambda instanceId,
    keyPair,
    region,
    subnetId,
    instanceType,
    securityGroupIds,
    keepImageSourceRegion,
    keepImage,
    keepImageDestinationRegion,
    noRebootInstanceBeforeTakingImage

(** Extract information from instance *)
let instanceIds = list(instanceId)
let instanceInfo = ec2.DescribeInstances(InstanceIds=instanceIds)
if instanceInfo == null {
   return
}
let sourceInstanceTypes = extractSourceInstanceType(instanceInfo)
let sourceAvailabilityZones = extractSourceAvailabilityZone(instanceInfo)
let sourceSubnetIds = extractSourceSubnetId(instanceInfo)
let sourceKeyPairs = extractSourceKeyPair(instanceInfo)
let sourceSecurityGroupIds = extractSourceSecurityGroupIds(instanceInfo)
let sourceRootDeviceNames = extractSourceRootDeviceName(instanceInfo)

(** Select *)
let sourceInstanceType = first(sourceInstanceTypes)
let sourceAZ = first(sourceAvailabilityZones)
let sourceSubnetId = first(sourceSubnetIds)
let sourceKeyPair = first(sourceKeyPairs)

(** Prepare the parameters for the new instance *)
let regionToUse = selectFirstNonEmpty(region, sourceAZ)
let isSameRegion = isSame(regionToUse, sourceAZ)
let instanceTypeToUse = selectFirstNonEmpty(instanceType, sourceInstanceType)
let subnetIdToUse = select(isSameRegion, subnetId, sourceSubnetId)
let securityGroupIdsToUse = 
  select(isSameRegion, securityGroupIds, sourceSecurityGroupIds)
// Create new image
let imageName = makeName(instanceId)
let newImage = ec2.CreateImage(
  InstanceId=instanceId, NoReboot=noRebootInstanceBeforeTakingImage, Name=imageName)
let newImageIds = extractImageId(newImage)
retry {
   let imageInfo = ec2.DescribeImages(ImageIds=newImageIds)
   let imageState = extractImageState(imageInfo)
} until (imageState[0] == "available")
if (imageState[0] != "available") {
   return
}
(** Tag image *)
let tags = makeTags(instanceId)
let _ = ec2.CreateTags(Resources=newImageIds, Tags=tags)
(** Launch instance *)
let tagSpecs = makeTagSpecs(instanceId)
let newImageId = first(newImageIds)
if isSameRegion {
  if keyPair == "" {
      let r1 = ec2.RunInstances(
          ImageId=newImageId,
          SubnetId=subnetIdToUse,
          InstanceType=instanceTypeToUse,
          SecurityGroupIds=securityGroupIdsToUse,
          TagSpecifications=tagSpecs,
          MinCount=1,
          MaxCount=1)
  } else {
      let r1 = ec2.RunInstances(
          ImageId=newImageId,
          SubnetId=subnetIdToUse,
          InstanceType=instanceTypeToUse,
          SecurityGroupIds=securityGroupIdsToUse,
          KeyName=keyPair,
          TagSpecifications=tagSpecs,
          MinCount=1,
          MaxCount=1)
  }
  let newInstanceIds = getInstanceId(r1)

  // Wait for the new instance to be running
  let newInstanceId = first(newInstanceIds)
  retry {
     let newInstances = ec2.DescribeInstanceStatus(InstanceIds=newInstanceIds)
     let newInstanceStatus = extractInstanceStatus(newInstances)
  } until ((len(newInstanceStatus) > 0) && (newInstanceStatus[0] == "running"))

  if !keepImageSourceRegion {
     let deregisterResult = ec2.DeregisterImage(ImageId=newImageId)
     let snapshots = getSnapshots(imageInfo)
     let snapshot = first(snapshots)
     let deleteSnapshotResult = ec2.DeleteSnapshot(SnapshotId=snapshot)
  }
}

(** Transformation Definitions *)
where
extractSourceInstanceType := "$.Reservations[0].Instances[0].InstanceType"
extractSourceAvailabilityZone := 
  "$.Reservations[0].Instances[0].Placement.AvailabilityZone"
extractSourceSubnetId := "$.Reservations[0].Instances[0].SubnetId"
extractSourceKeyPair := "$.Reservations[0].Instances[0].KeyName"
extractSourceSecurityGroupIds := 
  "$.Reservations[0].Instances[0].SecurityGroups..GroupId"
extractSourceRootDeviceName := "$.Reservations[0].Instances[0].RootDeviceName"
selectFirstNonEmpty := (X, Y) -> X == "" ? Y : X
select := (B, X, Y) -> B ? (X == "" ? Y : X) : X
isSame := (X,Y) -> X == Y
first := (X) -> X[0]
makeName := (X) -> "Api Doc AWSSupport-CopyEC2Instance LocalAmi for " + X
extractImageId := "$.ImageId"
extractImageState := "$.Images[0].State"
getSnapshots := "$.Images[0]..SnapshotId"
getInstanceId := "$.Instances[0].InstanceId"
extractInstanceStatus := "$.InstanceStatuses[0].InstanceState.Name"
makeTags := (X) -> [
   {"Key" : "Name", "Value" : "AWSSupport-CopyEC2Instance LocalAmi for" + X },
   {"Key" : "AWSSupport-CopyEC2Instance", "Value" : "some-unique-id" },
   {"Key" : "CreatedBy", "Value": "AWSSupport-CopyEC2Instance" }
]
makeTagSpecs := (X) -> [{
    "ResourceType": "instance",
    "Tags" : [
        { "Key" : "Name", "Value" : "AWSSupport-CopyEC2Instance Source: " + X },
        { "Key": "CreatedBy", "Value": "AWSSupport-CopyEC2Instance" }
    ]
}]
\end{syrenBreakableColorListing}

\pagebreak[1] 

\subsection{\textsc{ApiPhany} Benchmarks}

These are tasks adapted from previous literature on API-composing programs Synthesis, \textsc{ApiPhany}  \cite{apiphany}, that use Stripe and Slack APIs.

\subsubsection{ApiPhanyExample.1.1}\phantom{break}

\begin{syrenColorListing}
lambda location_id
let res = stripe.GetTransactions(Location=location_id)
let transactions = extractTransactionsOrderIds(res)
for transaction in transactions {
  let _ = console.Log(Message=transaction)
}
where
extractTransactionsOrderIds := (x) -> x..OrderId
\end{syrenColorListing}

\pagebreak[1] 

\subsubsection{ApiPhanyExample.1.2}\phantom{break}

\begin{syrenBreakableColorListing}
lambda user_email
let userLookupResult = slack.LookupByEmail(name=user_email)
let user_id = getUserIdFromLookupResult(userLookupResult)
let user_conversations = slack.ConversationsOpen(users=user_id)
let channelId = getChannelIdFromConversations(user_conversations)
let slack_response = slack.PostMessage(channel=channelId)
 where
  getChannelIdFromConversations := (x) -> ((x[0])[2]).channelId
  getUserIdFromLookupResult := (x) -> x.user..id
\end{syrenBreakableColorListing}

\pagebreak[1] 

\subsubsection{ApiPhanyExample.1.6}\phantom{break}

\begin{syrenColorListing}
lambda channelId, messageTimestamp
let slackResponse = slack.ChatPostMessage(channel=channelId, ts=messageTimestamp)
let messageThreadId = getThreadFromResponse(slackResponse)
let response = slack.ChatUpdate(channel=channelId, ts=messageThreadId)
 where
  getThreadFromResponse := (x) -> x..thread[1]
\end{syrenColorListing}

\pagebreak[1] 

\subsubsection{ApiPhanyExample.1.8}\phantom{break}

\begin{syrenColorListing}
lambda channelId
let conversationInfoResponse = slack.GetConversationInfo(channel=channelId)
let latestMessageTimestamp = getLatestMessageTimestamp(conversationInfoResponse)
let channelId = getChannelId(conversationInfoResponse)
let conversationHistoryResponse = slack.GetConversationsHistory(
    channel=channelId, oldest=latestMessageTimestamp)
 where
  getChannelId := (x) -> x.id
  getLatestMessageTimestamp := (y) -> y.latest
\end{syrenColorListing}

\pagebreak[1] 

\subsubsection{ApiPhanyExample.2.1}\phantom{break}

\begin{syrenColorListing}
lambda customer, productId
let api_res_2_0 = v1.GetPrices(currency="USD", productId=productId)
let typ = listTypes(api_res_2_0)
for (typ) in (types) { 
 et _r = v1.PostSubcriptions(customer=customer, productId=productId, type=typ) 
}
 where
  listTypes := (x) -> x..type
\end{syrenColorListing}

\pagebreak[1] 

\subsubsection{ApiPhanyExample.2.3}\phantom{break}

\begin{syrenColorListing}
lambda customerId, productCurrency, productName, unitAmount
let productResponse = v1.PostProducts(name=productName)
let productId = getProductIdFromProductResponse(productResponse)
let productPriceResponse = v1.PostPrices(
  currency=productCurrency, product=productId, unit_amount=unitAmount)
let productId = getPriceId(productPriceResponse)
let productInvoiceItemsResponse = v1.PostInvoiceItems(customer=customerId, price=productId)
 where
  getPriceId := (x) -> x..id[0]
  getProductIdFromProductResponse := (y) -> y..id[0]
\end{syrenColorListing}

\pagebreak[1] 

\subsubsection{ApiPhanyExample.2.5}\phantom{break}

\begin{syrenColorListing}
lambda customer
let invoiceApiResponse = v1.GetInvoices(customer=customer)
let invoiceIds = getInvoiceIdFromResponse(invoiceApiResponse)
for (invoiceId) in (invoiceIds) {
  let chargesResponse = v1.GetCharges(invoiceId=invoiceId)
}
where
  getInvoiceIdFromResponse := (x) -> x.data
\end{syrenColorListing}

\pagebreak[1] 

\subsubsection{ApiPhanyExample.2.6}\phantom{break}

\begin{syrenColorListing}
lambda subscriptionId
let subscriptionResponse = stripe.GetSubscription(SubscriptionExposedId=subscriptionId)
let subscription_id = getSubscriptionId(subscriptionResponse)
let subscriptionInvoice = stripe.GetInvoice(invoice=subscription_id)
let invoiceId = getInvoiceAmountPaid(subscriptionInvoice)
let refundedInvoice = stripe.Refund(charge=invoiceId)
 where
  getInvoiceAmountPaid := (x) -> x.amount_paid
  getSubscriptionId := (x) -> x..id[1]
\end{syrenColorListing}

\pagebreak[1] 

\subsubsection{ApiPhanyExample.2.7}\phantom{break}

\begin{syrenColorListing}
let customerApiResponse = v1.GetCustomers()
let filteredCustomers = getCustomerEmails(customerApiResponse)
for (customerEmail) in (filteredCustomers) {
  let customerEmail = mail.Show(email=customerEmail)
}
where
  getCustomerEmails := (x) -> x..email
\end{syrenColorListing}

\pagebreak[1] 

\subsubsection{ApiPhanyExample.2.10}\phantom{break}

\begin{syrenColorListing}
lambda customerId, defaultPaymentId
let api_res_1_0 = v1.GetSubscriptions(customer=customerId)
let subscriptionIds = getSubscriptionIds(api_res_1_0)
for (subcriptionId) in (subscriptionIds) {
  let _r =
    v1.PostSubscriptionPayment(
      default_payment=defaultPaymentId,
        subscriptionExposedId=subscriptionId)
}
where
 getSubscriptionIds := (x) -> x.data..id
\end{syrenColorListing}

\pagebreak[1] 

\subsubsection{ApiPhanyExample.2.11}\phantom{break}

\begin{syrenColorListing}
lambda customer_id
let getCustomerResult = stripe.GetCustomer(customer=customer_id)
let default_source = get_default_source(getCustomerResult)
let _ = stripe.DeletePaymentSource(customer=customer_id, id=default_source)
 where
  get_default_source := (x) -> x.default_source
\end{syrenColorListing}

\pagebreak[1] 

\subsubsection{ApiPhanyExample.2.13}\phantom{break}

\begin{syrenColorListing}
lambda customerId, price
let invoiceItemsResponse = stripe.InvoiceItems(customer=customerId, price=price)
let invoiceId = getInvoiceId(invoiceItemsResponse)
let invoice = stripe.Invoice(customer=customerId, invoice=invoiceId)
let invoice = stripe.SendInvoice(invoice=invoiceId)

where
getInvoiceId := (x) -> x..invoiceId[0]
\end{syrenColorListing}

\pagebreak[1] 

\subsubsection{ApiPhanyExample.3.1}\phantom{break}

\begin{syrenColorListing}
lambda locationId
let invoice_response = square.GetInvoices(locationId=locationId)
\end{syrenColorListing}

\pagebreak[1] 

\subsubsection{ApiPhanyExample.3.3}\phantom{break}

\begin{syrenColorListing}
lambda tax_id
let list_res = slack.GetCatalog()
let objects = extractObjects(list_res)
for (object) in objects {
  let tax_ids = taxIds(object)
  for (id) in tax_ids {
    let p = tax_predicate(tax_id, id)
    if (p)) {
      let x = console.Log(object)
    }
  }
}

where
extractObjects := (x) -> x.objects
taxIds := (x) -> x.item_data.tax_ids
tax_predicate := (x,y) -> x == y
\end{syrenColorListing}

\pagebreak[1] 

\subsubsection{ApiPhanyExample.3.4}\phantom{break}

\begin{syrenColorListing}
let catalogList = square.ListCatalog()
let catalogMessages = getDiscountTypeFromResponse(catalogList)
for (message) in (catalogMessages) {let discountDetails = console.Log(message=message) }
where
  getDiscountTypeFromResponse := (x) -> x..discountType
\end{syrenColorListing}

\pagebreak[1] 

\subsubsection{ApiPhanyExample.3.6}\phantom{break}

\begin{syrenColorListing}
let paymentsResponse = square.GetPayments()
let paymentMessages = getPaymentNotes(paymentsResponse)
for (paymentMessage) in (paymentMessages) {
  let paymentNotes = console.Log(message=paymentMessage)
}
where
  getPaymentNotes := (x) -> x..payment_note
\end{syrenColorListing}

\pagebreak[1] 

\subsubsection{ApiPhanyExample.3.7}\phantom{break}

\begin{syrenColorListing}
lambda location_id
let res = stripe.GetTransactions(Location=location_id)
let transactions = extractTransactionsOrderIds(res)
for transaction in transactions {
  let _ = console.Log(Message=transaction)
}
where
extractTransactionsOrderIds := (x) -> x..OrderId
\end{syrenColorListing}

\pagebreak[1] 

\subsubsection{ApiPhanyExample.3.8}\phantom{break}

\begin{syrenColorListing}
lambda locationId, transactionId
let res = square.BatchRetrieveOrders(LocationId=locationId, TransactionId=transactionId);
let orders = getOrders(orders)
for order in orders {
  let line_items = getLineItems(order)
  for line_item in line_items {
    let _ = console.Log(Message=line_item)
  }
}

where
getOrders := (x) -> x.orders
getLineItems := (x) -> x.line_items
\end{syrenColorListing}

\fi

\end{document}